\documentclass[a4paper,reqno]{amsart}
\usepackage[T1]{fontenc}
\usepackage[english]{babel}
\usepackage{amssymb,bbm,enumerate}
\usepackage{bm}
\usepackage{xcolor}
\usepackage[linktocpage=true,colorlinks=true, linkcolor=blue, citecolor=red, urlcolor=green]{hyperref}
\usepackage{multirow} 
\usepackage{graphicx}

\usepackage{float}
\restylefloat{table}
\usepackage{tabu}
\usepackage[makeroom]{cancel}

\renewcommand{\phi}{\varphi}
\renewcommand{\ker}{\Ker}
\renewcommand{\Re}{\operatorname{Re}}

\newcommand{\bb}[1]{\mathbb{#1}}
\newcommand{\mc}[1]{\mathcal{#1}}
\newcommand{\mf}[1]{\mathfrak{#1}}

\newcommand{\mbbm}[1]{\mathbbm{#1}}

\newcommand{\beq}{\begin{equation}}
\newcommand{\eeq}{\end{equation}}
\newcommand{\e}{\varepsilon}

\newcommand{\de}{\partial}
\newcommand{\la}{\lambda}
\newcommand{\hv}{h^{\vee}}

\DeclareMathOperator{\tdeg}{deg_{x}}
\newcommand{\QQ}{Q\widetilde{Q}}

\DeclareMathOperator{\End}{End}

\DeclareMathOperator{\ad}{ad}

\DeclareMathOperator{\Ker}{Ker}

\DeclareMathOperator{\rank}{rank}
\DeclareMathOperator{\Ht}{ht}
\DeclareMathOperator{\op}{op}
\DeclareMathOperator{\Op}{Op}
\newcommand{\ope}{\op_{\mathfrak{g}}}
\newcommand{\tope}{\widetilde{\op}_{\mathfrak{g}}}
\newcommand{\Ope}{\Op_{\mathfrak{g}}}

\theoremstyle{plain}
\newtheorem{theorem}{Theorem}[section]
\newtheorem{lemma}[theorem]{Lemma}
\newtheorem{conjecture}[theorem]{Conjecture}
\newtheorem{proposition}[theorem]{Proposition}
\newtheorem{corollary}[theorem]{Corollary}

\newtheorem{asu}{Assumption}

\theoremstyle{definition}
\newtheorem{definition}[theorem]{Definition}
\newtheorem{example}[theorem]{Example}

\theoremstyle{remark}
\newtheorem{remark}[theorem]{Remark}

\numberwithin{equation}{section}

\definecolor{light}{gray}{.9}

\usepackage{float}
\restylefloat{table}
\usepackage{mathtools}
\usepackage{tikz}
\usetikzlibrary{chains}

\tikzset{node distance=2em, ch/.style={circle,draw,on chain,inner sep=2pt},chj/.style={ch,join},every path/.style={shorten >=4pt,shorten <=4pt},line width=1pt,baseline=-1ex}

\let\dlabel=\alabel

\newcommand{\dnode}[2][chj]{%
\node[#1,label={below:\dlabel{#2}}] {};
}

\newcommand{\blackdnode}[2][chj]{%
\node[#1,fill={black},label={below:\dlabel{#2}}] {};
}

\newcommand{\dnodebr}[1]{%
\node[chj,label={below right:\dlabel{#1}}] {};
}

\newcommand{\blackdnodebr}[1]{%
\node[chj,fill={black},label={below right:\dlabel{#1}}] {};
}

\newcommand{\dydots}{%
\node[chj,draw=none,inner sep=1pt] {\dots};
}

\usepackage{datenumber}

\newcounter{dateone}
\newcounter{datetwo}

\newcommand{\difftoday}[3]{%
\setmydatenumber{dateone}{\the\year}{\the\month}{\the\day}%
\setmydatenumber{datetwo}{#1}{#2}{#3}%
\addtocounter{datetwo}{-\thedateone}%
\textcolor{red}{\the\numexpr-\thedatetwo day(s) late}
}

\title{Opers for higher states of quantum KdV models}

\author{D.Masoero, A.Raimondo}

\address{Grupo de F\'isica Matem\'atica da Universidade de Lisboa,
Av. Prof. Gama Pinto 2, 1649-003 Lisboa, Portugal.}
\email{dmasoero@gmail.com}

\address{Dipartimento di Ingegneria Gestionale, dell'Informazionae e della Produzione, Universit\'a di Bergamo,
Viale Marconi 5, I-24044 Dalmine (BG), Italy}
\email{andrea.raimondo@gmail.com}

\begin{document}

\pagestyle{plain}

\begin{abstract}
We study the ODE/IM correspondence for all states of the quantum $\widehat{\mathfrak{g}}$-KdV model,
where $\widehat{\mathfrak{g}}$ is the affinization of a simply-laced simple Lie algebra $\mf{g}$. 
We construct quantum $\widehat{\mathfrak{g}}$-KdV opers as an explicit realization of the class of opers introduced by Feigin and Frenkel \cite{FF11}, which are defined
by fixing the singularity structure at $0$ and $\infty$, and by allowing a finite number of additional singular terms
with trivial monodromy. We prove that the generalized monodromy data of the quantum $\widehat{\mathfrak{g}}$-KdV opers satisfy
the Bethe Ansatz equations of the quantum $\widehat{\mathfrak{g}}$-KdV model.
The trivial monodromy conditions are equivalent to a complete system of
algebraic equations for the additional singularities.
\end{abstract}

\maketitle

\tableofcontents

\section*{Introduction}

The purpose of the present paper is to explicitly construct and then study a class of opers, introduced 
by Feigin and Frenkel \cite{FF11}, corresponding to higher states of the quantum
$\widehat{\mf{g}}$-KdV model, where $\widehat{\mf{g}}$ is the untwisted affinization of a simply-laced simple
Lie algebra $\mf{g}$. The work is
the natural continuation of our previous papers in collaboration with Daniele Valeri on the
ODE/IM correspondence for the ground state of the quantum $\widehat{\mf{g}}$-KdV model, where we developed an effective method
to construct solutions of the Bethe Ansatz equations as generalised monodromy data of affine opers \cite{marava15,marava17}. 

The quantum $\widehat{\mf{g}}$-KdV model arises as the quantisation of the second Hamiltonian structure of the Drinfeld-Sokolov hierarchy
\cite{baluzaI,bazhanov02integrable,ff90}
-- equivalently Toda field theory \cite{feigin96} -- as well as the
continuous (conformal) limit of XXZ-like lattice models whose underlying symmetry is $U_q(\widehat{\mf{g}})$ \cite{doreyreview}.
Both the lattice models and the quantum field theories carry the structure of quantum integrability
(the quantum inverse scattering), so that each state is characterised by a solution of the (nested) Bethe Ansatz equations,
which in turn furnishes all the physical observables of the theory.

In a ground-breaking series of papers Dorey and Tateo \cite{doreytateo98}, followed by Bazhanov, Lukyanov, and
Zamolodchikov \cite{bazhanov01},
discovered that
the solution of the Bethe Ansatz equations of the ground state of  quantum $\widehat{\mf{sl}_2}-$KdV
(i.e. quantum KdV) admits a very simple and neat representation.
Let indeed $\Psi(x,E)$ the unique subdominant solution as  $x \to +\infty$ of
the Schr\"odinger equation
\begin{equation}\label{eq:schrodinger}
 -\psi''(x)+(x^{2\alpha}+\frac{l(l+1)}{x^2}-E)\psi=0,
\end{equation}
with $ \alpha>0$, and $\Re l > - 1/2$. Then, $Q(E)=\lim_{x\to0}x^{-l-1}\Psi(x,E)$ is
the required solution of the Bethe Ansatz equations, with the parameters $\alpha,l,E$ of the Schroedinger
equation corresponding to the the central charge $c$, the
vacuum parameter, and the spectral parameter of the quantum model, see \cite{BLZ04} for the precise identification.

Such a discovery, which was thereafter known as the  ODE/IM correspondence, has been generalised to many more pairs
of a quantum integrable model (solvable by the Bethe Ansatz) and a linear differential operator. Examples of
these generalisations, which are conjectural but supported by strong numerical evidence and deep mathematical structures,
include the correspondence between all higher states of the quantum $\widehat{\mf{sl}}_2$-KdV model and Schroedinger equation with
'monster potentials' \cite{BLZ04}, the correspondence between the ground state of massive
deformations of the quantum KdV model, such as quantum Sine Gordon
and quantum affine Toda theories, and the Lax operator of a dual classical theory \cite{lukyanov10,bazhanov14,dorey13}, and
the very recent discovery of the correspondence between an $O(3)$ non-linear Sigma model and a Schroedinger operator
\cite{bazhanov18} .
The appearance of the Thermodynamic Bethe Ansatz
in relation with BPS spectra in $\mc{N}=2$ Gauge theories \cite{gaiotto09}, Donaldson-Thomas invariants
\cite{bridgeland18}, and in more general quantum mechanics equations \cite{dtba,gaiotto14,ito18}, is also expected
to be manifestations of the same phenomenon.

All these particular ODE/IM correspondences are strong evidences of the existence of an overarching
ODE/IM correspondence, which can be informally stated as follows:

\textit{Given an integrable quantum field theory, and one state of that theory, there exists
a differential operator whose generalised monodromy data provide the solution of Bethe Ansatz
equations of the given state.}

One can make the the above conjecture much more precise for the case of the quantum $ \widehat{\mf{g}}-$KdV model. First of all,
as discovered  by Feigin and Frenkel \cite{FF11}, the differential operators on the ODE side of the correspondence are certain
${}^L \widehat{\mf{g}}$ opers, where ${}^L \widehat{\mf{g}}$ is the Langlands dual algebra of $\widehat{\mf{g}}$. This implies that in
the particular case under our analisys, namely when
 $\widehat{\mf{g}}$ is the untwisted affinization of a simply laced simple Lie algebra $\mf{g}$,
 we should consider operators with values in  ${}^L \widehat{\mf{g}} \cong \widehat{\mf{g}}$. 
 
The complete ODE/IM correspondence for $\widehat{\mf{g}}$, with $\mf{g}$ simply laced, can be then described as follows.
The quantum model is defined by the choice of
the central charge $c$ and the vacuum parameter $p \in \mf{h}$ of the free field representation \cite{baluzaI,bazhanov02integrable,kojima08,hj12}.
Every state of the Fock space is associated to a set of $\rank \mf{g}$ entire functions
$Q^{(l)}(\la),l=1 \dots \rank \mf{g}$, of the spectral parameter $\la$ -- first introduced in \cite{baluzaI},
later generalized in \cite{bazhanov02integrable,kojima08,hj12}, and finally settled in \cite{frenkel15,fh16} in the most general case --
which solve the following Bethe Ansatz equations:
\begin{equation}\label{eq:TBAintro}
\prod_{j = 1}^{\rank \mf g} e^{-2 i \pi \beta_jC_{\ell j}} \frac{Q^{(j)}\Big(e^{i \pi C_{\ell j}}\la^*\Big)}{Q^{(j)}
\Big(e^{- i\pi C_{\ell j}}\la^*\Big)}=-1
\end{equation}
for every zero $\la^*$ of $Q^{(l)}(\la)$. In the above formula, $C_{ij}$ is the Cartan matrix of $\mf{g}$,
and the phases $\beta_j$ as well as the relevant analytic properties of the functions $Q$'s depend
on the parameters  $c$ and $p$ -- see \cite{baluzaI,bazhanov01,dorey07,marava15}.
\\

In order to define the ODE/IM correspondence, one needs the following data:
a principal nilpotent element $f \subset \mf{g}$,
a Cartan decomposition $\mf{g}=\mf{n}_-\oplus\mf{h}\oplus \mf{n}_+$ such that
$f=\sum_if_i \in \mf{n}_-$ where $f_i$'s are the negative Chevalley generators of $\mf{g}$,
the dual Weyl vector $\rho^\vee\in\mf{h}$, a highest root vector $e_{\theta}$, the
dual of the highest root $\theta^\vee \in \mf{h}$, an arbitrary but fixed
element of the Cartan subalgebra $r\in \mf{h}$, an arbitrary but fixed real number
$\hat{k} \in (0,1)$, an arbitrary complex parameter $\lambda \in \bb{C}$, and finally
a possibly empty finite set $J$ of pairs $(w_j,X(j)) \in  \bb{C}^* \times  \mf{n}_+$, to be determined.

Given the above data, we say that a \textbf{quantum $\widehat{\mf{g}}$-KdV oper} is an oper admitting the following representation
\begin{equation}\label{eq:operintro}
 \mc{L}(z,\la)=\de_z+\frac{r-\rho^\vee+f}{z}+
 (1+  \la z^{-\hat{k}})e_{\theta}+
 \sum_{j \in J}\frac{-\theta^\vee+
 X(j) }{z-w_j}  \; ,
\end{equation}
where in addition the regular singularities $\lbrace (w_j,X(j))\rbrace_{j \in J}$
have to be chosen so that the \eqref{eq:operintro} has trivial monodromy
at each $w_j$ for every value of $\lambda$. These further conditions ensure that the residues
$-\theta^\vee+X(j)$ belong to a $2\hv-2$ dimensional subspace
of $\mf{b}_+$, namely $\mf{t}=\bb{C} \theta^\vee \oplus [\theta^\vee,\mf{n}_+]$,
which is strictly related to the $\bb{Z}-$gradation on $\mf{g}$ induced by the element
$\theta^\vee$, and carries a  natural symplectic structure. The quantum $\widehat{\mf{g}}-$KdV opers \eqref{eq:operintro}
provide an explicit realization of the opers proposed by
  Feigin and Frenkel in the paper \cite{FF11} (see also \cite{fh16}),  which was the main inspiration of the present work.
\\

How does one attach a solution of the Bethe Ansatz equations to the above opers? The method was derived in our previous papers
on the ground state oper \cite{marava15,marava17}, which build on previous progresses by \cite{dorey07,Sun12}.
Given a quantum $\widehat{\mf{g}}$-KdV oper, a solution of the Bethe Ansatz equation is be constructed as follows, see Section \ref{sec:BA}.
One considers the regular singularity at $0$, and the irregular singularity at $\infty$ of \eqref{eq:operintro}. The
generalised monodromy data of the oper are encoded in
the connection matrix between these two singularities. This is obtained
by expanding, in every fundamental representation of $\mf{g}$, the subdominant solution at $\infty$ in
the basis of eigensolution of the monodromy operator. These coefficients are the so-called $Q$ functions,
which satisfy the $Q\widetilde{Q}$ system, and hence
and satisfy the Bethe Ansatz equations.

After having introduced the $Q$ functions, the complete Feigin-Frenkel ODE/IM conjecture \cite[Section 5]{FF11}
for the quantum $\widehat{\mf{g}}-$KdV model,
with $\mf{g}$ simply-laced, can be restated as follows.
\begin{conjecture}\label{conj:odeim}
To any state of the quantum $\widehat{\mf{g}}-$KdV model there corresponds a unique quantum $\widehat{\mf{g}}$KdV oper \eqref{eq:operintro} whose
$Q$ functions coincide with the solution of Bethe Ansatz
equations of the given state. Moreover, the level $N \in \bb{N}$ of a state coincides with the cardinality $N$
of the set $J$ of additional singularities of the corresponding oper. In particular, the ground state corresponds to the case
$J =\emptyset$.
\end{conjecture}

In the present paper we address this correspondence, and -- together with some side results which have
their own independent interest in the theory of opers -- we provide strong evidence of its validity by proving the following statements:
\begin{enumerate}
 \item[Statement 1] The $Q$ functions of the quantum $\widehat{\mf{g}}-$KdV opers \eqref{eq:operintro} are entire functions 
 of $\la$, are invariant under Gauge transformations,  and  satisfy the Bethe Ansatz equations \eqref{eq:TBAintro}.
 This proves Conjecture 8.1 of Frenkel and Hernandez \cite{fh16}.
 \item[Statement 2] The quantum $\widehat{\mf{g}}-$KdV opers \eqref{eq:operintro} are the most general opers whose $Q$ functions
 solve Bethe Ansatz
 equations \eqref{eq:TBAintro} .
 \item[Statement 3] The parameters  $\lbrace (w_j,X(j))\rbrace_{j \in J}$ of the additional singularities of the
  quantum $\widehat{\mf{g}}-$KdV opers \eqref{eq:operintro} are determined by a complete set of algebraic equations
 which are equivalent to the trivial monodromy conditions. In the particular case of a single additional singularity,
 and for generic values of the parameters  $\hat{k}$ and $r$,  there are $\rank \mf{g}$ distinct quantum $\widehat{\mf{g}}-$ KdV opers;
 this number coincides with the dimension of level $1$ subspace
 of the  quantum $\widehat{\mf{g}}-$KdV model.
 \end{enumerate}

 \begin{remark}
 Conjecture \ref{conj:odeim} does not exactly coincide with the original conjecture by Feigin and Frenkel \cite[Section 5]{FF11},
  because the explicit construction of the $Q$ functions as the coefficients of a connection problem was still unknown, in the general case,
  at the time when \cite{FF11} was written.
  Indeed, this construction was later achieved in full generality in our previous papers \cite{marava15,marava17},
  where we proved that coefficients of the connection problem satisfy a system of relations which goes under the name of $Q \widetilde{Q}$ system. The latter system
  was itself conjectured to hold by Dorey et al. \cite{dorey07} and further studied by Sun \cite{Sun12}.
  Remarkably, the same $Q \widetilde{Q}$ system was then showed by Frenkel and Hernandez \cite{fh16} to hold as a universal system of relations in the
commutative Grothendieck ring $K_0(\mc{O})$ of the category $\mc{O}$ of representations of the
Borel subalgebra of the quantum affine algebra $U_q(\widehat{\mf{g}})$, a category previously introduced by Hernandez and Jimbo in \cite{hj12}.
  
  Summarising, the state-of-the-art of the ODE/IM conjecture for the quantum $\widehat{\mf{g}}$-KdV model is the following 
  (see \cite{fh16} for a thorough discussion of this point).  We have a putative triangular diagram whose vertices are 1) the Quantum $\widehat{\mf{g}}$-KdV opers of Feigin and Frenkel, 2)
  the states of the Quantum $\widehat{\mf{g}}$-KdV model, and 3) the solutions of the $Q\widetilde{Q}$ system with the correct analytic properties.
  Two arrows are now well-defined. The first, from opers to solutions of the $Q\widetilde{Q}$ system,
  is provided by the present work, the second, from states to solutions of the $Q\widetilde{Q}$ system, is provided in \cite{fh16}.
  The conjecture will then be proved when a third and bijective arrow, from the states of the quantum $\widehat{\mf{g}}$-KdV model
  to quantum opers, will be defined in such a way to make the diagram will be commutative. 
 \end{remark}
 
 \begin{remark}In the $\widehat{\mf{sl}}_2$ case the opers \eqref{eq:operintro} were shown in \cite{FF11} to coincide
 -- up to a change of coordinates -- 
 to the Schr\"odinger operators with 'monster potential' studied by Bazhanov, Lukyanov, and Zamolodchikov \cite{BLZ04}. Hence, in this case Conjecture \ref{conj:odeim} coincides with the one stated in \cite{BLZ04}.
 \end{remark}
 
 \begin{remark}
  The quantum KdV opers \eqref{eq:operintro} can be either thought of as multivalued $\mf{g}$ opers, or as single valued
  -- i.e. meromorphic --  $\widehat{\mf{g}}$ opers \cite{FF11,fh16}. Both view points will be discussed in Section \ref{sec:kdvopers}.
 \end{remark}

\subsection*{Organization of the paper}
The paper is divided in three main parts.
\begin{enumerate}
 \item A preamble collecting some preliminary material, on simple and affine Lie algebras, on opers
 and on singularities of opers; Sections \ref{sec:algebras},\ref{sec:opers},\ref{sec:singularities}.
 \item The definition and analysis of quantum KdV opers, including the proof of statements 1,2 above; Sections
 \ref{sec:kdvopers},\ref{sec:BA},\ref{sec:miura}. 
 \item The analysis of the trivial monodromy conditions for the quantum KdV opers, including the proof of Statement 3;
 Sections \ref{sec:gradation},
\ref{sec:singularpoint},\ref{sec:zeromonodromykdv},\ref{sec:computations}, \ref{sec:A1}.
\end{enumerate}

The preamble mostly consists  of known material, but it contains a simple introduction to opers and their singularities
which may be useful to the reader. Our approach to opers is intended to be suitable to computations and to make the paper self-contained and
easily accessible.

The quantum KdV opers are axiomatically defined in Section \ref{sec:kdvopers}, following Feigin and Frenkel \cite{FF11}.
The axioms fix the singularities' structure of the opers. They are meromorphic opers on the sphere such that
$0$ and $\infty$ are singularities with fixed coefficients, and all other possible singular points are regular and have trivial monodromy. 
To begin our analysis we drop the axiom on the trivial monodromy and
deduce -- after fixing an arbitrary transversal space of $\mf{g}$ -- the canonical form of those opers which
satisfy all other axioms;
see Proposition \ref{pro:quasinormal}.
Such a canonical form does not coincide with \eqref{eq:operintro}, because in the canonical form a regular singularity is not 
a simple pole of the oper.

In Section \ref{sec:BA}, we study the generalised monodromy data of quantum KdV oper making use of their canonical form.
We define the $Q$ functions  and prove that they satisfy the $Q\widetilde{Q}$ relations and thus the
Bethe Ansatz equations, see Theorem \ref{thm:QQ}.
This section is based on our previous work \cite{marava15}, as well as on a new approach to the monodromy representation
of multivalued opers.

In Section \ref{sec:miura} we prove that the quantum $\widehat{\mf{g}}-$KdV opers are Gauge equivalent to a unique oper of the form \eqref{eq:operintro},
see Corollary \ref{cor:normals}. To this aim we introduce and study an extended Miura map. This is defined as the map
that to an oper whose singularities are first order poles associates its canonical form. We prove that the extended Miura map,
when appropriately restricted, is bijective.

The analysis of the trivial monodromy conditions for the quantum $\widehat{\mf{g}}-$KdV opers \eqref{eq:operintro}
is divided in the five remaining sections.

In sections  \ref{sec:gradation} and
\ref{sec:singularpoint} we study the Lie algebra grading induced by the element $\theta^\vee$,  and we write the trivial
monodromy conditions as a system of equations on the Laurent coefficients of the oper at the singular point.
One of these equations is linear and is equivalent to require that the elements $-\theta^\vee+X(j)$, for $j\in J$,  belong to the $2\hv-2$ dimensional  symplectic subspace $\mf{t}\subset\mf{b}_+$. We introduce a canonical basis for the symplectic form on $\mf{t}$ and use it in Section \ref{sec:zeromonodromykdv} to derive system \eqref{systemonc}, which is equivalent to the trivial monodromy conditions. This is a complete
system of $(2\hv-2)|J|$ algebraic equations in the
$(2\hv-2)|J|$ unknowns $\lbrace (w_j,X(j)) \rbrace_{j \in J}$, which fixes the additional singularities and thus
completely characterise the quantum KdV opers.

In Section \ref{sec:computations} we specialise  system \eqref{systemonc} to the cases of the Lie algebras
$A_n, n\geq 2$, $D_n, n\geq 4$, and $E_6$ (we omit to show our computations in the case $E_7,E_8$ due to their excessive length).
By doing so we reduce \eqref{systemonc} to a 
system of $2|J|$ algebraic equations in $2|J|$ unknowns.
Finally, in Section \ref{sec:A1} we deal with the case $\mf{g}=\mf{sl}_2$, which was already considered in \cite{BLZ04,FF11,fioravanti05}
and requires a separate study.

\subsection*{Acknowledgements}
The research project who led to the present paper started in December 2015, in Lisbon, during a fortuitous lunch with
Edward Frenkel, which resulted in the paper \cite{fh16} and in the present one. 
We are very grateful to Edward for insight and support, our work would not have been possible without his help,
and for useful comments on a preliminary version of this paper.
We also thank Tom Bridgeland, Giordano Cotti, David Hernandez, Sergei Lukyanov, Ara Sedrakyan, Roberto Tateo and
Benoit Vicedo for many useful discussions,
and Moulay Barkatou for pointing out to us reference \cite{bava83}. 

The authors gratefully acknowledge support form the Centro de Giorgi, Scuola Normale Superiore di Pisa,
where some of the research for the paper was performed.
Davide Masoero gratefully acknowledges support from the Simons Center for Geometry and Physics, Stony Brook University
at which some of the research for this paper was performed.
Andrea Raimondo thanks the Group of Mathematical Physics of Lisbon university for the kind hospitality,
during his many visits.

The authors are partially supported by the FCT Project PTDC/MAT-PUR/ 30234/2017 `Irregular
connections on algebraic curves and Quantum Field Theory'. D. M. is supported by the
FCT Investigator grant IF/00069/2015 `A mathematical framework for the ODE/IM correspondence',
and partially supported by the FCT Grant PTDC/MAT-STA/0975/2014 .

\vspace{20pt}

\section{Affine Kac--Moody algebras}\label{sec:algebras}
\subsection{Simple Lie algebras}
Let $\mf{g}$ be a simply-laced simple Lie algebra of rank $n$, and let $h^\vee$ be the dual Coxeter number of $\mf{g}$\footnote{Since $\mf{g}$ is simply-laced, then $h^\vee=h$, the Coxeter number of $\mf{g}$. We prefer to use $h^\vee$ in place of $h$ in view of the extension of the results of the present paper to a generic (simple) Lie algebra $\mf{g}$, in which case the dual Coxeter number appears.}. Let $\mf{h}$ be a Cartan subalgebra and $\Delta\subset \mf{h}^\ast$ be the set of roots relative to $\mf{h}$.  The algebra $\mf{g}$ admits the roots space decomposition
\beq\label{rootdecomposition}
\mf{g}=\mf{h}\oplus\bigoplus_{\alpha\in\Delta}\mf{g}_\alpha,
\eeq
where $\mf{g}_\alpha=\{x\in\mf{g}\,|\,[h,x]=\alpha(h)x, h\in\mf{h}\}$ is the root space corresponding to the root $\alpha$. Set $I=\{1,\dots,n\}$. Fix a set of simple roots $\Pi=\left\{\alpha_i,i\in I\right\}\subset \Delta$, let $\Delta_+\subset \Delta$ be the corresponding set of positive roots and $\Delta_-=\Delta\setminus \Delta_+$ the negative roots.  For $\alpha=\sum_i m_i\alpha_i\in \Delta$ define its height as $\Ht(\alpha)=\sum_im_i\in \bb{Z}$.  Let $\Pi^\vee=\{\alpha_i^\vee, i\in I\}\subset\mf{h}$ be simple coroots, satisfying $\langle \alpha_i^\vee,\alpha_j\rangle=C_{ij}$ where $C=(C_{ij})_{i,j\in I}$ is the Cartan matrix of $\mf{g}$. Let
$$Q=\bigoplus_{j\in I}\bb{Z}\alpha_j,\qquad Q^\vee=\bigoplus_{j\in I}\bb{Z}\alpha_j^\vee$$
be respectively the root and the coroot lattice of $\mf{g}$.  Let $\mc{W}$ be the Weyl group of $\mf{g}$, namely the finite group generated by the simple reflections 
$$\sigma_i(\alpha_j)=\alpha_j-C_{ij}\alpha_i,\qquad i,j\in I.$$
The above action on $\mf{h}^\ast$ induces an action of $\mc{W}$ on $\mf{h}$, with simple reflections given by
$$\sigma_i(\alpha_j^\vee)=\alpha_j^\vee-C_{ji}\alpha_i^\vee,\qquad i,j\in I.$$
Denote by $\{\omega_i, i\in I\}$ (resp. $\{\omega^\vee_i, i\in I\}$) the fundamental weights (resp. coweights) of $\mf{g}$, defined by the relations
$$\alpha_i=\sum_{j\in I}C_{ji} \omega_j,\quad \alpha_i^\vee=\sum_{j\in I}C_{ji} \omega^\vee_j \qquad i\in I.$$
Corresponingly, we denote by
\beq\label{weightcoweight}
P=\bigoplus_{j\in I}\bb{Z}\omega_j,\qquad P^\vee=\bigoplus_{j\in I}\bb{Z}\omega^\vee_j
\eeq
the weight and coweight lattices of $\mf{g}$. For every $\omega\in P$, we denote by $L(\omega)$ the irreducible finite dimensional highest weight $\mf{g}-$module with highest weight $\omega$.  

Let $\{e_{i},f_{i},i\in I\}$ be Chevalley generators of $\mf{g}$, satisfying the relations
\beq
[\alpha_i^\vee,e_{j}]=C_{ij}e_{j}, \quad [\alpha_i^\vee,f_{j}]=-C_{ij}f_{j},\quad [e_{i},f_{j}]=\delta_{ij}\alpha_i^\vee
\eeq
for $i,j\in I$.  Let $\mf{n}_+$ (resp. $\mf{n}_-$) the nilpotent subalgebra of $\mf{g}$ generated by $\{e_i,i\in I\}$ (resp. $\{f_i,i\in I\}$), and recall the Cartan decomposition
$\mf{g}=\mf{n}_-\oplus\mf{h}\oplus \mf{n}_+$. In addition, denote $\mf{b}_+=\mf{h}\oplus\mf{n}_+$ the Borel subalgebra associated to the pair $(\mf{g},\mf{h})$.
%
Let $\mc{G}$ be the adjoint group of $\mf{g}$, denote by $\mc{B}$ the (maximal) solvable subroup of $\mc{G}$ whose Lie algebra is $\mf{b}_+$, by $\mc{H}$ the abelian torus with Lie algebra $\mf{h}$ and by $\mc{N}$ the unipotent subgroup of $\mc{G}$ whose Lie algebra is $\mf{n}_+$. Then $\mc{N}$ is a normal subgroup of $\mc{B}$ and $\mc{B}=\mc{N}\rtimes \mc{H}$. Consider the exponential map $\exp:\mf{n}_+\to\mc{N}$. Given $y\in\mf{n}_+$, the adjoint action of $\exp(y)\in\mc{N}$ on $\mf{g}$ is given by
$$\exp(y).x=x+\sum_{k\geq 1}\frac{1}{k!}(\ad_y)^k x,\qquad \,x\in \mf{g},$$
where $\ad_yx=[y,x]$.
Define a bilinear non-degenerate symmetric form $(\cdot\vert\cdot)$ on $\mf{h}$ by the equations
\beq\label{hhC}
(\alpha_i^\vee\vert\alpha_j^\vee)=C_{ij}, \qquad i,j\in I
\eeq
and introduce the induced isomorphism $\nu:\mf{h}\to\mf{h}^\ast$ as
$$\langle h',\nu(h)\rangle=(h'\vert h), \qquad h,h'\in\mf{h}.$$
Note that in particular we have $\nu(\alpha_i^\vee)=\alpha_i$, $i\in I$, and the induced bilinear form $(\cdot |\cdot)$ on $\mf{h}^\ast$ satisfies:
\beq\label{aaC}
(\alpha_i\vert\alpha_j)=C_{ij}, \qquad i,j\in I
\eeq
follows. As proved in \cite{kac90}, there exists a (unique) nondegenerate invariant symmetric bilinear form $(\cdot | \cdot)$ on $\mf{g}$ such that
\begin{subequations}\label{normalizedkilling}
\begin{align}
& (\mf{h}|\mf{h})\,\,\text{is defined by \eqref{hhC}},\\
& (\mf{g}_\alpha | \mf{h})=0,\quad\alpha\in\Delta,\\
& (\mf{g}_\alpha | \mf{g}_\beta)=0\quad \alpha,\beta\in\Delta, \,\alpha\neq-\beta,\\
& [x,y]=(x|y)\nu^{-1}(\alpha), \quad x\in\mf{g}_\alpha, y\in\mf{g}_{-\alpha}, \alpha\in\Delta.
\end{align}
\end{subequations}
We will consider this bilinear form on $\mf{g}$ from now on.

Let $\rho=\sum_{i\in I}\omega_i\in \mf{h}^\ast$ be the Weyl vector, and denote
$$\rho^\vee=\nu^{-1}(\rho)=\sum_{i,j\in I}(C^{-1})^{ij}\alpha_j^\vee.$$ 
The \emph{principal gradation} of $\mf{g}$ is defined as
\beq
\mf{g}=\bigoplus_{i=-h^\vee+1}^{h^\vee-1}\mf{g}^i, \qquad \mf{g}^i=\left\{x\in\mf{g}\;\vert\; [\rho^\vee,x]=ix\right\}.
\eeq
We denote by $\pi^j$ the projection from $\mf{g}$ onto $\mf{g}^j$:
\beq
\pi^j: \;\mf{g}\to \mf{g}^j\label{projectionprinc}
\eeq
The element
\beq
f=\sum_{i\in I}f_i,
\eeq
is a principal nilpotent element. Clearly, $f\in\mf{g}^{-1}$, and moreover one can prove \cite{Kos59} that $f$ satisfies the following properties: $\Ker\ad_f\subseteq \mf{n}_-$,  $[f,\mf{n}_+]\subset \mf{b}_+$ and $\ad_{\rho^\vee}[f,\mf{n}_+]\subseteq [f,\mf{n}_+]$. Since $\rho^\vee$ is semisimple, it follows that there exists an $\ad_{\rho^\vee}$-invariant subspace $\mf{s}$ of $\mf{b}_+$ such that
\beq\label{bfns}
\mf{b}_+=[f,\mf{n}_+]\oplus\mf{s},
\eeq
and since $(\Ker\ad_f)|_{\mf{n}_+}=0$, then $\dim\mf{s}=\dim\mf{b}_+-\dim\mf{n}_+=n$.
The choice of $\mf{s}$ is not unique, and as a possible choice of $\mf{s}$ one can always take $\mf{s}=\Ker\ad_e$, where $e$ is that unique element of $\mf{g}$ such that $\{f,2\rho^\vee,e\}$ is an $\mf{sl}_2$-triple. However, in this paper we do not make this specific choice, and we consider an arbitrary subspace $\mf{s}$ satisfying \eqref{bfns}. The affine subspace $f+\mf{s}$ is known as transversal subspace; by a slight abuse of terminology, we also refer to the subspace $\mf{s}$ as a transversal subspace. The space $f+\mf{s}$ has the property that every regular orbit of $\mc{G}$ in $\mf{g}$ interects $f+\mf{s}$ in one and only one point. In addition, for every $x\in \mf{b}_+$ there exist a unique $a\in\mc{N}$ and a unique $s\in \mf{s}$ such that $a.(f+s)=f+x$, where $a.$ denotes the adjoint action of $\mc{N}$ on $\mf{g}$. More precisely, the map
$$\mc{N}\times (f+\mf{s})\to f+\mf{b}_+$$
provided by the adjoint action is an isomorphism of affine varieties.  Last, we introduce the concept of exponents of the Lie algebra $\mf{g}$. Decomposing \eqref{bfns} with respect to the principal gradation, one obtains a set ot equations of the form $\mf{g}^i=[f,\mf{g}^{i+1}]\oplus\mf{s}^i$, for $i=0,\dots,h^\vee-1$, where $\mf{s}^i=\mf{s}\cap\mf{g}^i$ and $\mf{g}^{h^\vee}=\{0\}$. Since $(\Ker\ad_f)|_{\mf{n}_+}=0$, then $\dim\mf{s}^i=\dim\mf{g}^i-\dim\mf{g}^{i+1}$.  If $\dim\mf{s}^i>0$, then $i$  is
said to be an exponent of $\mf{g}$, and $\dim\mf{s}^i$ is the multiplicity of the exponent $i$.
Counting multiplicities, there are $n=\rank\mf{g}$ exponents, which we denote by $d_1,\dots,d_n$.

\subsection{ A basis for $\mf{g}$}\label{gbasis} Let $\{e_i,\alpha^\vee_i,f_i,i\in I\}\subset\mf{g}$ be generators of $\mf{g}$ defined as above. Following \cite{kac90} we define a basis for $\mf{g}$ as follows. For every pair of simple roots $\alpha_i,\alpha_j$, $i,j\in I$, let 
\beq\label{bimultiplicative1}
\varepsilon_{\alpha_i,\alpha_j}=
\begin{cases}
(-1)^{C_{ij}} \qquad & i<j,\\
-1 \qquad & i=j,\\
1 \qquad & i>j,
\end{cases}
\eeq
and extend this to a function $\varepsilon:Q\times Q\to \{\pm 1\}$ by bimultiplicativity:
\beq\label{bimultiplicative2}
\varepsilon_{\alpha+\beta,\gamma}=\varepsilon_{\alpha,\gamma}\varepsilon_{\beta,\gamma},\quad \varepsilon_{\alpha,\beta+\gamma}=\varepsilon_{\alpha,\beta}\varepsilon_{\alpha,\gamma}, \qquad \alpha,\beta,\gamma\in Q.
\eeq
Then, for $\alpha\in \Delta$ there exists nonzero $E_\alpha\in\mf{g}$,  with $E_{\alpha_i}=e_i$, $E_{-\alpha_i}=-f_i$, $i\in I$, uniquely characterized by the relations
\beq\label{commrel}
\begin{cases}
[h,E_\alpha]=\langle h, \alpha\rangle E_\alpha,\qquad &h\in\mf{h},\,\alpha\in\Delta,\\
[E_\alpha,E_{-\alpha}]=-\nu^{-1}(\alpha),\qquad &\alpha\in\Delta,\\
[E_\alpha,E_\beta]=\varepsilon_{\alpha,\beta}E_{\alpha+\beta},\qquad &\alpha,\beta,\alpha+\beta\in\Delta,\\
[E_\alpha,E_\beta]=0,\qquad  &\alpha,\beta\in\Delta,\,\alpha+\beta\notin\Delta\cup\{0\}.
\end{cases}
\eeq
We clearly have the root space decomposition
\beq\label{rootspacedecomposition}
\mf{g}=\mf{h}\oplus\bigoplus_{\alpha\in\Delta}\bb{C}E_\alpha.
\eeq 
In addition, it follows from \eqref{normalizedkilling} and \eqref{commrel} that for $\alpha, \beta\in\Delta$ we have
\beq\label{Enormalized}
(E_\alpha|\mf{h})=0,\qquad (E_\alpha|E_\beta)=-\delta_{\alpha,-\beta}, 
\eeq
where $(\cdot | \cdot)$ is the normalized invariant form defined in \eqref{normalizedkilling}. The following result will be useful in Section \ref{sec:zeromonodromykdv}.
\begin{lemma}\label{lemmaepsilon} 
For every $\beta,\gamma\in Q$, then
\begin{enumerate}[i)]
\item  $\varepsilon_{\beta,-\gamma}=\varepsilon_{\beta,\gamma}=\varepsilon_{-\beta,\gamma}$,
\item $\varepsilon_{0,\beta}=\varepsilon_{\beta,0}=1$,
\item $\varepsilon_{\beta,\beta}=(-1)^{\frac{1}{2}(\beta | \beta)}$,
\item $\varepsilon_{\beta,\alpha}\varepsilon_{\alpha,\beta}=(-1)^{(\alpha|\beta )}$.
\end{enumerate}
\end{lemma}
\begin{proof}
i) Let $\beta=\sum_j \beta^j\alpha_j, \gamma=\sum_j \gamma^j\alpha_j\in Q$. Then using \eqref{bimultiplicative2} we have 
\begin{align*}
\varepsilon_{\beta,-\gamma}&=\prod_{i,j=1}^n(\varepsilon_{\alpha_i,\alpha_j})^{-\beta_i\gamma_j}=\prod_{j=1}^n\left(\prod_{i<j}(-1)^{-\beta_iC_{ij}\gamma_j}(-1)^{-\beta_j\gamma_j}\right)\\
&=\prod_{j=1}^n\left(\prod_{i<j}(-1)^{\beta_iC_{ij}\gamma_j}(-1)^{\beta_j\gamma_j}\right)=\varepsilon_{\beta,\gamma}.
\end{align*}
ii) From point i) we get $\varepsilon_{0,\beta}=\varepsilon_{\gamma-\gamma,\beta}=\varepsilon_{\gamma,\beta}\varepsilon_{-\gamma,\beta}=(\varepsilon_{\gamma,\beta})^2=1$. \\
iii) Let $\beta=\sum_{i=1}^n\beta^i\alpha_i\in Q$, so that $(\beta |\beta)=\sum_{ij}\beta^iC_{ij}\beta^j=2\sum_{j}\left(\sum_{i<j}\beta^iC_{ij}\beta^j+(\beta^j)^2\right)$, where in the last equality we used the relations $C_{ji}=C_{ij}$ and $C_{ii}=2$. Thus we have
 $$\sum_{j=1}^n\left(\sum_{i<j}\beta^iC_{ij}\beta^j+(\beta^j)^2\right)=\frac{1}{2}(\beta | \beta)$$
 for every $\beta=\sum_i\beta^i\alpha_i\in Q$. We now compute $\varepsilon_{\beta,\beta}$. Using \eqref{bimultiplicative1} and \eqref{bimultiplicative2} we obtain
\begin{align*}
\varepsilon_{\beta,\beta}&=\prod_{i,j=1}^n(\varepsilon_{\alpha_i,\alpha_j})^{\beta_i\beta_j}=\prod_{j=1}^n\left(\prod_{i<j}(-1)^{\beta_iC_{ij}\beta_j}(-1)^{(\beta_j)^2}\right)\\
&=(-1)^{\sum_{j=1}^n\left(\sum_{i<j}\beta_iC_{ij}\beta_j+(\beta_j)^2\right)}=(-1)^{\frac{1}{2}(\beta | \beta)}.
\end{align*}
iv) Replacing in iii) $\beta$ with $\alpha+\beta$ and using \eqref{bimultiplicative2} we get $\varepsilon_{\beta,\alpha}\varepsilon_{\alpha,\beta}=(-1)^{(\alpha |\beta)}$.  \end{proof}

\subsection{Affine Kac--Moody algebras}

Let $\mf{g}$ be a simple Lie algebra, $\mf{h}\subset \mf{g}$ a Cartan subalgebra, and fix a nondegenerate invariant bilinear form $(\cdot | \cdot)$ on $\mf{g}$ as in \eqref{normalizedkilling}. The untwisted affine Kac-Moody algebra $\hat{\mf{g}}$ associated to the simple Lie algebra $\mf{g}$ can be realized in terms of $\mf{g}$ as the space
$$\hat{\mf{g}}=\mf{g}[\lambda,\lambda^{-1}]\oplus \bb{C}K \oplus \bb{C}d,$$
with the commutation relations
\begin{align*}
[\lambda^m\otimes & x\oplus a K\oplus b d,\lambda^n\otimes y\oplus a' K\oplus b' d]\\&=
\left(\lambda^{m+n}\otimes [x,y]-b'm \,\lambda^m\otimes x+b\, n \,\lambda^n\otimes y\right)\oplus m\delta_{m,-n}(x|y)K, \notag
\end{align*}
where $a,b,a',b'\in\bb{C}$, $m,n\in\bb{Z}$ and $x,y\in\mf{g}$. Note that $K$ is a central element,
while $d$ acts as the derivation $\lambda \partial_\la$. The Cartan subalgebra of $\hat{\mf{g}}$ is the finite dimensional subalgebra 
$$\hat{\mf{h}}=\mf{h}\oplus \bb{C}K\oplus\bb{C}d.$$
Let $\{e_i,f_i,i\in I\}$ be Chevalley generators of $\mf{g}$, as above, and for $i\in I$ set
$\hat{e}_i=1\otimes e_i$ and $\hat{f}_i=1\otimes f_i$. Moreover, let $e_\theta\in\mf{g}_\theta$
(resp. $e_{-\theta}\in\mf{g}_{-\theta}$) be a highest (resp. lowest) root vector for $\mf{g}$ and set
$\hat{e}_0=\lambda^{-1}\otimes e_{-\theta}$, $\hat{f}_0=\lambda\otimes e_{\theta}$. Putting
$\hat{I}=\{0,\dots,n\}$, then $\{\hat{e}_i,\hat{f}_i,i\in\hat{I}\}$ is a set of generators for
$\hat{\mf{g}}$. We denote by $\hat{f}$ the element $\hat{f}=\sum_{i \in \hat{I}}\hat{f}_i$.

\vspace{20pt}
\section{Opers}\label{sec:opers}
In this Section we review the concept of $\mf{g}$-\textit{opers}  and some of its basic theory. This is done in order to keep
the paper as self-contained as possible and to fix the notation; consequently we
follow a basic and purely algebraic approach, suitable
to computations. For more details on the subject, including the geometric approach and
the extension to more general groups and algebras, the reader may consult
\cite{DS85,bedr02,fr07,lacroix18} and references therein.

For any open and connected subset $D$ of the Riemann sphere $\bb{P}^1$, we call $O_D$ the ring of regular functions on $D$, and
$K_D$ the field of meromorphic functions on it. Given a $\bb{C}$ vector space $V$, we denote $V(O_D)=O_D \otimes V$ and
$V(K_D)=K_D \otimes V$, namely the space of
the regular/meromorphic functions on $D$ with values in $V$. Opers are, locally, equivalence classes of differential operators modulo
Gauge transformations. In this work  we consider classes of meromorphic differential operators
modulo meromorphic Gauge transformations. 

The local operators under consideration belong to the classes $\ope(K_D),\tope(K_D)$, which we define below.
\begin{definition}\label{def:opetope} Let $z$ be a local holomorphic coordinate on $\bb{P}^1$ that identifies $\bb{P}^1$ with $\bb{C} \cup \lbrace \infty \rbrace$, and let $\mc{L}$ be a differential operator in $z$. We say that $\mc{L}$ belongs to $\ope(K_D)$ if it is of the form  \cite{bedr02}:
\beq\label{eq:Lgeneral}
\mc{L}=\partial_z+ f+ b
\eeq
for some $b \in \mf{b}_+(K_D)$.  We say that  $\mc{L}$ belongs to $\tope(K_D)$ if it is of the form 
\beq\label{eq:tLgeneral}
\mc{L}=\partial_z+\sum_{i=1}^n \psi_i f_i+ b
\eeq
where $b \in \mf{b}_+(K_D)$, and 
$\psi_i \in K_D \setminus \lbrace 0 \rbrace$, $i=1,\dots,n$.
\end{definition}
The local Gauge groups we consider are $\mc{N}(K_D),\mc{H}(K_D), \mc{B}(K_D)$, which we introduce below together
with their actions on $\ope(K_D),\tope(K_D)$.
\begin{definition}\label{def:gaugegroupN}
  The unipotent Gauge group is the set
 \begin{equation}
 \mc{N}(K_D)=\lbrace \exp{y}\, , \; y \in n_+(K_D) \rbrace
 \end{equation}
 with the natural group structure inherited from $\mathcal{N}$. The (adjoint) action of $\mc{N}(K_D)$ on $\mf{g}(K_D)$ is defined as
  \begin{equation*}
  \exp{(\ad y)}.g=\sum_{k\geq0}\frac{1}{k!}(\ad_y)^k \, g, \qquad y \in \mf{n}_+(K_D),\, g \in \mf{g}(K_D).
 \end{equation*}
 The adjoint action of $\mc{N}(K_D)$ on $\partial_z$ is expressed by Dynkin's formula 
\begin{equation}\label{dynkinformula}
 \exp{(\ad y)}.\partial_z=\partial_z-\sum_{k\geq0}\frac{1}{(k+1)!}(\ad_{y})^k\frac{d y}{dz} , \qquad y \in \mf{n}_+(K_D),
\end{equation}
which is equivalent to
$N.\partial_z=\partial_z-\frac{dN}{dz}N^{-1}$, for $N=\exp{y}$.
\end{definition}
\begin{remark}
Let us extend the algebra structure of $\mf{n}_+(K_D)$ to the space $\mf{n}_+(K_D) \oplus \bb{C} \partial_z$ by the formula $[\partial_z,y]=\frac{dy}{dz}$. Then formula  \eqref{dynkinformula} for the action of $\exp{y}$ on $\partial_z$
coincides with 
the adjoint action according to the bracket of the extended algebra. Indeed,
$$\sum_{l=0}^{\infty}\frac{1}{l!}(\ad_y)^l\partial_z=
\partial_z+\sum_{l=1}^{\infty}\frac{1}{l!}\ad^{l-1}_{y}[y,\partial_z]=
\partial_z-\sum_{k=0}^{\infty}\frac{1}{(k+1)!}\ad^k_y \frac{dy}{dz}.$$
\end{remark}

\begin{definition}\label{def:gaugegroupH}
We denote $\mc{H}(K_D)$ the abelian mutiplicative group generated by
elements of the form $\phi^{\lambda}$ for $\phi \in K_D\setminus \lbrace 0 \rbrace$ and $\lambda \in P^{\vee}$, the co-weight lattice \eqref{weightcoweight}. Since $\rank P^\vee=n$ then $\mc{H}(K_D)$ is isomorphic to
$\left( K_D\setminus \lbrace 0 \rbrace \right)^n $
The (adjoint) action of $\mc{H}(K_D)$ on $\mf{g}(K_D)$ is given   by means of the root space decomposition \eqref{rootdecomposition}: if $g=g_0+\sum_{\alpha \in \Delta} g_{\alpha}\in\mf{g}(K_D)$, with  $g_0\in\mf{h}(K_D)$ and $g_{\alpha} \in \mf{g}_{\alpha}(K_D)$ then 
 \begin{equation*}
  \phi^{\ad\lambda}.g=g_0+ \sum_{\alpha \in \Delta} \phi^{\alpha(\lambda)}g_{\alpha}  \; .
 \end{equation*}
The adjoint action of $\mc{H}(K_D)$ on the operator $\partial_z$ is given by
\begin{equation*}
 \phi^{\ad\lambda}. \partial_z=\partial_z -\frac{\phi'}{\phi}\lambda \;.
\end{equation*}
Finally, the action of $\mc{H}(K_D)$ on $\mf{n}_+(K_D)$ induces an action on $\mc{N}(K_D)$ as follows
\begin{equation*}
\phi^{\ad\lambda}.\exp{y}=\exp{\big(\phi^{ \ad\lambda}. y\big)},\qquad  y \in \mf{n}_+(K_D), \, \phi \in K_D\setminus \{0\}.
\end{equation*}
\end{definition}

\begin{definition}\label{def:gaugegroupB}
Given the above action of $\mc{H}(K_D)$ on $\mc{N}(K_D)$, we define $\mc{B}(K_D)=\mc{H}(K_D) \rtimes \mc{N}(K_D)$ as the semidirect product induced by it.
\end{definition}

Summing up the previous definitions, we can explicitly write the free action of
$\mc{N}(K_D)$ on  $\tope(K_D)$ (and in particular on $\ope(K_D)$) as
\begin{align}\label{eq:gaugeaction}
 \exp{(\ad y)} . \big( \partial_z+\sum_{i}\psi_if_i+b\big)=\,&\partial_z+ \sum_{k\geq0}  \frac{1}{k!}(\ad_y)^k \big(\sum_i \psi_i f_i+b\big)  \nonumber\\
 &-\sum_{k\geq0}\frac{1}{(k+1)!}(\ad_y)^k\frac{d y}{dz},
\end{align}
with $y\in\mf{n}_+(K_D)$. Similarly, the action of $\mc{H}(K_D)$, and thus
of $\mc{B}(K_D)$, on $\tope(K_D)$ is given by:

\begin{align}\label{eq:Haction}
\phi^{\ad\lambda} .\big( \partial_z+\sum_{i}\psi_if_i+b\big)=&
\partial_z+ \sum_{i}\phi^{-\alpha_i(\lambda)}\psi_if_i+\nonumber\\
&- \frac{\phi'}{\phi}\lambda+b_0+ \sum_{\alpha \in \Delta_+}\phi^{\alpha(\lambda)}b_{\alpha},
\end{align}
where $b=b_0+\sum_{\alpha\in\Delta_+}b_\alpha$, with $b_0\in\mf{h}(K_D)$ and $b_\alpha\in\mf{g}_\alpha(K_D)$. Formula \eqref{eq:Haction} has two immediate consequences:
\begin{enumerate}
 \item The only element in $\mc{H}(K_D)$ that leaves the set $\ope(K_D)$ invariant is the identity
 \item For any choice of the functions $\psi_i$, $i\in I$, there is a unique element in $\mc{H}(K_D)$ that maps
 $\mc{L} \in \tope(K_D)$ to an operator in $\ope(K_D)$; explicitly this is
$\prod_j \psi_j(z)^{\omega_j^\vee}$ where $\omega_j^\vee$, $j\in I$, are the fundamental co-weights.
\end{enumerate}
It follows from the above that there is a bijection between the sets of equivalence classes 
 $\ope(K_D)/\mc{N}(K_D)$ and $\tope(K_D)/\mc{B}(K_D)$.
\begin{definition}Let $D$ be an open, connected and simply-connected subset of $\bb{P}^1$. The space of opers $\Ope(D)$ is defined as $\ope(K_D)/\mc{N}(K_D)\cong\tope(K_D)/\mc{B}(K_D)$. We denote by $[\mc{L}]$ the equivalence class (i.e. the oper) of the operator $\mc{L}$.
\end{definition}
%

Fixed a transversal space $f+\mf{s}$, then each equivalence class of operators in $\ope(K_D)$ admits a
unique representative of the form $\partial_z+f+s$, with $s\in\mf{s}(K_D)$ The space of opers on a domain $D$ of the Riemann sphere was essentially described (in the holomorphic case) in \cite{DS85}; in the sequel we need a slightly extended version of that proposition, hence we review its proof too.

\begin{definition}\label{def:pifpis}
Let $f+\mf{s}$ a transversal space. Given the splitting $\mf{b}_+=[f,\mf{n}_+]\oplus \mf{s}$, 
we denote
 $\Pi_{f}:\mf{b}_+\to [f,\mf{n}_+]$ and $\Pi_{\mf{s}}:\mf{b}_+\to \mf{s}$ the respective projections. 
\end{definition}
\begin{proposition}[cf. Proposition 6.1 in \cite{DS85}]\label{lem:gaugemiura} Let $f+\mf{s}$ be a transversal space.
 For every meromorphic differential operator $\mc{L}=\partial_z+f+b \in \ope(K_D)$, there exists a
 unique meromorphic function $s \in \mf{s}(K_D)$ and a unique Gauge transform $N\in \mc{N}(K_D)$ such that
 $N.\mc{L}=\partial_z+f+s$. Furthermore, the set of singular points of $s$ is a subset of the set of singular points of $b$. 
 \begin{proof}
 We first prove the existence of the pair $N,s$, and then its uniqueness. We construct -- by induction with respect to the principal gradation -- the pair $N,s$ as $N=N_{h^\vee-1}\cdots  N_1$, with $N_i=\exp y^i$ and $y^i\in\mf{g}^i(K_D)$ and  $s=\sum_{i=1}^{h^\vee-1}s^i$, with $s^i\in\mf{s}^i(K_D)$. Let $\mc{L}=\partial_z+f+b\in\ope(K_D)$, and let $b=\sum_{i=0}^{h^\vee-1}b^i$, with $b^i\in\mf{g}^i(K_D)$. Introduce $N_1=\exp y^1$ with $y^1\in\mf{g}^1(K_D)$  and set $\mc{L}_1=N_1\mc{L}$. Due to  \eqref{eq:gaugeaction} then $\mc{L}_1=\partial_z+f+b^0+[y^1,f]+\sum_{i=1}^{h^\vee-1}\bar{b}^i$, for certain $\bar{b}^i\in\mf{g}^1(K_D)$. Note that $b^0\in\mf{h}\subseteq [f,\mf{n}_+]$, and since $\Ker\ad_f$ is trivial on $\mf{n}_+$, we take $y^1$ to be the unique solution of the equation $b^0+[y^1,f]=0$, so that $\mc{L}_1$ takes the form $\mc{L}_1=\partial_z+f+\sum_{i=1}^{h^\vee-1}\bar{b}^i$. Note that by construction $y^1$ has at most the same singularities of $b^0$, which is an element of $\mc{L}$. Since $\mc{L}_1$ is generated from $\mc{L}$ by the iterated adjoint action of $y^1$, then the set of singular points of $\mc{L}_1$ is contained in the set of singular point of $\mc{L}$. Now fix $j>1$ and assume we found elements $N^{l}=\exp{n^l}$, $l=1,\dots,j-1$ with  $n^l \in \mf{g}^l(K_{D})$, as well as $s^l\in\mf{s}^l(K_D)$, $l=1,\dots,j-2$ such that $$\mc{L}_{j-1}:=N_{j-1}\cdots N_1.\mc{L}=\partial_z+f+\sum_{l=1}^{j-2}s^l+\sum_{l=j-1}^{h^\vee-1}c^l,$$
 for some $c^l\in\mf{g}^l(K_D)$. Assume moreover that the set of singular points of $\mc{L}_{j-1}$ is contained in $S_b$. Introduce $N_j=\exp y^j$ with $y^j\in\mf{g}^j(K_D)$ and set $\mc{L}_j=N_j\mc{L}_{j-1}$. Using \eqref{eq:gaugeaction} we obtain
$$\mc{L}_j=\partial_z+f+\sum_{l=1}^{j-2} s^l+ [y^j, f]+ \bar{c}^{j-1}+\sum_{l=j}^{h^\vee-1} \bar{c}^l,$$
for some $\bar{c}^l\in\mf{g}^l(K_D)$. We are interested in the term $[y^j, f]+ \bar{c}^{j-1}$. Recalling the projection operators given in Definition \ref{def:pifpis}, then we define $s^j=\Pi_{\mf{s}}(\bar{c}^{j-1})\in\mf{s}^j(K_D)$,  and we take $y^j$ to be the unique solution of the equation $[y^j,f]+\bar{c}^{j-1}=0$. Such a solution exists and is unique since $(\Ker\ad_f)|_{\mf{n}_+}=0$. Then, $\mc{L}_j$ takes the form
$$\mc{L}_j=\partial_z+\sum_{l=1}^{j-1} s^l+\sum_{l=j}^{h^\vee-1} \bar{c}^l.$$
By construction, $y^j$ has at most the same singularities of $\bar{c}^{j-1}$, which is an element of $\mf{L}_{j-1}$. Since $\mc{L}_j$ is generated by the action of $y^j$ on $\mc{L}_{j-1}$, the singular locus of
$\mc{L}_j$ is a subset of the singular locus of $\mc{L}_{j-1}$. Iterating the above procedure, one obtains elements $N=N_{h^\vee-1}\dots N_1$, with $N_j=\exp y^j$ and $y^j\in\mf{g}^j(K_D)$, and $s=\sum_{i=1}^{h^\vee-1}s^i$ with $s^i\in\mf{s}^i(K_D)$, so that $N.\mc{L}=\partial_z+f+s$, and the set of singular points of $N.\mc{L}$ -- namely the singular points of $s$ -- is contained in the set of singular points of $\mc{L}$,
namely the singular points of $b$. Note incidentally that $s^i=0$ if $i$ is not an exponent of $\mf{g}$.
  
  The pair $(N,s)$ constructed above is unique, because the action of $\mc{N}(K_D)$ on $\ope(K_D)$ is free, and if two operators of the form
  $\partial_z+f+s$, $\partial_z+f+s'$, with $s, s' \in \mf{s}(K_D)$ are gauge equivalent then $s=s'$.
  We prove the latter
  statement as follows. Let the two operators be Gauge equivalent, by the transformation $M=\exp{m}, m \in \mf{n}_+(K_D)$, then $m=0$.
  Indeed suppose $m\neq0$ and let $m^i \neq0, m_i \in \mf{g}^i(K_D)$ be the non-trivial term of
  $m$ with lowest principal degree.
  Then $\Pi_f(\exp{m}(\partial_z+f+s)-\partial_z-f)$ has a non-trivial term
  of degree $i-1$, namely $[m^i,f]$, hence it is not zero.
 \end{proof}
\end{proposition}

As a corollary we have the following characterisation of opers
\begin{proposition}\label{pro:opersondisc}
 Let $D$ be an open, connected, and simply connected subset of $\bb{C}$. After fixing a transversal space $\mf{s}$, the set 
 $\Ope(D)$ can be identified with $\mf{s}(K_D)$.
\end{proposition}

\begin{definition}
 We say that an operator $\mc{L} \in \ope(K_D)$ is in canonical form if it is of the form $\mc{L}=\partial_z+f+s$ with  $s \in\mf{s}(K_D)$. We also say that $\mc{L}_s=\partial_z+f+s$ with $s\in\mf{s}(K_D)$ is the canonical form of any element of $\tope(K_D)$ Gauge-equivalent to it.
\end{definition}

\subsection{Change of coordinates - Global theory}
The global theory of opers was developed in
\cite[Section 3]{bedr02}; see also \cite[Chapter 4]{fr07} or \cite[Section 6.1]{lacroix18}, which we follow.
Here we just address the simplest aspect of the global theory, that is the coordinate transformation laws of opers.
Let $\Sigma$ be a Riemann surface (we will be interested here in the case
$\Sigma=\bb{C}\bb{P}^1$ only), and $D$ a chart on $\Sigma$ with coordinate $z$. Let $\mc{L}\in \tope(K_D)$ be of the form 
$$\mc{L}=\partial_z+\sum_i \psi_i(z) f_i+b(z).$$ 
If $z=\varphi(x)$ is a local change of coordinates we define the transformed operator of $\mc{L}$ as
\begin{equation}\label{eq:varphioperator}
 \mc{L}^{\varphi}=\partial_x+\varphi'(x)\big(\sum_{i}\psi_i(\varphi(x))f_i+b(\varphi(x)) \big) \;.
\end{equation}
thus considering $\tope(K_D)$ as a space of meromorphic connections on the trivial bundle $D \times \mf{g} \to D$.
We note that if $\widetilde{\mc{L}}=\exp{n(z)}.\mc{L}$ then $ (\widetilde{\mc{L}})^{\varphi}=\exp{n\big(\varphi(x)\big)}.
\mc{L}^{\varphi}$, which implies that the transformation law is compatible with quotienting by the Gauge groups.

Hence, one can define
a sheaf of (meromorphic) opers $\Ope(\Sigma)$ on
the Riemann surface $\Sigma$ as follows. For $A$ a set, let $\{U_{\alpha}\}_{\alpha\in A}$ be an
open covering of charts in $\Sigma$, with transition functions $\varphi_{\alpha,\beta}$ whenever
$U_\alpha\cap U_\beta\neq \emptyset$, and let $[\mc{L}_{\alpha}] \in \Ope(U_{\alpha})=\tope(K_{U_\alpha})/\mc{B}(K_{U_\alpha})$ a collection of local sections of opers.
An oper on $\Sigma$, namely an element on $\Ope(\Sigma)$, is then defined as $\{[\mc{L}_{\alpha}], \alpha \in A\}$, with the additional requirement that on each non-empty intersection $U_{\alpha}\cap U_{\beta}$ we have that $[\mc{L}_{\alpha}]=[\mc{L}_{\beta}^{\varphi_{\alpha,\beta}}]\in\Ope(U_{\alpha}\cap U_{\beta})$, where $\mc{L}_{\beta}^{\varphi_{\alpha,\beta}}$ is given by  formula \eqref{eq:varphioperator}, with $\mc{L}=\mc{L}_{\beta}$ and $\varphi=\varphi_{\alpha,\beta}$. 
\begin{remark}
For a given $\mc{L} \in \ope(K_D)$, in general $\mc{L}^{\varphi}$ belongs to $\tope(K_{\varphi^{-1}(D)})$ but not to
$\ope(K_{\varphi^{-1}(D)})$. It is convenient to define, for any $\varphi$, an element $\widetilde{\mc{L}} \in \ope(K_{\varphi^{-1}(D)})$
equivalent to $\mc{L}^{\varphi}$. We make the following choice:
 \beq\label{eq:opercoordchange}
\widetilde{\mc{L}}= \varphi'(x)^{\rho^\vee}\mc{L}^{\varphi}=\partial_x+f-\frac{\varphi^{\prime\prime}(x)}
{\varphi^\prime(x)}\rho^\vee+\sum_{i=0}^{h^\vee-1}\left(\varphi^\prime(x)\right)^{i+1} b^i(\varphi(x)) \in
\ope(K_{\varphi^{-1}(D)}) \; ,
\eeq
where we have decomposed $b(z)=\sum_{i=0}^{\hv-1}b^i(z)$ according to the principal
 gradation. Hence $[\mc{L}^{\varphi}]=[\widetilde{\mc{L}}] \in \Ope(K_{\varphi^{-1}(D)})$.
\end{remark}

In the present work, we deal with meromorphic opers on the sphere $\Ope(\bb{P}^1)$, whose space of global sections we characterise
here. We cover $\bb{P}^1$ by two charts $U_{0}, U_{\infty}$
with coordinates $z,x$ and transition function $z=\frac{1}{x}$.
Suppose that we are given an operator $\partial_z+f+b(z)$ in $\ope(U_0)$ and one operator
$\partial_x+f+\tilde{b}(x)$ in $\ope(U_{\infty})$. These are local sections of the same global oper
if and only if $\partial_x+f+\tilde{b}(x)$ is gauge equivalent to the following operator
 \beq\label{eq:operatinfty}
\mc{L}=\partial_x+f-\frac{2\rho^\vee}{x}+\sum_{i=0}^{h^\vee-1}\left(\frac{-1}{x^2}\right)^{i+1} b^i(\frac{1}{x}).
\eeq
Hence the operator $\partial_z+f+b(z)$, defined locally on $U_0$, can be extended to a global meromorphic oper on the sphere if
and only if
$b(z)$ admits a meromorphic continuation at infinity, i.e. $b(z)$ is a rational function.
From this, it follows immediately that the space of global sections of meromorphic opers on the Riemann sphere $\Ope(K_{\bb{P}^1})$ is
isomorphic to $\mf{s}(K_{\bb{P}^1})$: an oper on the sphere is defined by the choice of a transversal space
and of $n$ arbitrary rational functions.

\vspace{20pt}
\section{Singularities of opers}\label{sec:singularities}
In this section we address the theory of regular and irregular singularities for differential
operators in $\ope(K_D)$ as well as for opers in $[\mc{L}]\in\Ope(K_D)$.  This theory was already addressed in
the opers literature, see \cite{bedr02,fr07,frgr09} among others. Here we both review known facts and include
results  from the literature of singularities of connections, in particular from \cite{bava83,bolibruch}. We will always
point out below whenever our nomenclature deviates from the one commonly used in the opers literature.

Since we are both interested in single operators
and in equivalence classes, we need to distinguish properties
which are Gauge invariant and properties which are not. For example a singular point for an operator may be
a regular point for a Gauge equivalent one, because we allow singular (meromorphic) Gauge transformations.
Hence we start with the following
\begin{definition}\label{def:removable}
 We say that a pole $w$ of $b \in \mf{b}^+(K_D)$ is a removable singularity
 of the differential operator $\mc{L}=\partial_z+f+b\in \ope(K_D)$ if there exists $N \in \mc{N}(K_D)$ such that
 $N.\mc{L}$ is regular at $w$ \footnote{Other authors define a removable singularity a
 a regular singularity whose monodromy,
 in the adjoint representation, is trivial. However, in order to remove such a singularity
 one needs to consider meromorphic Gauge transformations which take values in the full adjoint group, see
 e.g. Proposition \ref{prop:bava} below}.
\end{definition}

The theory of singular points begins with a dichotomy, the distinction between regular and irregular singular point.
In order to define it, we need to introduce the concept of algebraic behaviour.
\begin{definition}
Let $\bb{D}$ be the punctured disc of centre $w$.
We say that a, possibly multivalued, function $f: \bb{D} \to \bb{C}^n, n \geq 0 $ has algebraic behaviour at $z=w$ if,
fixed a closed sector $S$ of opening less than $2 \pi$, the following estimate holds $|f(z)|=o(|z-w|^{\alpha})$
for some $\alpha \in \bb{R}$. 
\end{definition}
\begin{definition}
 A singularity  $w \in D$ of the operator $\mc{L}=\partial_z+f+b\in \ope(K_D)$ is called regular if the
 following property holds for every finite dimensional module $V$ of $\mf{g}$: every local solution
 $y: \bb{C} \to V$ of the linear equation $\mc{L}.y=0$ has algebraic behaviour at $w$. A singular point that
 is not regular is named irregular.
\end{definition}
The above definition is clearly Gauge invariant.  It is in practice a notoriously difficult task the one
of establishing whether the singularity of a connection is regular or not, see e.g. \cite[Chapter 5]{bava83}.
However, this problem can be easily solved for the class of operators belonging to $\tope(K_D)$,
as we show in Proposition \ref{thm:fuchsian} below. To this aim we start by introducing the
concept of slope of the singular point \cite{deligne69}, \cite{frgr09}
\footnote{For computational convenience, our slope is equal to the slope defined in \cite{frgr09} $+1$.}.
\begin{definition}\label{def:slope}
Let $\mc{L}=\partial_z+f+b\in \ope(K_D)$. Let $w$ be a singularity of $b\in\mf{b}_+(K_D)$,
and decompose $b=\sum_{i=0}^{h^\vee-1} b^{i}$ according to the principal gradation of $\mf{g}$, with $b^i\in\mf{g}^i(K_D)$.
Let $\bar{b}^i (z-w)^{-\delta_i}$, with $\bar{b}^i \in \mf{g}^i$ and $\delta_i \in \bb{Z}$, be the most singular term of $b^i$ in
the Laurent expansion at $z=w$. Denote 
\beq\label{slopebbar}
\mu=\max\lbrace 1,\max_{i}\frac{\delta_i}{i+1}\rbrace,\qquad \bar{b}=\sum_{\frac{\delta_i}{i+1}=\mu}\bar{b}_i.
\eeq
We call $\mu\in\bb{Q}$ the slope of the singularity $w$. The principal coefficient of the singularity is defined as
$f-\rho^\vee+\bar{b}$ if $\mu=1$, and as $f+\bar{b}$ if $\mu>1$.
\end{definition}
\begin{definition}\label{def:fuchsian}
A pole $w \in D$ of $b\in\mf{b}_+(K_D)$ is called a Fuchsian singular point of $\mc{L}=\partial_z+f+b\in \ope(K_D)$
if it has the slope $\mu=1$. Equivalently, $w$ is Fuchsian  if $(z-w)^{i+1}b^{i}(z)$ is analytic at $w$ for all $i$.
\end{definition}
\begin{definition}\label{rem:slopeatinfinfinity}
 What happens at $\infty$? Let $\bar{b}^i\in\mf{g}^i(K_D)$, and assume $b^i(z)=O(z^{\delta_i})\bar{b}^i$ as $z\to \infty$. Define
 $\mu_{\infty}=\max_{i}\frac{\delta_i}{i+1}$. Letting $x=\frac{1}{z}$, we
 may choose as local representative of $[\mc{L}]$ at $\infty$, the differential
 operator \eqref{eq:operatinfty}.
 The slope of the latter operator at $x=0$ (i.e. the slope of $\mc{L}$ at $\infty$) is
 $\max\lbrace1,2+\mu_{\infty}\rbrace$. Hence, we say that the singularity at $\infty$ is
 Fuchsian if and only if $\mu_{\infty}\leq-1$.
\end{definition}
\begin{remark}
The authors of \cite{bedr02} use a different nomenclature:
Equivalence classes of opers with a Fuchsian singularity, with respect to the action of Gauge transformations regular at $w$,
are called $(\leq1)$-singular opers. We prefer to use the name Fuchsian, because in the case of $\mf{sl}_n$ opers, the definition
coincides with the one of Fuchsian scalar ODEs, see Corollary \ref{cor:fuchs} below.
\end{remark}

The reason for the previous definition comes form the following observation.
Let $w$ be a singularity of $\mc{L}$, with slope $\mu$. Introduce a branch of $(z-w)^{\mu}$, and let $\widehat{K}_D$
be the finite extension of $K_D$ obtained by adjoining $(z-w)^{\mu}$. The Gauge transform $(z-w)^{\mu\ad \rho^\vee}
\in \mc{H}(\widehat{K}_D)$ has
the following action on $\mc{L}$:
\begin{equation}\label{zwaction}
 (z-w)^{\mu \ad\rho^\vee}\mc{L}=\partial_z-\frac{ \mu \rho^\vee}{z-w}+ \frac{f+\bar{b}}{(z-w)^{\mu}}+o((z-w)^{-\mu}),
\end{equation}
where $\bar{b}$ is given by \eqref{slopebbar}. If the singularity $w$ is Fuchsian (namely if $\mu=1$), then $\mc{L}$ is
locally Gauge equivalent to a differential operator with a
first order pole. Its associated connection is then Fuchsian (in the sense of connections) at $w$, hence the
singularity is regular. We can also establish a partial converse of the above statement in case the function $b$
takes values in any subset of $f+\mf{b}$ whose only nilpotent is $f$; this is proved in the lemma below together with other results
that will be used in the sequel.

\begin{lemma}\label{lem:onsingularities} 
\begin{enumerate}
 \item Let $\mc{L}=\partial_x+f+b \in \ope(K_D)$ and $w$ a pole of $b$.  If the singularity $w$ is Fuchsian then it is a
regular singularity.\label{part1sing}
\item Let $\mf{m} \subset \mf{b}_+$ a vector subspace of $\mf{b}_+$ that satisfies the following property:
$f+m$ with $m \in \mf{m}$ is nilpotent if and only if $m=0$. Let
$\mc{L} \in \ope(K_D)$ be of the form $\partial_z+f+m$, with $m \in \mf{m}(K_D)$ singular at $w$.
The singularity at $w$ is regular if and only if it is Fuchsian.
\item If $\mc{L},\mc{L}' \in \ope(K_D)$ are two Gauge equivalent operators with a Fuchsian singularity at $w$, then the principal
coefficient of $\mc{L}$ at $w$ is conjugated in $\mc{N}$ to the principal coefficient of $\mc{L}'$ at $w$.
\item Let $\mc{L} \in \ope(K_D)$, $w$ a pole of $b$, and $\rho:\mf{g} \to End(V)$ be a non-trivial irreducible
representation of $\mf{g}$ such that
all local solutions of the equation $\mc{L}\psi=0$ have algebraic growth. Then $w$ is a regular singularity.
\end{enumerate}
\end{lemma}
\begin{proof}
We can assume $w=0$.
\begin{enumerate}
\item
Due to \eqref{zwaction}, if $0$ is Fuchsian then $z^{\rho^\vee}\mc{L}$ has a simple pole at $0$.
Hence in every representations every solution has algebraic growth, hence
$z=0$ is regular.

\item Because of (1), we just need to prove that not-Fuchsian implies irregular.
Suppose then that $z=0$ is not Fuchsian, so that  $\mu>1$.  Due to \eqref{zwaction}, applying the gauge transform
$z^{\mu \ad \rho^\vee}$ to $\mc{L}$ then we get
\beq\label{connpf3}
\partial_z+z^{-\mu} \big( f+ \bar{m}\big)+ o(z^{-\mu}) \; .
\eeq
where $\bar{m}\in\mf{m}$ is non-zero since $\mu>1$ (cf. Definition \ref{def:slope}).
Since $\bar{m}\neq0$, by hypothesis on $\mf{m}$,  we have that the principal coefficient
$f+\bar{m}$ is not nilpotent.  It follows that, fixed the adjoint representation,
the operator \eqref{connpf3} has a singularity with Poincar\'e rank greater than $1$ and
with a not-nilpotent principal coefficient, hence the singularity is irregular. See e.g. \cite{wasowAs}.


\item The proof is deferred to Lemma \ref{lem:Nloc} (ii).

\item After \cite[Theorem 5.2]{bava83},  the operator $\mc{L}$ can be brought -- by means of a meromorphic Gauge transformation -- into one of the two following forms:
\begin{enumerate}[(i)]
 \item $\partial_z+\frac{A}{z}+O(z^{-1+\e})$, with $A\in \mf{g}$,
 \item $\partial_z+\frac{B}{z^{r}}+O(z^{-r+\e})$ with $r \in \bb{Q}$, $r>1$, and $B\in \mf{g}$ is not nilpotent.
\end{enumerate}
Let $V$ be a non-trivial $\mf{g}$-module. Assuming we are in case $(ii)$, then the matrix operator representing $\mc{L}$ in $V$ has a singularity at $0$ of order $r>1$ with a non nilpotent coefficient. It follows that in this case there exists at least one solution with non-algebraic behaviour. Then, all solutions (in any representation) are regular at $0$ if and only if $\mc{L}$ can be brought to the form $(i)$. But this implies that $0$ is a regular singularity.
\end{enumerate}
 \end{proof}

As shown in Lemma \eqref{lem:onsingularities}, the subspaces of $f+\mf{b}^+$ such that
$f$ is the only nilpotent play an important role in the study of regular singularities for operators in
$\ope(K_D)$. Clearly $f+\mf{h}$ is one example of such subspaces. Other examples are the transversal spaces $f+\mf{s}$:
\begin{proposition}[Kostant]\label{prop:konstant}
Let $f+\mf{s}$ be a transversal space. Then $f$ is the only nilpotent element in $f+\mf{s}$.
 \begin{proof}
 The proof of Kostant \cite{ko63} follows the steps:
Any transversal space is in bijection with regular orbits. The only nilpotent regular orbit is the principal
nilpotent orbit. Since $f$ is principal nilpotent, it is the only nilpotent element in the transversal space.
 \end{proof}

\end{proposition}

Combining the above lemma and proposition, we deduce that if an operator is in its canonical form then 
a singularity is regular if and only if it is Fuchsian.
\begin{proposition}\label{thm:fuchsian}
Fix a transversal space $f+\mf{s}$ and let $\mc{L} \in \ope(K_D)$ be in canonical form
$\mc{L}=\partial_z+ f+ s$, $s \in \mf{s}(K_D)$. A point $w \in \bb{C}$ is a regular singular point of
$\mc{L}$ if and only if it is a not-removable Fuchsian singular point.
\begin{proof}
 Because of Proposition \ref{prop:konstant}, $\mf{s}$ satisfies the hypothesis of Lemma \ref{lem:onsingularities}(2),
 hence for an oper in canonical form a singular point is regular if and only if is Fuchsian. Moreover, since
 the singular locus of an operator in canonical form is a subset of the singular locus of any operator Gauge equivalent to it,
 then a singularity of an operator in canonical form cannot be removed.
\end{proof}

\end{proposition}

The above proposition has a two immediate corollaries. The first is a characterisation of regular singularities
for $\mf{g}$ opers.
\begin{corollary}\label{cor:normals}
 Fix a transversal space $f+\mf{s}$, let $\mc{L} \in \ope(K_D)$ and $\mc{L}_{\mf{s}}$ be its canonical form.
 All regular points of $\mc{L}$ are regular points of $\mc{L}_{\mf{s}}$, and all
regular singular points of $\mc{L}$ are either regular points or not-removable
Fuchsian singular points of $\mc{L}_{\mf{s}}$.
\begin{proof}
Let $N\in\mc{N}(K_D)$ be the Gauge transformation mapping $\mc{L}$ to its canonical form, namely $\mc{L}_{\mf{s}}=N.\mc{L}$.
From Lemma \ref{lem:gaugemiura} it follows that singular locus of $N$ coincides
 with the singular locus of $\mc{L}$. Therefore if $w$ is a regular point of $\mc{L}$,
 it is also a regular point of $N$, hence of $\mc{L}_{\mf{s}}$. If else $w$ is a regular singular point of $\mc{L}$, then $w$ is either
 a regular point of $\mc{L}_{\mf{s}}$ or a regular singular point of $\mc{L}_{\mf{s}}$; in the latter case, by virtue of 
 Proposition \ref{thm:fuchsian}, $w$ is a Fuchsian not-removable singularity.
\end{proof}
\end{corollary}
%

Another consequence is an algebraic proof of a well-known Theorem due to L. Fuchs
\begin{corollary}[Fuchs]\label{cor:fuchs}
Consider the scalar differential equation $$y^{(n)}(z)+a_2(z)y^{(n-2)}+\dots+a_{n}(z)y(z)=0 \; .$$
The singular point $w$ is a regular singular point for the scalar equation if and only if $(z-w)^{k}a_k(z)$ is analytic at $z=w$.
\begin{proof}
Let $\mf{g}=A_{n-1}$, let $V=\bb{C}^n$ by the standard representation, and choose as  transversal space $\mf{s}$ the space of companion matrices. More precisely, $\mf{s}$ is the space of traceless matrices whose coefficients are all zero outside the first row. We can then choose a basis
$\lbrace s_1,\dots,s_{n-1} \rbrace$
of $\mf{s}$ such that the scalar equation can be written in  the matrix form $\mc{L}y=0$, where
$\mc{L}=\partial_z+f+\sum_k (-1)^ka_{k+1}(z) s_k$. Suppose that $w$ is a regular singular point, namely all
solutions have algebraic growth.
Then by Lemma \ref{lem:onsingularities}(4), $w$ is a regular singular point of the operator $\mc{L}$, and due to
Proposition \ref{thm:fuchsian} it follows that $w$ is a Fuchsian singularity if and only if $(z-w)^ka_k(z)=O(1)$, $\forall k$. Suppose now that $(z-w)^ka_k(z)=O(1)$, $\forall k$. Then by Proposition \ref{thm:fuchsian} $w$ is a Fuchsian singularity
of $\mc{L}$ hence by Lemma \ref{lem:onsingularities}(1) $w$ is a regular singularity.
\end{proof}

\end{corollary}

\vspace{20pt}
\section{Quantum $\widehat{\mf{g}}$-KdV opers}\label{sec:kdvopers}
In this rest of the paper, we develop the following program
\begin{enumerate}
 \item Following \cite{FF11,fh16}, we introduce a class of $\mf{g}$-opers
 \footnote{More precisely, of $\widehat{\mf{g}}-$opers, where $\widehat{\mf{g}}$ is the untwisted affine Kac-Moody algebra
 associated to the simply laced Lie algebra $\mf{g}$, see below}, for $\mf{g}$ simply laced,
as the largest class of opers which can provide solutions to the Bethe Ansatz equations. We call them
Quantum $\widehat{\mf{g}}$-KdV Opers.
\item We prove that these opers actually provide solutions of the Bethe Ansatz equations.
\item We characterise these opers explicitly by means of the solution of a fully determined system of algebraic equations.
\end{enumerate}
We recall that in the $\mf{g}=\mf{sl}_2$ case, the above program was addressed and solved in \cite{BLZ04} by Bazhanov,
Lukyanov, and Zamolodchikov.
In this Section, following the proposal of Feigin and Frenkel \cite[Section 5]{FF11} (see also \cite[Section 8]{fh16}), we introduce the Quantum $\widehat{\mf{g}}$-KdV opers in the case $\mf{g}$ is simply laced\footnote{The not simply laced case is, at the time of writing, not yet fully understood.}, and we give
to these opers a first characterisation, which will be used to fully comply with the above program.

\subsection{The ground state oper.}\label{sub:ground}
The Quantum KdV opers are a suitable modification of the simplest opers proposed \cite[Section 5]{FF11}, which we studied in our previous papers \cite{marava15,marava17} in collaboration with Daniele Valeri. These opers are expected to correspond to the ground state of
the model. Explicitly, they have the form
\beq\label{eq:L0x}
\mc{L}(x,E)=\partial_x+f+\frac{\ell}{x}+(x^{Mh^\vee}-E) e_\theta,
\eeq
for arbitrary $\ell \in \mf{h}$ and $M>0$. As observed in \cite{fh16}, after the change of variable
\beq\label{eq:xzchange}
z=\varphi(x)= \left(\frac{1-\hat{k}}{h^\vee}\right)^{h^\vee} x^{\frac{h^\vee}{1-\hat{k}}},
\eeq 
the operator \eqref{eq:L0x} is Gauge equivalent to 
\begin{align}\label{eq:L0}
\mc{L}_G(z,\la)&=\de_z+f+\frac{r}{z}+z^{1-h^\vee}\big( 1 +\lambda z^{-\hat{k}}\big)e_{\theta},
\end{align}
which is a form more convenient for the present work. In the above formula
$0<\hat{k}<1$, $\la \in \bb{C}$ and $r \in \mf{h}$ are defined by the relations
\begin{align}\label{eq:M}
\ell=\frac{h^\vee}{1-\hat{k}} (r-\rho^\vee)+\rho^\vee, \qquad
M=\frac{\hat{k}}{1-\hat{k}}, \qquad
E=-\left(\frac{1-\hat{k}}{h^\vee}\right)^{(1-\hat{k})h^\vee} \lambda \;.
\end{align}

In order to avoid any ambiguity in the definition of the Quantum KdV opers,
we fix a transversal space $\mf{s}$ and consider the canonical form of the ground-state oper.
\begin{proposition}\label{prop:LGs}
The canonical form $\mc{L}_{G,\mf{s}}$ of the ground state oper \eqref{eq:L0} is
 \begin{equation}\label{eq:LGs}
\mc{L}_{G,\mf{s}}(z,\la)=\partial_z+f+\sum_{i=1}^n\frac{\bar{r}^{d_i}}{z^{d_i+1}}+z^{-\hv+1}(1+  \la z^{-\hat{k}})e_{\theta} \; ,
\end{equation}
where $\bar{r}=\sum_i\bar{r}^{d_i}$, with $\bar{r}^{d_i}\in\mf{s}^{d_i}(K_{\bb{P}^1})$, is the unique element in $\mf{s}$ such that
the Lie algebra elements $f-\rho^\vee+r$ and $f-\rho^\vee+\bar{r}$ are conjugated. 
\end{proposition}
\begin{proof}
The term $z^{1-h^\vee}\big( 1 +\lambda z^{-\hat{k}}\big)e_{\theta}$ is invariant under unipotent Gauge transformation.
Hence, if
$\bar{\mc{L}}$ is the canonical form of $\partial_z+f+\frac{r}{z}$ then $\mc{L}_{G,\mf{s}}=\bar{\mc{L}}+z^{1-h^\vee}
\big( 1 +\lambda z^{-\hat{k}}\big)E_{\theta}$. The operator $\partial_z+f+\frac{r}{z}$ 
is regular in $\bb{C}^*$ and has (at most) Fuchsian singularities at $z=0,\infty$. Due to
 Proposition \ref{thm:fuchsian}, this implies that its canonical form is regular in $\bb{C}^*$ and has (at most)
Fuchsian singularities at $0,\infty$. Hence it will take the form
$\partial_z+f+\sum_{i=1}^n\frac{\bar{r}^{d_i}}{z^{d_i+1}} $ for some $\bar{r}^{d_i} \in \mf{s}^{d_i}$.
From Lemma \ref{lem:onsingularities}(2), the principal coefficients at
$0$ of an operator and of its canonical form are conjugated. Since the principal coefficient  at $0$ of 
\eqref{eq:L0}  is $f-\rho^\vee+r$ and that of \eqref{eq:LGs} is $f-\rho^\vee+\bar{r}$, we deduce the thesis. 
\end{proof}
We notice here, as it will be important in the next Section, that in any finite dimensional representation the element $f-\rho^\vee+\bar{r}$ has the same spectrum as $f-\rho^\vee+r$, which in turn has the same spectrum as $-\rho^\vee+r$.

\begin{remark} The operator \eqref{eq:L0} is not meromorphic on the Riemann sphere, because the term $\lambda z^{1-h^\vee-\hat{k}}e_{\theta}$ is multi-valued. However, the element $\lambda z^{1-h^\vee-\hat{k}}e_{\theta}$ is fixed by the action of the Gauge group $\mc{N}(K_{\bb{P}^1})$, so that it prefectly makes sense to study which properties of $\mc{L}_G(z,\la)$ are preserved under the action of the meromorphic Gauge groups. Schematically, we have:
$$
\ope(K_{\bb{P}^1})/\mc{N}(K_{\bb{P}^1})+\lambda z^{-\hat{k}}e_{\theta} \cong
\lbrace \ope(K_{\bb{P}^1})+\lambda z^{-\hat{k}}e_{\theta} \rbrace /\mc{N}(K_{\bb{P}^1}) \; .
$$
The above comment is consistent with the following fact \cite{fh16}. Recalling the affine Lie algebra $\hat{\mf{g}}$ introduced in Section
\ref{sec:algebras}, with $d=\lambda\partial_\lambda\in\hat{\mf{g}}$ be the corresponding derivation. Then the oper  $\mc{L}_G(z,\la)$ is Gauge equivalent, by means
of the affine Gauge transformation $z^{k d}$, to the affine  (i.e. $\hat{\mf{g}}$-valued) meromorphic oper
\begin{equation}
 \partial_z+\hat{f}+\frac{r+k d}{z}+ z^{-\hv+1 }e_{\theta},
\end{equation}
where $\hat{f}=\sum_{i\in\hat{I}}\hat{f}_i$ is the sum of the negative Chevalley generators of $\hat{\mf{g}}$.
In the language of \cite{lacroix18}, the term $\frac{k}{z}$ is the twist function of the quasi-canonical normal form.

\end{remark}

The construction of Bethe Ansatz solutions from $\mc{L}_G(z,\la)$ can be briefly summarised as follows. The oper
$\mc{L}_G(z,\la)$ has two singular points,
 $z=0$ and $z=\infty$. The point $z=0$ is a Fuchsian singularity with principal coefficient $f-\rho^\vee+r$.
The point $z=\infty$ is an irregular singularity,
with slope $1+\frac{1}{\hv}$ and principal term $f+e_{\theta}$. 
The \textit{connection problem} between the two
singular points is encoded in the $Q$ functions, which we prove to be solutions of
the Bethe Ansatz equation for the quantum-$\mf{g}$ KdV model:
\begin{equation}\label{eq:TBAsec4}
\prod_{j \in I} e^{-2 i \pi \beta_jC_{\ell j}} \frac{Q^{(j)}\Big(e^{i \pi C_{\ell j}}\la^*\Big)}{Q^{(j)}
\Big(e^{- i\pi C_{\ell j}}\la^*\Big)}=-1,\qquad  i\in I
\end{equation}
where $\lambda^\ast$ is a zero of $Q^{(i)}(\lambda)$. As it was recalled in the Introduction, the quantum $\widehat{\mf{g}}-$KdV model is specified by a choice of the vacuum parameter $p \in \mf{h}$ 
and by the central
charge $c$. These determine uniquely the phases $\beta_j$'s of the Bethe Ansatz equations, and the the order of growth $\mu$
(that is, the asymptotic growth)
of the solutions $Q^{(i)}$'s for large $\la$, see \cite{dorey07}. 
At the level of the oper \eqref{eq:LGs} the phases $\beta_j$'s turn out to be linear functions of the element $r \in \mf{h}$,
and the growth-order is $\frac{1}{\hat{k}\hv}$ \cite{dorey07,marava15}. Hence the residue at $0$ and the slope at $\infty$
fixes uniquely the quantum model.

The natural question is: can the oper $\mc{L}_G(z,\la)$ be modified in such a way that it still
provides solutions of the same Bethe Ansatz equations, possibly corresponding
to higher states of the same quantum model? The answer is yes, as we show in the sequel of the paper.
\subsection{Higher states. First considerations}
Without losing generality, the most general meromorphic deformation of the ground state oper can be written as
\begin{equation}\label{eq:KdVopers}
\mc{L}(z,\la)=\mc{L}_{G,\mf{s}}(z,\la)+s(z),
\end{equation}
where $s(z)$ is an, a priori, arbitrary element of $\mf{s}(K_{\bb{P}^1})$. We make four assumptions, equivalent to the
ones given in \cite[Section 5]{FF11} (see also \cite[Section 8.5]{fh16}), 
on the above opers and we show that when these conditions are met solutions of the Bethe Ansatz can be obtained.
We thus say that a Quantum $\widehat{\mf{g}}$-KdV oper is an oper of the form \eqref{eq:KdVopers}, which satisfies the following assumptions:
\begin{asu}\label{asu1}
The local structure of the solutions at $0$ does not depend on $s(z)$.
\end{asu}
\begin{asu}\label{asu2}
The local structure of the solutions at $\infty$ does not depend on $s(z)$. 
\end{asu}
\begin{asu}\label{asu3}
All additional singular points are regular and the corresponding principal coefficients are conjugated
to the element $f-\rho^\vee-\theta^\vee$.
\end{asu}
\begin{asu}\label{asu4}
All additional singular points have trivial monodromy for every $\la \in \bb{C}$.
\end{asu}
\begin{remark}
 These assumptions deserve a brief explanation.
The solutions of the Bethe Ansatz equations \eqref{eq:TBAsec4} are obtained from $\mc{L}_G$ by
considering the connection problem between an irregular singularity at $\infty$ and a regular singularity at $0$.
Moreover, as recalled above,
the phases $\beta_j$'s and the order of growth
of their solutions $Q^{(i)}$'s are fixed uniquely by the residue at $0$ and by the slope at $\infty$.
It follows from this that Assumptions 1 and 2 are necessary conditions
to obtain solutions of the same Bethe Ansatz equations by the methods developed in \cite{marava15}.

Concerning Assumption \ref{asu4}, if $s \neq 0 $, then $\mc{L}(z,\la)$ has additional singularities, and  the connection
problem from $0$ to $\infty$
is only well defined if these additional singularities have trivial monodromy. In fact, in case of non-trivial monodromy,
the connection problem depends on which path in the punctured $\bb{C}^*$ one chooses to connect $0$ to $\infty$.

We finally discuss Assumption \ref{asu3}. If we assume that the additional singularity is regular, then the triviality of the monodromy
(in any representation) implies that the principal coefficient must be conjugated to $f-\rho^\vee+h$, where $h$ belongs to the co-root
lattice of $\mf{g}$, see Proposition \ref{prop:bava5}.
According to \cite{FF11}, the choice $h=-\theta^\vee$ is, for generic $(r,\hat{k}$,
a necessary condition for having trivial monodromy for any $\la \in \bb{C}$, because if
$h$ is a co-root different from $\rho^\vee$ then the trivial monodromy condition
is, generically, an overdetermined system for the local coefficients the singularity.
In \cite{BLZ04}, following \cite{duistermaat86},
the condition $h=-\theta^\vee$ was shown to be strictly necessary in the case $\mf{g}=\mf{sl}_2$, and for
opers satisfying Assumptions \ref{asu1} and \ref{asu2}.

We remark that, for the sake of the ODE/IM correspondence, the existence of non-generic opers (i.e. with $h\neq -\theta^\vee$)
may be actually immaterial. In fact, if the ODE/IM correspondence holds true, such non-generic
opers, and Bethe Ansatz solutions attached to them, will not presumably correspond to a state of the generalised quantum KdV model.
The same remakr is valid also for the case of an additional monodromy-free irregular singularity; moreover,
we are not aware of any result in the literature about this case and we will not pursue this possibility here.
\end{remark}

We organize our analysis of Quantum KdV opers as follows.
In the remaining part of the present section we classify the canonical form of opers of type \eqref{eq:KdVopers}
satisfying Assumption \ref{asu1}, \ref{asu2} and \ref{asu3}. In Section \ref{sec:BA} we construct solutions of
the Bethe Ansatz equations when the fourth postulate is met. In Section \ref{sec:miura} we prove that the canonical form of
the quantum KdV opers is Gauge equivalent to a form where all regular singularities are first order pole.
The remaining sections of the paper are devoted to the analysis of Assumption \ref{asu4}, and thus to the complete
classification of the Quantum-KdV opers.

\subsection{The first three assumptions}
We provide a more rigorous description of Assumption \ref{asu1} and \ref{asu2} by means of the following definition.
\begin{definition} 
Let $\mc{L}$ be given by \eqref{eq:KdVopers}, for some $s\in \mf{s}(K_{\bb{P}^1})$.
 We say that $s$ is subdominant with respect to $\mc{L}_{G,\mf{s}}$ at $0$ (resp. at $\infty$) if the slope and the principal
coefficient of the singularities at $0$ (resp. at $\infty$) of $\mc{L}$ does not depend on $s$. 
\end{definition}

\begin{lemma}\label{lemma:conditions12} Let $s\in \mf{s}(K_{\bb{P}^1})$, and write it as $s=\sum_{i=1}^ns^{d_i}$, with $s^{d_i}\in  \mf{s}^{d_i}(K_{\bb{P}^1})$. Then $s$ is subdominant with respect to $\mc{L}_{G,\mf{s}}$ at $0$  if and only if
\beq\label{eq:sat0}
s^{d_i}(z)=O(z^{-d_i}),\qquad z\mapsto 0,  
\eeq
and it is subdominant with respect to $\mc{L}_{G,\mf{s}}$ at $\infty$ if and only if
\beq\label{eq:satinfty}
s^{d_i}(z)=O(z^{-d_i-1}),\qquad  z\mapsto \infty.
\eeq
\end{lemma}
\begin{proof}
The slope at $0$ of $\mc{L}_{G,\mf{s}}$ is $1$, thus $s$ is subdominant at $0$ if and only if
$ \lim_{z \to 0}z^{d_i+1} s^{d_i}(z)=0$ for all $i=1,\dots,n$. 
The slope at $\infty$ of $\mc{L}_G$ is $1+\frac{1}{\hv}$. Let $s^{d_i}(z)=O(z^{-c_i})$, as $z \to \infty$. Then (cf. Definition \ref{rem:slopeatinfinfinity}), 
$s(z)$ is subdominant at $\infty$
if and only if $\frac{c_i}{d_i+1}>1-\frac{1}{\hv}$, $\forall i$.
In other words $c_i>i+1-\frac{d_i+1}{\hv}$. Since $c_i \in \bb{N}$ and $0\leq d_i\leq \hv-1$, the latter inequality
is satisfied if and only if $c_i\geq d_i+1$.
\end{proof}

The rational functions $s^{d_i}(z)$ satisfying the conditions of the above lemma can be written using a partial fraction
decomposition.
\begin{lemma}\label{lem:partialfraction}
 Let $f$ be a rational function such that
 \begin{enumerate}[(i)]
  \item $z^i f(z)$ is regular at $z=0$\label{lem:partialfractioni}
  \item $z^{i+1}f(z)$ is regular at $z=\infty$\label{lem:partialfractionii}
  \item $f$ is regular in $\bb{C}\setminus\lbrace 0,\infty\rbrace$ except for a finite (possibly empty) set of points  $\{w_j, j \in J\}$, where $f$ has a pole of order $m(j)\geq1$.\label{lem:partialfractioniii}
 \end{enumerate}
Then there exist $x_l(j)\in\bb{C}$, with $j\in J$ and $0\leq l \leq m(j)-1$, such that
\begin{equation*}
 f(z) = z^{-i}\sum_{j \in J}\sum_{l=0}^{m(j)-1}\frac{x_l(j)}{(z-w_j)^{m(j)-l}} \; .
\end{equation*}
\begin{proof}
Let $g(z)=z^i f(z)$. By hypotheses \eqref{lem:partialfractioni},\eqref{lem:partialfractionii} $g$ has poles only at $w_j, j \in I$. Since $g(\infty)=0$,
we can represent $g$ as a simple partial fraction without polynomial terms: 
there exist $x_l(j)\in\bb{C}$, with $j \in J$ and  $0\leq l \leq m(j)-1$,  such that 
$g(z)=\sum_{j \in J}\sum_{l=0}^{m(j)-1}\frac{x_l(j)}{(z-w_j)^{m(j)-l}}$.
\end{proof}
\end{lemma}
As a corollary,
we can write explicitly the canonical form of an operator satisfying Assumptions \ref{asu1}, \ref{asu2} and \ref{asu3}.
\begin{proposition}\label{pro:quasinormal}
An operator $\mc{L}(z,\la)$ of the form \eqref{eq:KdVopers} satisfies Assumptions \ref{asu1},\ref{asu2},\ref{asu3}  if and only if  there exists a (possibly empty) arbitrary finite
collection of non-zero mutually distinct complex numbers $\lbrace w_j \rbrace_{j \in J}\subset \bb{C}^*$ and a collection  $ s_l^{d_i}(j)$ of arbitrary elements of $\mf{s}^{d_i}$, with $0 \leq l \leq d_i$ and  $j \in J$, such that
 \begin{align}\nonumber
 \mc{L}(z,\la)= & \partial_z+f+\sum_{i=1}^n\frac{\bar{r}^{d_i}}{z^{d_i+1}}+z^{-\hv+1}
 (1+  \la z^{-\hat{k}})e_{\theta} + \\ \label{eq:prenormalform}
& +\sum_{j \in J}\sum_{i=1}^{n} z^{-d_i}\sum_{l=0}^{d_i}\frac{s^{d_i}_l(j)}{(z-w_j)^{d_i+1-l}} ,
 \end{align}
 where
 \begin{itemize}
  \item $\bar{r}=\sum_i\bar{r}^{d_i}$ is the unique element in $\mf{s}$ such that
the Lie algebra elements $f-\rho^\vee+r$ and $f-\rho^\vee+\bar{r}$ are conjugated.
\item The element $\bar{s}=\sum_i s_0^{d_i}(j)$ is independent of $j\in J$, and it is the unique element in $\mf{s}$ such that $f-\rho^\vee-\theta^\vee$ and $f-\rho^\vee+\bar{s}$ are conjugated.
 \end{itemize}

\begin{proof} Part of formula \eqref{eq:prenormalform} was already obtained in Proposition \ref{prop:LGs}, when considering the canonical form of the ground state oper $\mc{L}_G(z,\lambda)$. Due to Lemma \ref{lem:partialfraction}, Assumptions \ref{asu1},\ref{asu2} are satisfied if and only if the function $s^{d_i}(z)$ is of the form 
\beq\label{sdiproof}
s^{d_i}(z)=z^{-d_i}\sum_{j \in J}\sum_{l=0}^{m_i(j)-1}\frac{s^{d_i}_l(j)}{(z-w_j)^{m_i(j)-l}},
\eeq
for some $m(j)\in\bb{N}$, and $s^{d_i}_l(j)\in \mf{s}^{d_i}$. For $j\in J$, the principal coefficient of
$w_j$ is given by $f-\rho^\vee+\bar{s}$, where $\bar{s}=\sum_i s_0^{d_i}(j)$. Assumption \ref{asu3} states that for every $j\in J$ the additional singularity $w_j$ has to be regular, and its principal coefficient $f-\rho^\vee+\bar{s}$ is conjugated to the element $f-\rho^\vee -\theta^\vee$. In particular, $\bar{s}$ in independent of $j\in J$. Due to Proposition \ref{thm:fuchsian},  a singular point $w$ for an oper in canonical form is regular if and only if it is Fuchsian, form which it follows that in \eqref{sdiproof} we have $m_i(j)=d_i+1$ for every $i,j$, proving the proposition. 
\end{proof}

\end{proposition}

\vspace{20pt}
\section{Constructing solutions to the Bethe Ansatz}\label{sec:BA}
In this section, adapting the techniques of \cite{marava15}, we construct solutions of the Bethe Ansatz equations as coefficients
of the central connection problem for opers $\mc{L}$ of type \eqref{eq:KdVopers} and satisfying Assumption
\ref{asu1}, \ref{asu2}, \ref{asu3} and \ref{asu4}. According to Proposition \ref{pro:quasinormal},
we restrict our analysis to the subset of operators of the form \eqref{eq:prenormalform} such that all
additional singularities $\lbrace w_j \rbrace_{j \in J}$ have trivial monodromy. The latter condition implies
that all solutions $\psi$ of the differential equation $\mc{L}\psi=0$ are meromorphic functions on the universal
cover of $\bb{C}^*$, whose (possible) singularities are pole singularities located at the lift of the points $w_j, j \in J$.

\begin{definition}
Let $\widehat{\bb{C}}$ be the universal cover of $\bb{C}^*$, minus the lift of the points $w_j,j \in J$.
If $V$ is a $\mf{g}$-module, and fixed $\la \in \bb{C}$, we consider solutions of
$\mc{L}(z,\lambda)\psi(z,\lambda)=0$ as analytic functions $\psi(\cdot,\lambda):\widehat{\bb{C}}\to V$.
\end{definition}
\begin{remark}
For sake of notation simplicity, we assume a branch cut on the negative real semi-axis and
use the coordinate $z$ of the base space for the first sheet
of the covering.
Whenever we write $f(e^{2\pi i}z)$ we mean that we evaluate the function $f$ on the second sheet. This corresponds
to the counter-clockwise analytic continuation of the function $f(z)$ along a simple Jordan curve encircling $z=0$.
\end{remark}
\begin{definition}
Let $O_\la$ denote the ring of entire functions of the variable $\la$.
If $V$ is a $\mf{g}$-module, we denote by $V(\la)$ the set of solutions of
$\mc{L}(z,\la)\psi(z,\la)=0$ which are entire functions of $\la$, i.e.
they are analytic functions $\psi:\widehat{\bb{C}} \times \bb{C}\to V$.
\end{definition}
\begin{lemma}
 $V(\la)$ is a free module over the ring $O_\la$, and its rank is the dimension of
$V$. That is $V(\la)=V \otimes_{\bb{C}}O_\la$.
\end{lemma}
\begin{proof}
In order to find an $O_\la$-basis of solutions it is sufficient to find a set  $\lbrace\psi_i(z,\la),i=1,\dots,\dim V\rbrace$
of elements in $V(\la)$ which is a $\bb{C}$-basis of solutions of $\mc{L}(z,\lambda)\psi(z,\lambda)=0$ for every fixed $\la$.
 Let then $\lbrace\psi_i,i=1,\dots,\dim V\rbrace$ be a basis of $V$. Fix a regular point $z_0$ and let $\psi_i(z,\la)$ be the solution
 of $\mc{L}\psi=0$ satisfying the Cauchy problem $\psi_i(z_0,\la)=\psi_i$ for all $\la \in \bb{C}$. The solutions
 $\psi_i(z,\la) \in V(\la)$ because the differential equation depends analytically on $\la$, and are -- by construction -- a basis of $V$ for each
 fixed $\la$.
\end{proof}
For $t\in\bb{R}$, we define the following twisted operator and twisted solution:
\begin{align}\label{eq:twistedoper}
 & \mc{L}^t(z,\la)=\mc{L}(e^{2i \pi t}z,e^{2i \pi t \hat{k}}\la) \\ \label{eq:twistedsolution}
 & \psi_{t}(z,\la)=e^{2 i \pi t \rho^\vee}\psi(e^{2\pi i t}z, e^{2 \pi i t \hat{k}}\la)
\end{align}
Applying the change of variable formula \eqref{eq:opercoordchange} to \eqref{eq:prenormalform}, we have that
\begin{equation}\label{eq:twistedoperformula}
 \mc{L}^t(z,\la)=\partial_z+f+\sum_i\frac{\bar{r}^{d_i}}{z^{d_i+1}}+e^{- 2\pi i t } z^{1-\hv}\big(1+\la z^{-\hat{k}} \big)e_{\theta}+
 \sum_{m=1}^n e^{2\pi i t (d_m+1)} s^{d_m}(e^{2\pi i t}z) \;,
\end{equation}
where 
\beq\label{eq:sdm}
s^{d_m}(z)= z^{-d_m}\sum_{j \in J}\sum_{l=0}^{d_m}\frac{s^{d_m}_l(j)}{(z-w_j)^{d_m+1-l}}.
\eeq
By the same formula, it is straightforward to see that
the twisted function $\psi_t(z,\la)$ is a solution of the twisted operator: $\mc{L}^t(z,\lambda)\psi_t(z,\lambda)=0$. Note that when $t=1$ then $\mc{L}^{1}(z,\la)=\mc{L}(z,\la)$, while in general  $\psi_1(z,\la)$ is not equal to $\psi(z,\la)$. We define on $V(\la)$ the following $O_\la$-linear operator, the \emph{monodromy operator}:
\begin{equation}\label{eq:monodromy}
M: V(\la) \to V(\la), \qquad  M(\psi)(z,\la)= e^{2i\pi\rho^\vee}\psi(e^{2\pi i}z,e^{2\pi i \hat{k} }\la).
\end{equation}

\begin{remark}
Since $\mc{L}(z,\la)$ is a multivalued function of $z$, the monodromy operator cannot be defined on solutions $\mc{L}(\psi(z))=0$
for fixed $\la$. 
\end{remark}

\begin{remark}
In the case of the ground state, we have
\begin{align*}
 \mc{L}_{G,\mf{s}}^t(z,\la)=\partial_z+f+\sum_i\frac{\bar{r}^{d_i}}{z^{d_i+1}}+
 e^{- 2\pi i t } z^{1-\hv}\big(1+\la z^{-\hat{k}} \big)e_\theta.
\end{align*}
Hence the the twist by $t$ is tantamount to a change
$e_{\theta}\to e^{- 2\pi i t } e_{\theta}$, which in turn can be interpreted as an automorphism of a Kac Moody algebra
with a different loop variable than $\la$. This is the point of view that we
used in our previous paper. We drop this interpretation in the present work, for it cannot be simply extended to the more general operators we are considering.
\end{remark}

\subsection{The $n$ fundamental modules}
Let $\omega_i$ be the $i-$th fundamental weight of the Lie algebra $\mf{g}$, and denote by $L(\omega_i)$ the $i-$th fundamental representation of $\mf{g}$. Let the function $p: I \to \bb{Z}/{2\bb{Z}}$ be defined inductively as follows:
$p(1)=0$, and $p(j)=p(i)+1$ whenever $C_{ij}<0$ (i.e. $p$ alternates on the Dynkin diagram). Note that the principal coefficient at $\infty$ of \eqref{eq:prenormalform} is given by $f+e_{\theta}$, where $f=\sum_{i\in I}f_i$ and  $f_i$, $i\in I$, are fixed (negative) Chevalley generators of $\mf{g}$. The highest root vector $e_{\theta}$ is defined up to a scalar multiple, and the spectrum of the principal coefficient $f+e_\theta$ in $L(\omega_i)$ -- hence the asymptotic behaviour of solutions $\mc{L}\psi=0$ -- depends on such a choice. We choose it according to the following Proposition.
\begin{proposition}\label{prop:maximaleigenvalue}\cite[Proposition 4.4]{marava15}
One can choose the element $e_{\theta}$ in such a way that for every $i \in I$ the linear operator representing $f+(-1)^{p(i)}e_{\theta}$ in $L(\omega_i)$ has a
unique eigenvalue $\lambda^{(i)}$ with maximal real part, which is furthermore real, positive, and simple. In fact, the array of eigenvalues $(\lambda^{(1)},\dots,\lambda^{(n)})$ can be characterised
as the Perron-Frobenius eigenvector of the incidence matrix $B=2 \mbbm{1}_n-C$ of the Dynkin diagram of $\mf{g}$:
\begin{equation}
 \sum_j B_{ij} \lambda^{(j)}=2 \cos(\frac{\pi}{\hv}) \lambda^{(i)},\quad  \lambda^{(1)}=1, \quad B_{ij}=2 \delta_{i,j}-C_{ij}.
\end{equation}
\end{proposition}
Note that due to the definition of $p(i)$,  from \eqref{eq:twistedoperformula} it follows that for $i\in I$ we have
\begin{align}\label{eq:lpi}
 \mc{L}^{\frac{p(i)}{2}}(z,\la)=&\partial_z+f+\sum_{m=1}^n\frac{\bar{r}^{d_m}}{z^{d_m+1}}+(-1)^{ p(i) } z^{1-\hv}\big(1+\la z^{-\hat{k}} \big)e_{\theta}\notag\\
 &+\sum_{m=1}^nz^{-d_m}\sum_{j \in J}\sum_{l=0}^{d_m}\frac{s^{d_m}_l(j)}{(z+(-1)^{p(i)+1}w_j)^{d_m+1-l}}.
\end{align}
\begin{definition}\label{Vlai}
Fixed $e_{\theta}$ as in Proposition \ref{prop:maximaleigenvalue} above, for each $i\in I$
we set $V^i(\la)$ to be the $O_\la$-module of solutions of the  differential equation
\begin{equation}
\mc{L}^{\frac{p(i)}{2}}(z,\la)\psi(z,\la)=0, \qquad \psi(z,\la):
\widehat{\bb{C}}\times \bb{C} \to L(\omega_i),
\end{equation}
where $\mc{L}^{\frac{p(i)}{2}}$ is given by \eqref{eq:lpi}.
\end{definition}

\subsection{The singularity at $0$}
In order to study the monodromy of solutions about $z=0$, we address the local behaviour
of solution in a neighbourhood of the singular point $z=0$.  Applying  formula \eqref{eq:prenormalform} to \eqref{eq:lpi}, we have that
\begin{equation}\label{eq:L0fuchsian}
 z^{\ad \rho^{\vee}} \mc{L}^{\frac{p(i)}{2}}=
 \partial_z+ \frac{f-\rho^{\vee} + \bar{r}}{z} +(-1)^{p(i)}e_{\theta}\la z^{-\hat{k}}+O(1),
\end{equation}
where $\bar{r}=\sum_i \bar{r}^{d_i}$. Hence, $z=0$ is a Fuchsian singular point, and the principal
coefficient $f-\rho^{\vee} + \bar{r}$ is independent on the sign $p(i)$.

We remark again that the singularity at $0$ is also a ramification point of the potential.
If $\hat{k}$ is rational, then the operator $\mc{L}$ can be made single-valued by a change of variable,
and the standard Frobenius method applies. This is the case we considered in \cite{marava15}.
If $\hat{k}$ is irrational then the operator cannot be made single-valued by a change of variables.
Therefore the standard theorems on Frobenius series
do not apply. We develop below an appropriate modification of the Frobenius method in the case
$\hat{k}$ is irrational and satisfies a genericity assumption that implies that no logarithms are
present in the local expansion
at $0$. Filling this gap, we also complete our previous works.

The local behaviour of the solutions at $0$ depends on the spectrum of
$ f-\rho^{\vee} +\bar{r}$ in the representation we are considering.
As we remarked earlier this element is conjugated to $f-\rho^\vee+r$, which has the
same spectrum as $r-\rho^\vee \in \mf{h}$. Recall the weight lattice $P$ introduced in \eqref{weightcoweight}.
If $P_{\omega_i}\subset P$ denotes the set of weights of the representation $L(\omega_i)$,
then the spectrum of $ f-\rho^{\vee} + \bar{r} $ in this representation is the set $\lbrace \omega(r-\rho^\vee),
\omega \in P_{\omega_i} \rbrace$. In order to proceed, we need to consider separately the case when $\hat{k}\in (0,1)$ is 
rational and the case when it is irrational. First, we have
\begin{definition} Let $i\in I$, and let $P_{\omega_i}\subset P$ be the weights of the fundamental
representation $L(\omega_i)$. Moreover, put $T=\{a=n+m\hat{k}\,:\, n,m\in\bb{Z}, (n,m)\neq (0,0)\}\subset \bb{C}$.

If $\hat{k}$ is  irrational then the pair  $(r,\hat{k})\in \mf{h}\times (0,1)$ is said to be generic if for every $i \in I$ and for every $\omega\in P_{\omega_i}$
the spectrum of the matrix representing $r-\rho^{\vee}-\omega(r-\rho^{\vee})$ is contained in $\bb{C}\setminus T$.

If $\hat{k}$ then $\hat{k}=\frac{p}{q}$, with $q>p \in \bb{N}$.
We say that the pair $(r,\hat{k})\in \mf{h}\times (0,1)$ is
generic if for any $\omega \in P_{\omega_i}$, the spectrum of the matrix $q (r-\rho^\vee)-q\omega(r-\rho^\vee)$ does not belong to $\bb{N}^*$
(the set of positive natural numbers).
\end{definition}

We have the following result.
%
%
\begin{proposition}\label{thm:frobeniusz}
Let $(r,\hat{k})\in \mf{h} \times(0,1)$ be a generic pair.
Let $\big(\omega(r-\rho^\vee),\chi_{\omega}\big) $ be
an eigenpair composed of an eigenvalue and a corresponding eigenvector of $f-\rho^\vee +\bar{r}$ in $L(\omega_i)$.
A unique solution $\chi_{\omega}(z,\la)$ in $V^i(\la)$ is determined by the expansion
\begin{equation}\label{eq:frobeniusz}
\chi_{\omega}(z,\la)=z^{- \rho^\vee}z^{-\omega(r-\rho^\vee)} F(z,\lambda z^{-\hat{k}} ),
\end{equation}
where $F(z,w)$ is an $L(\omega_i)-$valued analytic function in a neighbourhood of $(0,0)$ such that $\lim_{z \to 0}
F(z,\lambda z^{-\hat{k}})=\chi_{\omega}.$
\begin{proof}
 For the case of $\hat{k}$ irrational the proof is in the Appendix.
 
 The case of a rational $\hat{k}$ was proven in \cite[Section 5]{marava15}.
 We skecth here the proof. Applying the transformation \eqref{eq:varphioperator} with $z=\varphi(x)=x^{q}$
 to the operator $z^{\rho^\vee}\mc{L}^{\frac{p(i)}{2}}$ \eqref{eq:L0fuchsian}, one obtains
\begin{align*}
\widetilde{\mc{L}}=&\partial_x+q\frac{f- r-\rho^\vee)}{x}+b(x) \; .
\end{align*}
where $b \in \mf{b}_+(K_{\bb{P}^1})$ is regular at $0$. The latter differential operator admits
the convergent Frobenius solution
$$
\widetilde{\chi_{\omega}}(x,\la)=x^{-q \omega(r-\rho^\vee)}\big( \chi_{\omega} +\sum_{j \geq 1} a_j x^j\big)
$$
for any eigenpair $(\omega(r-\rho^\vee),\chi_{\omega})$
of $f+r-\rho^\vee)$ provided
the genericity condition is satisfied.
Hence the equation $\mc{L}^{\frac{p(i)}{2}}\psi=0$ admits the solution
$$
\chi_{\omega}(z,\la)= z^{-\rho^\vee}\widetilde{\chi_\omega}(z^{\frac{1}{q}},\la)=z^{-\rho^\vee}z^{-\omega(r-\rho^\vee)}
\big( \chi_{\omega} +\sum_{j \geq 1} a_j z^{\frac{j}q}\big) \; .
$$
Moreover, a closer inspection of the series $\chi_{\omega}+\sum_{j \geq 1} a_j z^{\frac{j}q}$
shows that it is of the required form; see \cite{marava15} for details.
\end{proof}

\end{proposition}
A direct computation shows that the solution $\chi_{\omega}(z,\la)$ of Proposition \ref{thm:frobeniusz}
is an eigenfunction of the monodromy operator \eqref{eq:monodromy}. It follows from this that if
$r+\rho^\vee -f $ is semi-simple (i.e. diagonalizable) then solutions of the form  \eqref{eq:frobeniusz}
are an eigenbasis of the monodromy operator in the module $V^i(\la)$, for $i\in I$. More precisely, we have: 
\begin{corollary}\label{cor:psiomega}
Let $(r,\hat{k})$ be a generic pair, and 
$\chi_{\omega}(\la,z)$ be the solution constructed in Proposition \ref{thm:frobeniusz}. Then
\begin{equation}\label{eq:chiinvariance}
M(\chi_{\omega})(z,\la)=  e^{2\pi i\rho^\vee}
\chi_{\omega}\big(e^{2 \pi i}z,e^{2 \pi i\hat{k}}\la\big)=e^{2i\pi \omega(\rho^\vee-r) } \chi_{\omega}\big(z,\la\big),
\end{equation}
which means that $\chi_{\omega}(z,\la)$ is an eigenfunction of the monodromy operator.

Let moreover $f-\rho^\vee+\bar{r}$ be semisimple, and 
$\lbrace \chi_{\omega}\rbrace_{\omega \in P_{\omega_i}}$ -- taking into consideration weight multiplicities --
be a basis of $L(\omega_i)$ made of eigenvectors of $f-\rho^\vee+\bar{r}$. Then
$\lbrace \chi_{\omega}(z,\la)\rbrace_{\omega \in P_{\omega_i}}$ is a $O_\la$-basis of $V^i(\la)$.
\end{corollary}
\begin{proof}
The first part of the Lemma is a direct consequence of the Proposition. If $f-\rho^\vee+\bar{r}$ is semisimple,
then it admits a basis of
eigenvectors $\chi_{\omega},\omega \in P_{\omega_i}$.
It follows that $\lbrace \chi_{\omega}(z,\la)\rbrace_{\omega \in P_{\omega_i}}$ is a $\bb{C}$ basis of solutions for each fixed $\la$,
and hence a $O_\la$-basis of $V^i(\la)$.
\end{proof}

We finally consider the transformation of solutions of type \eqref{eq:frobeniusz}
under Gauge transformations.
These results will be useul later,
to show that the $Q$-functions (which satisfy the Bethe Ansatz) are 
Gauge invariant.
In other words the $Q$-functions are properties of the opers, and not just of the single differential operators.

Let $y\in\mf{n}_+(K_{\bb{P}^1})$, so that $\exp(y)\in\mc{N}(K_{\bb{P}^1})$, and denote
$\bar{\mc{L}}^{\frac{p(i)}{2}}=\exp(\ad y).\mc{L}^{\frac{p(i)}{2}}$, and $\bar{\chi}_{\omega}=\exp(y).\chi_{\omega}$.
By construction we have $\bar{\mc{L}}^{\frac{p(i)}{2}}\bar{\chi}_\omega=0$, and since $y$ is meromorphic, applying
the monodromy operator \eqref{eq:monodromy} on $\bar{\chi}_{\omega}$ and using \eqref{eq:chiinvariance} we get 
$M(\bar{\chi}_{\omega})=e^{-2i\pi \omega(r-\rho^\vee) }\bar{\chi}_{\omega}$. Thus, solutions of type \eqref{eq:frobeniusz}
satisfy the relation
$$M(\exp(y).\chi_\omega)=\exp(y).M(\chi_\omega), \qquad y \in\mf{n}_+(K_{\bb{P}^1}).$$
In addition, if $\lbrace \chi_{\omega}(z,\la)\rbrace_{\omega \in P_{\omega_i}}$ is a $O_\la$-basis of solutions
for the equation $\mc{L}^{\frac{p(i)}{2}}\psi=0$ (namely, a basis for the $O_\la$-module $V(\lambda)^i$ introduced in
Definition \ref{Vlai}), then 
$\lbrace \exp(y).\chi_{\omega}(z,\la)\rbrace_{\omega \in P_{\omega_i}}$ is a $O_\la$-basis
of solutions for the equation $\bar{\mc{L}}^{\frac{p(i)}{2}}\psi=0$.

\subsection{The singularity at $\infty$}
Now we move to the analysis of the irregular singularity at $\infty$.
In order to compute the asymptotic behaviour of solutions around $\infty$, we define
the function $q(z,\la)$ as the truncated Puiseaux series that coincides with
$\big(z^{1-h^\vee}(1+\lambda z^{-\hat{k}})\big)^{\frac{1}{h^\vee}}$ up to a remainder $o(z^{-1})$:
\begin{equation}\label{eq:Qzla}
 q(z,\la)=z^{\frac{1}{h^\vee}-1} \big ( 1 + \sum_{l=1}^{\lfloor  \frac{1}{\hat{k}h^\vee} \rfloor}
 c_l \la^l z^{-l \hat{k}}\big) \;.
\end{equation}
Here $c_l$ are the coefficients of the McLaurin expansion of $(1-w)^{\frac{1}{h^\vee}}$
\footnote{We remark that in the physical literature \cite{BLZ04} the condition
$\frac{1}{\hat{k} h^\vee}<1$, corresponding to $q(z,\la)=z^{-1+\frac{1}{h^\vee}}$
is called the semiclassical region of parameters.}. If we apply the Gauge transformation $q(z,\la)^{-\ad \rho^{\vee}}$ to the operator \eqref{eq:lpi} we obtain
\begin{equation}\label{eq:L0AtInfinity}
q(z,\la)^{-\ad \rho^{\vee}}.\mc{L}^{\frac{p(i)}{2}}(z,\la)=\partial_z+ q(z,\la) \Lambda^i + O(z^{-1-\e})  \; ,
\end{equation}
where $\Lambda^i=f+(-1)^{p(i)}e_{\theta}$ is (the image in the evaluation representation of) the cyclic element of the Kac-Moody algebra $\hat{\mf{g}}$, and $\e$ a
positive real number.

The transformed operator \eqref{eq:L0AtInfinity} has Poincar\'e rank $\frac{1}{\hv}$ 
with semi-simple principal coefficient  $\Lambda^i$, and a pertubration which is integrable at $\infty$.
It follows that the dominant part of the
asymptotic expansion of the solutions near $\infty$ is fully characterised by the spectrum of the principal
coefficient \cite{wasowAs}. In particular, the subdominant behaviour as $z \to + \infty$ is dictated
by the eigenvector with maximal real part. Indeed, we
have the following proposition which is adapted from \cite[Theorem 3.4]{marava15}.

\begin{proposition}\label{thm:fundamentalsolution}
Let $\psi^{(i)} \in L(\omega_i)$ be an eigenvector of $f+(-1)^{p(i)}e_{\theta}$ with
the maximal eigenvalue $\lambda^{(i)}$, as defined in Proposition \ref{prop:maximaleigenvalue}.
\begin{enumerate}
\item For every $\la \in \bb{C}$ there exists a unique solution $\Psi^{(i)}(z,\la) $ such that
\begin{equation}\label{eq:fundamentalsolution}
 \Psi^{(i)}(z,\la)=e^{-\la^{(i)} S(z,\la)}q(z,\la)^{\rho^\vee} \left( \psi^{(i)} +o(1) \right),
\end{equation}
as $z \to +\infty$, where $ S(z,\la)=\int^z q(w,\la) dw$.
\item For any $\la \in \bb{C}$,
if $\psi(\cdot,\la): \bb{C} \to L(\omega_i)$ is a solution of $\mc{L}^{\frac{p(i)}{2}}\psi=0 $
then
 $\Psi^{(i)}(z,\la)=o\big(\psi(z,\la)\big)$ as $z\to \infty$ unless $\psi(z,\la)= C \Psi^{(i)}(z,\la)$ for a $C \in \bb{C}$.
 \item $\Psi^{(i)}(z,\la)\in V^i(\la)$, i.e. it depends analytically on $\la$.
\end{enumerate}
\end{proposition}
\begin{proof}
 It follows from \eqref{eq:L0AtInfinity} and Proposition \ref{prop:maximaleigenvalue}. The detailed proof
 can be found in \cite[Theorem 3.4]{marava15}.
\end{proof}
By means of the characterization given in Theorem
\ref{thm:fundamentalsolution}(2) above, we can define the subdominant solution $\Psi^{(i)}$ for any choice operator
Gauge equivalent to $\mc{L}^{\frac{p(i)}{2}}$.
Let $\exp(y)\in\mc{N}(K_{\bb{P}^1})$, and denote
$\bar{\mc{L}}^{\frac{p(i)}{2}}=\exp(\ad y).\mc{L}^{\frac{p(i)}{2}}$.
Then there is a unique (up to a scalar mutliple) solution $\overline{\Psi}^{(i)}(z,\la)$
of $\bar{\mc{L}}^{\frac{p(i)}{2}}\psi=0$, belonging to $V^{i}(\la)$,
and satisfying Theorem
\ref{thm:fundamentalsolution}(2). This is indeed $\exp( y) \Psi^{i}(z,\la)$.

\subsection{The $\Psi$-system}
The next, and main algebraic step, towards constructing solutions of the Bethe Ansatz equations is
the $\Psi$-system, derived in \cite{marava15}, which the reader should consult for all details.

Fix $i\in I$, and recall the definition of the incidence matrix $B=2\mbbm{1}_{n}-C$. Then
consider the $\mf{g}-$modules $\bigwedge^2 L(\omega_i)$ and
$\bigotimes_{j\in I}L(\omega_j)^{\otimes B_{ij}} $. These are, in general, not isomorphic.
However, they have the same highest weight $\eta_i=\sum_j B_{ij} \omega_j$ and if we denote by
$\psi_{\omega}$ a vector of $L(\omega_i)$ with weight $\omega \in P_{\omega_i}$, then
the highest weight vector of $\bigwedge^2 L(\omega_i)$ is
 $\psi_{\omega_i}\wedge \psi_{\omega_i-\alpha_i}$, while the highest weight vector of
 $\bigotimes_{j\in I}L(\omega_j)^{\otimes B_{ij}}$ is $\bigotimes_{j\in I}\psi_{\omega_j}^{\otimes B_{ij}}$.
Hence, for every $i\in I$, we have
a well-defined homomorphisms of representations
\begin{equation}\label{eq:m_i}
m_i: \bigwedge^2 L(\omega_i) \to
\bigotimes_{j\in I}L(\omega_j)^{\otimes B_{ij}} \; , \quad m_i(\psi_{\omega_i}\wedge \psi_{\omega_i-\alpha_i})=
\bigotimes_{j\in I}\psi_{\omega_j}^{\otimes B_{ij}} \; ,
\end{equation}
uniquely defined by requiring that it annihilates the (possibly trivial) submodule
$U_i \subset \bigwedge^2 L(\omega_i)$ such that
$\bigwedge^2L(\omega_i)=L(\eta_i)\oplus U_i$. 
\begin{proposition}\label{pro:Psisystem}
Let $\Psi^{(i)}(z,\la)$ be the sub-dominant solution defined in Proposition \ref{thm:fundamentalsolution},
and $\Psi^{(i)}_{\frac{1}{2}}(z,\la)$ the same solution, twisted according to formula \eqref{eq:twistedsolution}. We can choose a normalisation of the solutions $\Psi^{(i)}(z,\la)$'s in such a way that
the following set of identities -- known as $\Psi-$system -- holds true:
\begin{equation}\label{eq:Psisystem}
m_i\big(  \Psi_{-\frac{1}{2}}^{(i)}(z,\la) \wedge
\Psi_{\frac{1}{2}}^{(i)}(z,\la) \big) =\otimes_{j\in I} \Psi^{(j)}(z,\la)^{\otimes B_{ij}}\,,
\quad i\in I,
\end{equation}
where $m_i$ is the morphism of $\mf{g}$ modules defined in \eqref{eq:m_i}.
\end{proposition}
\begin{proof}
It follows from Proposition \ref{thm:fundamentalsolution}. See \cite[Theorem 3.6]{marava15}
for details.
\end{proof}

\subsection{The $\QQ$ system and the Bethe Ansatz}
We are now in the position of proving the $\QQ$ system,
which implies the Bethe Ansatz equations. We suppose that $r \in \mf{h}$ is generic with respect to $\hat{k}$
and $f+r-\rho^\vee$ is semisimple, so that,
after Corollary \ref{cor:psiomega}, the set $\lbrace \chi_{\omega}(z,\la)\rbrace_{\omega \in P_{\omega_i}}$ is a $O_\la$-basis
of $V^i(\la)$ (weight-multiplicity is considered). Therefore we have the following decomposition
\begin{equation}\label{eq:Psidecomposition}
 \Psi^{(i)}(z,\la)=\sum_{\omega \in P_{\omega_i}}Q_{\omega}(\la) \chi_{\omega}(z,\la), \qquad i\in I,
\end{equation}
where the coefficients $Q_{\omega}(\la)$ are entire functions of $\la$. Note that
$\omega_i$ has mutlipliicty $1$, as well as any weight of the form $\sigma(\omega_i)$
for all $\sigma\in\mc{W}$, where $\mc{W}$ is the
Weyl group of $\mf{g}$; in particular the weight $\omega_i-\alpha_i$ belongs to the 
$\mc{W}$ orbit of $\omega_i$.

We show that, as a direct consequence of the $\Psi$-system, they satisfy the $Q\widetilde{Q}$-system
and thus the Bethe Ansatz equations.

Since after a Gauge transformation $N=e^{\ad y},y \in \mf{n}^+(K_{\bb{P}^1})$, the solution $\Psi^{(i)}$
and the solutions $\chi_{\omega}$ transforms as vectors, namely $\Psi^{(i)} \to e^{y}\Psi^{(i)}$,
$\chi_{\omega}\to e^{y}\chi_{\omega}$, it follows immediately that the entire functions
$Q_\omega(\la)$, $\omega\in P_{\omega_i}, i \in I$ are invariant under Gauge transformations.
This shows that the solutions of the Bethe Ansatz equations we construct are not just a properties of the operator
$\mc{L}\in\ope(K_{\bb{P}^1})$, but of the oper $[\mc{L}]\in \Ope(\bb{P}^1)$.

\begin{theorem}\label{thm:QQ}
Suppose that  the pair $(r,\hat{k}) \in \mf{h}\times (0,1)$ is generic, and $f+r-\rho^\vee$ is semisimple, 
so that the decomposition \eqref{eq:Psidecomposition} holds. Fix an arbitrary element $\sigma$ of the
Weyl group $\mc{W}$ of $\mf{g}$, and for every $\ell\in I$ 
denote $Q_\sigma^{(\ell)}=Q_{\sigma(\omega_\ell)}$ and $\widetilde{Q}^{(\ell)}_\sigma=Q_{\sigma(\omega_\ell-\alpha_\ell)}$.

One can normalise the solutions 
$\chi_{\sigma(\omega_\ell)},\chi_{\sigma(\omega_\ell-\alpha_\ell)},\ell \in I$ so that 
the following identity -- known as $Q\widetilde{Q}$-system -- holds for every $\ell\in I$
 \begin{align}\label{eq:QQtilde}
 \prod_{j \in I}\left(Q^{(j)}_\sigma(\la)\right)^{B_{\ell j}}
& =e^{i \pi \theta_\ell}
Q^{(\ell)}_\sigma(e^{-\pi i \hat{k}}\la)\widetilde{Q}_\sigma^{(\ell)}(e^{\pi i \hat{k}}\la) \nonumber\\ 
& -e^{-i \pi \theta_\ell }
Q_\sigma^{(\ell)}(e^{\pi i \hat{k}}\la)\widetilde{Q}_\sigma^{(\ell)}(e^{-\pi i \hat{k}}\la)
\,,
\end{align}
where $\theta_\ell=\sigma(\alpha_\ell)(r-\rho^\vee)$.
\end{theorem}
\begin{proof}
It is a straightforward computation: plug the decomposition \eqref{eq:Psidecomposition} and the expansion
\eqref{eq:frobeniusz}
into the $\Psi$-system \eqref{eq:Psisystem}.
\end{proof}
\begin{remark} The $Q\widetilde{Q}-$system, first obtained in
\cite{marava15,marava17}, was shown in \cite{fh16} to be a universal system of relations in the
commutative Grothendieck ring $K_0(\mc{O})$ of the category $\mc{O}$ of representations of the
Borel subalgebra of the quantum affine algebra $U_q(\widehat{\mf{g}})$.
\end{remark}

The Bethe Ansatz equation is a straightforward corollary of the $Q\widetilde{Q}$ system.
\begin{corollary}\label{cor:ba} Let $(r,\hat{k})$ be a generic pair.
 Let us assume that
the functions $Q^{(i)}_\sigma(\la)$ and $\widetilde{Q}^{(i)}_\sigma(\la)$
do not have common zeros. For any zero $\la^*_i$ of $Q^{(i)}(\la)$, the following system
of identities -- known as $\mf{g}$-Bethe Ansatz -- holds
\begin{equation}\label{eq:Qsystemthm}
\prod_{j = 1}^{n} e^{-2 i \pi \beta_jC_{\ell j}} \frac{Q_\sigma^{(j)}\Big(e^{i \pi C_{\ell j}}\la_\ell^*\Big)}{Q_\sigma^{(j)}
\Big(e^{- i\pi C_{\ell j}}\la_\ell^*\Big)}=-1
\,,
\end{equation}
with $\beta_j=\sigma(\omega_j)(r-\rho^\vee)$.
\end{corollary}

\begin{remark}
 What happens when $(r,\hat{k})$ is a non-generic pair? In that case in general the monodromy operator is not diagonalizable. Hence
 we can define the $Q^{\ell}_{\sigma}$
 functions only for $\sigma$'s belonging to a proper subset of $ \mc{W}$ \cite{marava15,marava17}. The same phenomenon occurs
 also at the level of the quantum 
 KdV model: for some values of $(r,\hat{k})$ not all $Q$ functions can be defined. Hence the ODE/IM correspondence is expected to hold
 also for non-generic values of the parameters.
\end{remark}

\vspace{20pt}
\section{Extended Miura map for regular singularities} \label{sec:miura}
Due to Proposition \ref{pro:quasinormal}, a $\mf{g}$-oper $\mc{L}(z,\la)$ of type \eqref{eq:KdVopers} satisying Assumptions \ref{asu1},\ref{asu2},\ref{asu3} takes the form  \eqref{eq:prenormalform}.  We decompose $\mc{L}$ as 
$$\mc{L}(z,\lambda)=\mc{L}_\mf{s}(z,\la)+z^{-\hv+1}(1+  \la z^{-\hat{k}})e_{\theta},$$
where
\begin{align}\label{eq:sglob}
& \mc{L}_{\mf{s}}(z,\la)=\partial_z+f+\sum_{i=1}^n\frac{\bar{r}^{d_i}}{z^{d_i+1}}+\sum_{j \in J}\sum_{i=1}^{n}
 z^{-d_i}\sum_{l=0}^{d_i}\frac{s^{d_i}_l(j)}{(z-w_j)^{d_i+1-l}},
\end{align}
and the coefficients $\bar{r}^{d_i},s_l^{d_i}(j)\in \mf{s}^{d_i}$ satisfy the conditions of Proposition \ref{pro:quasinormal}. The operator \ref{eq:sglob} is regular outside the Fuchsian singularities
$0,\infty$, and $w_j$, $j\in J$.
In this section we address the question of whether the operator $\mc{L}_{\mf{s}}$ is Gauge equivalent to
an operator such that all additional singularities $w_j,j \in J$ are first order poles. In
Theorem \ref{thm:miuraKdV} below, we answer positively to this question, showing that $\mc{L}_{\mf{s}}$ is Gauge equivalent to
\begin{equation}\label{eq:1glob}
 \mc{L}_1=\de_z+f+\frac{r}{z}+\sum_{j \in J}\frac{1}{z-w_j}\sum_{i=0}^{h^\vee-1}
 \frac{X^i(j)}{z^i},
\end{equation}
where $r \in \mf{h}$ and  $X^i(j) \in \mf{g}^i$.  More precisely, we prove that $\mc{L}_{\mf{s}}$ is the canonical form of $\mc{L}_1$ and the induced map
from the space of operators \eqref{eq:1glob} to the space of operators \eqref{eq:sglob} is surjective,
and -- modulo a (dotted) Weyl group action -- injective.
Since $z^{-\hv+1}(1+  \la z^{-\hat{k}})e_{\theta}$ in invariant under $\mc{N}(K_{\bb{P}^1})$,
it follows that the Quantum KdV opers can be uniquely written as
\beq\label{eq:Lglob}
\mc{L}(z,\la)=\mc{L}_1(z,\lambda)+z^{-\hv+1}(1+  \la z^{-\hat{k}})e_{\theta}.
\eeq
The choice of the above form is motivated by the Bethe Ansatz equations. Indeed, according to Assumption \ref{asu4} in order to construct solutions of the Bethe Ansatz equations, we need to impose on
the additional singularities $w_j$ the trivial-monodromy conditions, which will result in a complete set of algebraic equations for
the coefficients of the operator. Even though the location of the poles is independent of the choice of the Gauge, all local coefficients of course do depend on this choice. For theoretical and practical reasons we have chosen to work in the Gauge where all
additional singularities are first order poles. Indeed, this Gauge does not
depend on the choice of a transversal space, and the computation of the trivial monodromy conditions turns out
to be much simpler. The computation of the monodromy at $z=w_j$ for opers of type \eqref{eq:Lglob} will be made
in Section \ref{sec:zeromonodromykdv}, where we will also show that the coefficients $X^i(j)\in\mf{g}^i$ actually take values in the (symplectic) vector space $\mf{t}\subset \mf{b}_+$, which can be described as the orthogonal complement (with respect to the Killing form) of the subspace $\Ker\ad {e_{-\theta}}$, where $e_{-\theta}$ is a lowest root vector of $\mf{g}$. 
\begin{remark}
Applying $z^{\rho^\vee}$ to $\mc{L}_1$, we obtain
 \begin{equation*}
 z^{\rho^\vee} \mc{L}_1=\de_z+\frac{f-\rho^\vee+r}{z}+\sum_{j \in J}\sum_{i=0}^{h^\vee-1}
 \frac{X^i(j)}{z-w_j} \; .
 \end{equation*}
The (connection asssociated to the) above operator is totally Fuchsian: it is meromorphic on the Riemann sphere and all its
singularities are first order poles. Hence, we can conclude that $\mc{L}_{\mf{s}}$ is Gauge equivalent
to a totally fuchsian operator. We remark that the analysis of the similar question, namely whether a
connection with only regular singularities
is Gauge equivalent to a connection with only simple poles (i.e. a Fuchsian connection), is
of primary importance in the theory of the Riemann-Hilbert problems and led to the negative solution of the
Hilbert's 21st problem, see
\cite{bolibruch}.
\end{remark}

\subsection{Local theory of a Fuchsian singularity}
The operator $\mc{L}_{\mf{s}}$ given by \eqref{eq:sglob} is fixed by the choice of the coefficients
\beq\label{eq:parsi}
s^{d_i}_l(j)\in\bb{C},\qquad 0\leq l\leq d_i,\,i \in I, 
\eeq
namely by the choice of the singular coefficients of the Laurent expansion at any $w_j$. Similarly,
the operator $\mc{L}_{\mf{1}}$ given by \eqref{eq:1glob} is fixed by the choice of 
\beq\label{eq:parxi}
X^{i}(j)\in\mf{g}^{i},\qquad i=0 \dots \hv-1.
\eeq
Since the canonical form of $\mc{L}_1$ is $\mc{L}_{\mf{s}}$, then the Gauge sending $\mc{L}_1$ to $\mc{L}_s$
induces a map from the parameters \eqref{eq:parxi} to the parameters \eqref{eq:parsi}. This map is the object of our study.
Our method of analysis is based on the reduction of the global problem to a simpler local one. This is the
problem of proving that an operator of the form
\begin{equation*}
  \mc{L}=\partial_x+f+\sum_{i=1}^{\hv-1}\sum_{l=0}^{i}\frac{X^i_l}{x^{i+1-l}},
 \end{equation*}
with given $X_l^i \in \mf{g}^{i}$, is Gauge equivalent to an operator with a first order pole
  \begin{equation*}
 \mc{L}=\partial_x+f+\sum_{i=0}^{\hv-1}\frac{X^i}{x},
 \end{equation*}
for some $X^i \in \mf{g}^i$. In order to do so, we describe the local structure of
both operators and Gauge transformations at a Fuchsian singular point. We first embed the space $\ope(K_D)$
into a Lie algebra and then proceed with the localization.
\begin{definition}
Let $D$ be a domain in $\bb{C}$.
We denote
$$\mf{g}'(K_D)=\mf{g}(K_D)\oplus \bb{C}\partial_z$$
the extension of $\mf{g}(K_D)$ by the element $\partial_z$, with the relation $[\partial_z,p(z)]=\frac{d p}{dz}$.
\end{definition}
It is clear that we have an injective map $\ope(K_D)\hookrightarrow \mf{g}'(K_D)$.

\begin{definition}\label{def:gradings}
Let $w \in \bb{C}$ and set $x=z-w$. We denote
\beq\label{def:locg}
\mf{g}'((x))=\mf{g}\otimes\bb{C}((x)) \oplus \bb{C}\partial_x
\eeq
the Lie algebra defined by the relations 
$$[g_1\otimes x^m,g_2\otimes x^p ]=[g_1,g_2]\otimes x^{m+p},\qquad [\partial_x,g \otimes x^{m}]= g \otimes mx^{m-1}$$ 
for $g_1,g_2,g\in\mf{g}$ and $m,p\in\bb{Z}$. 
\end{definition}
The interpret the Lie algebra \eqref{def:locg} as a localized version (at the point $x=z-w$) of
operators in $\mf{g}'(K_D)$, and in particular of operators in $\ope(K_D)$.
To make this statement more precise, we assign two different degrees for elements in $\mf{g}'((x))$: the principal degree, given by
$$\deg \partial_x=0,\qquad \deg g^i \otimes x^j=i, \quad g^i \in \mf{g}^i,$$
and the total degree, given by
$$\tdeg \partial_x=-1,\qquad \tdeg g^i \otimes x^j=i+j,\quad  g^i \in \mf{g}^i.$$
For every $k \in \bb{Z}$, we denote $\mf{g}^{\geq k}((x)) \subset \mf{g}'((x)) $ the
 subspace generated by elements with total degree
 greater than or equal to $k$.  We also define the localized Gauge groups as 
 \begin{align*}
\mc{N}_{loc}&=\lbrace \exp{y}:\, y \in \mf{n}_+\otimes \bb{C} ((x)) \rbrace,\\
\mc{N}_{loc}^{\geq0}&= \lbrace \exp{y}:\, y \in \mf{n}_+\otimes \bb{C} ((x)) \cap \mf{g}^{\geq0}((x)) \rbrace.
\end{align*}

For $w\in D$, the localization map
$$Loc_w: \mf{g}'(K_D) \to \mf{g}'((x)),$$
is defined by setting $Loc_w(\partial_z)=\partial_x$, and $Loc_w(g)$
to be the Laurent series at $w$ of $g\in\mf{g}(K_D)$.
Fixed $w \in D$ the above map is an injective morphism of Lie algebras. We denote $ Loc(g)$ as $g_w$, for $g\in\mf{g}'(K_D)$. By definition, if we localise $\mc{L} \in \ope(K_D)$ at a point $w$ we obtain an element $\mc{L}_w$ of $\mf{g}'((x))$,
and if we localise a Gauge transformation $Y \in \mc{N}(K_D)$ at a point $w$ we obtain an element $Y_w$ of $ \mc{N}_{loc}$.
Since the localisation map is a morphism then $(Y.\mc{L})_w=Y_w.\mc{L}_w$.
\begin{lemma}\label{lem:Nloc}
\begin{enumerate}
 \item Let $\mc{L} \in \ope(K_D)$. Then $w \in D$ is a Fuchsian singularity of $\mc{L}$ if and only if $(\mc{L})_w \in \mf{g}^{\geq-1}((x))$. 
If $w$ is Fuchsian and $\mc{L}_{\mf{s}}=Y.\mc{L}$ is the canonical form of $\mc{L}$, with $Y \in \mc{N}(K_D)$, then the localisation of $Y$ at $w$ belongs to $\mc{N}_{loc}^{\geq0}$.
\item If $\mc{L},\widehat{\mc{L}} \in \ope(K_D)$ are Gauge equivalent operators with a Fuchsian singularity at $w$, then the principal coefficients of the singularity are conjugated in $\mc{N}$.\label{deferredlemma}
\end{enumerate}
\end{lemma}
\begin{proof}
(1) Let $\mc{L}=\partial_z+f+b\in\ope(K_D)$. By definition of Fuchsian singularity, then $z=w$ is Fuchsian if and only if
$$(\mc{L})_w=\partial_x+f+\sum_{i\geq0,m\geq0} \frac{b^i_m}{x^{i+1-m}},$$
for some $b^i_m\in\mf{g}^i$. But since each summand has total degree $\geq-1$, this is precisely the condition that $(\mc{L})_w\in\mf{g}^{\geq -1}((x))$.   Now let  $w$  be a Fuchsian singularity for $\mc{L}$, and let $\mc{L}_s=Y\mc{L}$ be the canonical form of $\mc{L}$, for some $Y\in\mc{N}(K_D)$. We want to prove that $Y_w\in\mc{N}_{loc}^{\geq0}$. Since $w$ is Fuchsian for $\mc{L}$ then $(\mc{L})_w\in \mf{g}^{\geq-1}((x))$, and due to Corollary \ref{cor:normals} $w$ is Fuchsian  also for $\mc{L}_{\mf{s}}$, so that $(\mc{L}_{\mf{s}})_w \in \mf{g}^{\geq-1}((x))$. Note that by construction we have  $Y_w\mc{L}_w=(Y\mc{L})_w=(\mc{L}_{\mf{s}})_w$, from which we infer that $Y_w\mc{L}_w\in \mf{g}^{\geq-1}((x))$.  We prove that $Y_w\in\mc{N}_{loc}^{\geq0}$ by showing that if $\overline{\mc{L}} \in \mf{g}^{\geq-1}((x))$,  and $Y\notin \mc{N}_{loc}^{\geq0} $, then  $Y \overline{\mc{L}} \notin \mf{g}^{\geq-1}$.  Indeed, let $Y=\exp{y}$ with $y\in\mf{n}_+(K_D)$. Since $Y\notin \mc{N}_{loc}^{\geq0}$, there exists a maximal $k>0$ such that the projection of $y$ into the subspace of total degree $-k$ is non zero. Let then $\frac{y^i}{x^{i+k}}$, with $0\neq y^i \in \mf{g}^i$ be the term of $y$ of total degree $-k$ and of lowest principal degree $i$.
 Then the projection of $\exp{y}.\overline{\mc{L}} $ onto the subspace of total degree $-k-1$ is non trivial, as $[f,\frac{y^i}{x^{i+k}}]\neq 0$ is the unique term in $\exp{y}.\overline{\mc{L}}$ with total degree $-k-1$ and principal degree $i-1$. Hence, $Y \overline{\mc{L}} \notin \mf{g}^{\geq-1}$.
 
 (2) It is enough to prove the statement when $\widehat{\mc{L}}=\mc{L}_{\mf{s}}$ is the canonical form of $\mc{L}$. Let $Y\in\mc{N}(K_D)$ be the Gauge transformation such that $\mc{L}_{\mf{s}}=Y\mc{L}$.  Since $w$ is Fuchsian, then due to part (1) we have that $Y_w \in \mc{N}_{loc}^{\geq0}$, and a direct calculation shows that
  $x^{\ad \rho^\vee}Y_w$ is regular at $x=0$. In other words, $x^{\ad\rho^\vee}Y_w=\exp{(\sum_{k \geq 0}y_k x^k)}$, for some $y_k \in \mf{n}^+$.
Again using the fact that $w$ is Fuchsian, we obtain (cf. equation \eqref{zwaction}) that $x^{\ad\rho^\vee}\mc{L}_w=\partial_x+\frac{a}{x}+O(1)$ and $x^{\ad\rho^\vee}(\mc{L}_{\mf{s}})_w=\partial_x+\frac{b}{x}+O(1)$, where $a,b\in\mf{g}$ are the principal coefficients of $\mc{L}$ and $\mc{L}_{\mf{s}}$ respectively. Since $(\mc{L}_{\mf{s}})_w=Y_w\mc{L}_w$ and   $x^{\ad \rho^\vee}Y_w$ is regular at $x=0$, we obtain the relation $b=\exp{y_0}.a$.
\end{proof}
We now introduce three important classes of operators in $\mf{g}^{\geq -1}((x))$.
\begin{definition}
We say that $\mc{L}$ is a $\mf{g}$-Bessel (or simply a Bessel) operator if
 \begin{equation}\label{def:gbessel}
  \mc{L}=\partial_x+f+\sum_{i=0}^{\hv-1}\sum_{l=0}^{i}\frac{X^i_l}{x^{i+1-l}}, \qquad X_l^i \in \mf{g}^{i}.
 \end{equation}
Given a transversal space $\mf{s}$, we say that $\mc{L}$ is a $\mf{s}$-Bessel operator if
 \begin{equation}\label{def:sbessel}
  \mc{L}=\partial_x+f+\sum_{i=1}^{n}\sum_{l=0}^{d_i}\frac{s^{d_i}_l}{x^{d_i+1-l}},\qquad s_l^{d_i} \in \mf{s}^{d_i}.
 \end{equation}
We denote by $V=V_{\mf{s}}$ the affine vector space of $\mf{s}-$Bessel operators. We say that $\mc{L}$ is a 1-Bessel operator if
 \begin{equation}\label{def:1bessel}
 \mc{L}=\partial_x+f+
 \frac{X}{x},\qquad X\in\mf{b}_+.
 \end{equation}
We denote by $U$ the affine vector space of $1$-Bessel operators.
\end{definition}
In the case $\mf{g}=\mf{sl}_2$, Bessel operators coincide with the operators of the Bessel differential equation. As shown below, the canonical form of every Bessel
operator is an $\mf{s}$-Bessel operator.
\begin{lemma}\label{lem:simpletokostant}
 Any Bessel operator \eqref{def:gbessel} is Gauge equivalent to an $\mf{s}$-Bessel operator \eqref{def:sbessel}. The corresponding Gauge transformation belongs to the finite dimensional subgroup $\overline{\mc{N}}_{loc}\subset \mc{N}_{loc}^{\geq0}$
 generated by elements in $\mf{n}_+((x))$ without regular terms 
 $$
 \overline{\mc{N}}_{loc}=\lbrace \exp{y}\,: \, y =\sum_{i=1}^{\hv-1}\sum_{j=1}^{i}\frac{y^i_j}{x^j}, \, y^i_j \in \mf{g}^i
 \rbrace \; .
 $$
 \end{lemma}
 \begin{proof}
 A simple computation shows that the set of Bessel operators is invariant under the action of the group $\overline{\mc{N}}_{loc}$. We can then prove the Lemma using the same steps as in the proof of Proposition \ref{lem:gaugemiura}.  We factorize $Y\in\overline{\mc{N}}_{loc}$ as follows: $Y=Y_{\hv-2} \dots Y_1 Y_0$ where $Y_j= Y^{\hv-1}_j \dots Y^{j+1}_j$
 and $Y^i_j=\exp{\frac{y^{i}_j}{x^{i-j}}}$ for some $y^i_j \in \mf{g}^i$. The transformation
 $Y^i_j=\exp{\frac{y^{i}_j}{x^{i-j}}}$ is then defined by recursively imposing that,
  after its application, the terms of the resulting operator with  total degree $\leq j-1$ and principal degree $\leq i-1$ are in canonical form.
 \end{proof}
Due to the previous lemma, we have a well-defined map from Bessel operators
to $\mf{s}$-Bessel operators. We are interested in the restriction of this map to the class of 1-Bessel operators.  Note that once further restricted to the class of 1-Bessel operators of the form $\partial_x+f+X^0/x$, with $X^0\in\mf{g}^0=\mf{h},$ then this map should be thought as a local version, at a regular singular point, of the so-called 'Miura map'. Bessel operators will play a prominent role later in this section, to obtain a normal form for Quantum $\mf{g}-$KdV opers.

The space $U$ of 1-Bessel operators \eqref{def:1bessel} can be described by means of the graded affine space 
$$U=\bigoplus_{i=0}^{\hv-1} U_i,$$ 
where
$$U_0=\lbrace \partial_x+f+x^{-1} X^0\,|\, X^0 \in \mf{h} \rbrace,\qquad U_i=x^{-1}\mf{g}^i,\quad i>0.$$
Note that $\tdeg U_i=i-1$. In the sequel we will often identify $U_i$ with $\mf{g}^i$ and 
$\oplus_{i\geq 1}U_i$ with $\mf{n}^+$.  Similarly, the space $V$ of $\mf{s}$-Bessel operators \eqref{def:sbessel} can be written as
$$V=\bigoplus_{i=0}^{\hv-1} V_{i},$$ 
where
$$V_{0}=\lbrace \partial_x+f+\sum_{i=1}^n \frac{s^{d_i}}{x^{d_i+1}} \,|\, s^{d_i} \in \mf{s}^{d_i} \rbrace ,\qquad V_{i}=  \lbrace \sum_{d_j\geq i}\frac{s^{d_j}}{x^{d_j+1-i}}\,|\,  s^{d_j} \in \mf{s}^{d_j}\rbrace, \quad i>0.$$
Note that $\tdeg{V_{i}}=i-1$. 
\begin{lemma}
The space $U$ of 1-Bessel operators and the space $V$ of $\mf{s}$-Bessel operators have the same dimension. More precisely,  $\dim U=\dim V=\dim \mf{b}_+=(\frac{\hv}{2}+1)n$.
\end{lemma}
\begin{proof} It is clear that $U\simeq \mf{b}^+$, so in particular $\dim U=\dim\mf{b}_+=(\frac{\hv}{2}+1)n$. We prove by induction on $i$ that $U_i$ and $V_i$ have the same dimension. First, $U_0\cong\mf{h}$ and $V_0\cong \mf{s}$, so $\dim U_0=\dim \mf{h}=n=\dim\mf{s}=\dim V_0$. Then, it is clear from the definition of $V^i$ that $\dim V^{i}=\dim V^{i+1}+\dim \mf{s}^i$, where $\mf{s}^i=\mf{s}\cap \mf{g}^i$, and by definition of transversal space we have $\dim\mf{s}^i=\dim \mf{g}^{i}-\dim \mf{g}^{i+1}$. Since $U^i\simeq \mf{g}^i$, then we have $\dim V^{i}-\dim V^{i+1}=\dim \mf{s}^i=\dim \mf{g}^{i}-\dim\mf{g}^{i+1}=\dim U^{i}-\dim U^{i+1}$. Hence $\dim U^{i}=\dim V^i$ implies $\dim U^{i+1}=\dim V^{i+1}$.
\end{proof}

\begin{definition}
 We denote as
$\Phi: U  \to V$ be the map that associates to any $1-$Bessel operator its canonical form. 
We define $\Phi_i$ as the projection of $\Phi$ onto $V_i$, so that the decomposition $\Phi=\oplus_i \Phi_i$ holds true.
\end{definition}
\begin{remark}
Let $\mc{L}=\partial_x+f+x^{-1}X\in U$ be a $1-$Bessel operator, with $X\in\mf{b}_+$.
Let $X=\sum_{i\geq 0}X^i$ be the decomposition of $X$ according to the principal gradation,
with $X^i\in\mf{g}^i$. By abuse of notation, we write  $\Phi(X^0,\dots,X^{h^\vee-1})$ to denote $\Phi(\mc{L})$.
\end{remark}
After Lemma \ref{lem:simpletokostant}, the Gauge transformation $N$ mapping a $1-$Bessel operator to its canonical form
belongs to $\overline{\mc{N}}_{loc}$. In particular, $Y=\exp y$ with $y\in\mf{g}^{\geq0}$, that is a linear combination of terms of non-negative
total degree.
It follows form this that, for each $i$, the map $\Phi_i$ depends only on $\oplus_{j\leq i}U_i$. More precisely, we have
\begin{lemma} Let $\mc{L}=\partial_x+f+x^{-1}X\in U$, with $X\in\mf{b}_+$, be a 1-Bessel operator, and let $X=\sum_{i\geq0}X^i$, with $X^i\in\mf{g}^i$. The map
 $\Phi:U\to V$ which associates to $\mc{L}$ its canonical form $\mc{L}_{\mf{s}}$ admits the triangular decomposition:
\begin{equation}\label{eq:Phidecomposition}
 \Phi(X^0,\dots,X^{\hv-1})=\sum_{i=0}^{\hv-1} \Phi_{i}(X^0,\dots,X^i), \qquad \Phi_i: \bigoplus_{j\leq i}U_j \to V_i.
\end{equation}
In other words, the terms of total degree $i-1$ in $\mc{L}_{\mf{s}}$ depend on the terms of total degree $\leq i-1$ of $\mc{L}$ only.
\end{lemma}
\begin{proof}
We prove the following equivalent statement: if $\mc{L}=\mc{L}'$ up to terms of total order $\leq i-1$, then the canonical form of $\mc{L}$ coincides up to terms of total order $\leq i-1$ with the canonical form of $\mc{L}'$.
 
 Let $\mc{L}$ be a $1$-Bessel operator and $Y=Y_{\hv-2}...Y_0$ be the corresponding Gauge transformation -- factorized as in the proof of Lemma \ref{lem:simpletokostant} -- mapping $\mc{L}$ to its canonical form
 $\mc{L}_{\mf{s}}$.
 If $\mc{L}'$ coincides with $\mc{L}$ for terms of total degree $\leq i-1$, then the terms of total degree $\leq i-1$ of  $Y_i \dots Y_0 \mc{L}'$ and $Y_i \dots Y_0 \mc{L} $ are also the same. By construction,  $Y_i \dots Y_0 \mc{L} $ coincides with  its canonical form $\mc{L}_{\mf{s}}$ up to terms of total degree $\leq i-1$.
 It follows that $Y_i \dots Y_0 \mc{L}'$ is in canonical form except that for terms of total degree $\geq i$. Hence,
 the transformation $Y'$ mapping $\mc{L}'$ to its canonical form can be factorised as $Y'=Y'_{\hv-2}\dots Y'_{i+1}Y_i \dots Y_0$ where $Y'_j=\exp y_j$, with  $j \geq i+1$, and $\tdeg y_j=j$. It follows that
 $Y_i \dots Y_0 \mc{L}'$ coincides with the canonical form of $\mc{L}'$ up to terms of total degree $\leq i-1$.
\end{proof}

\subsection{The maps $\mathbf{\Phi_i, i\geq0}$.}
In order to study the properties (in particular the surjectivity) of the map $\Phi$,
we first study the map $\Phi_0$, and then the maps $\Phi_i, i>1$.
We begin with the following
\begin{lemma}\label{lem:rho+s}
Let $f+\mf{s}$ a transversal space and $h \in \mf{h}$ an element of the Cartan subalgebra.
Then $f+h+\mf{s}$ is a transversal space too.
\begin{proof}
Since $\mf{s}$ is transversal, every $b\in\mf{b}_+$ can be written as $b=[f,m]+s$, for some  $m\in\mf{n}_+$ and $s\in\mf{s}$.
We prove that there exists an $y \in \mf{n}^+$ such that  $b=[f,y]+h+s$. Since $\ad f_{|\mf{g}^1}: \mf{g}^1 \to \mf{h}$ is invertible,
we denote by $g\in \mf{g}^1$ the unique element such that
$[f,g]=h$. Setting $y=m-g$, then $b=[f,y]+h+s$.
\end{proof}

\end{lemma}

We next prove that the map
$\Phi_0:U_0 \cong \mf{h} \to V_0\cong\mf{s}$ is surjective and
it is invariant under the 
shifted Weyl group $\mc{W}$ action on $\mf{h}$. The latter is defined as
\begin{equation}\label{eq:Weyldot}
\sigma\cdot h= \sigma(h-\rho^\vee)+\rho^\vee,
\end{equation}
for $\sigma \in \mc{W}$ and $h \in \mf{h}$ \footnote{The reader should confront the map
$\Phi_0$ with the map $res$ considered in \cite[3.8.11-3.8.13]{bedr02}, and the commutative diagram
\cite[(3.3)]{frenkel05gau}}.
\begin{lemma}\label{lem:Phi0}
Let $\mc{L}_0=\partial_x+f+ x^{-1}X^0$ be a 1-Bessel operator, with $X^0\in\mf{h}$, and fix a trasversal subspace $\mf{s}$. Let $(s,y_0)\in  \mf{s}\times \mf{n}_+$ be the unique pair of elements such that $\exp(y_0).(f-\rho^\vee+X^0)=f-\rho^\vee+s$.
\begin{enumerate}
 \item The canonical form of $\mc{L}_0$ is 
 $$\mc{L}_{\mf{s}}=\partial_x+f+\sum_{i=1}^n\frac{s^{d_i}}{x^{d_i+1}},$$
where $s^{d_i}\in\mf{s}^{d_i}$  is the restriction of $s$ to $\mf{s}^{d_i}$.
\item  Let $y_0=\sum_i y_0^i$, with  $y_0^i\in \mf{g}^i$, and let  $Y_0=\exp\sum_{i}\frac{y_0^i}{x^i}\in \overline{\mc{N}}_{loc}$. Then $\mc{L}_{\mf{s}}=Y_0.\mc{L}_0$. In particular, $\tdeg \sum_{i}\frac{y_0^i}{x^i}=0$.
\item The map 
$$\Phi_0: U_0 \cong \mf{h} \to V_0,\qquad  \Phi_0(X^0)=s$$ 
is surjective.
\item $\Phi_0(h)=\Phi_0(h') $ if and only if there exists  $\sigma \in \mc{W}$ such that $\sigma\cdot h=h'$.
\end{enumerate}
\end{lemma}
\begin{proof}
A proof of (1) and (2) is already contained in the proof of Proposition \ref{prop:LGs}.
We give here another proof, more algebraic in nature. Fix $X^0 \in \mf{h}$ and consider the operator $\mc{L}=\partial_x+f+x^{-1}X^0$.
Since $f-\rho^\vee+\mf{s}$ is a transversal space then the map $\mc{N} \times\lbrace f-\rho^\vee+\mf{s} \rbrace 
\to f+\mf{b}^+ $, $(Y,s) \mapsto Y.(f-\rho^\vee+s) $ is an isomorphism of affine varieties
\cite{ko63}; in particular given $X^0$ there exists a unique  pair $(s,y_0) \in \mf{s} \times \mf{n}^+ $ such that $\exp{y_0}.(f-\rho^\vee+X^0)=f-\rho^\vee+s$.
Hence, 
$$\exp{y_0}.\big(\partial_x+ \frac{f-\rho^\vee+X^0}{x}\big)=\partial_x+ \frac{f-\rho^\vee+s}{x}.$$ 
From this,  it follows that the Gauge $Y_0=x^{-\ad\rho^\vee}\exp{y_0}=\exp{\sum_i \frac{y_0^i}{x^i}}$  maps $\partial_x+f+ x^{-1}X^0$ to
$\partial_x+f+\sum_{i} x^{-d_i-1}s^{d_i}$,
where $s^{d_i}$ is the projection of $s$ onto $\mf{s}^{d_i}$.

(3,4) Recall that a transversal space is in bijection with the regular $G$ orbits.
Hence (3) $\Phi_0$ is surjective if and only if every regular $G$ orbit intersects the affine space $f+\mf{h}$,
and (4)  $\Phi_0(h)=\Phi_0(h')$ if and only if $f-\rho^\vee +h$ and $f-\rho^\vee +h'$ belong to the same $G$ orbit. It is proved in \cite{ko63} that every regular $G$ intersects $f+\mf{h}$ and two elements $f+l,f+l'$, $l,l'\in \mf{h}$
belong to the same $G$ orbit if and only if $l$ and $l'$ belong to the same  $W$ orbit.
\end{proof}

We now turn our attention to the maps $\Phi_i$, $i\geq 1$. For each $i\geq1$, the map $\Phi_i$ is an affine function with respect to the variable $X^i$.
\begin{lemma}\label{lem:Phidecomposition}
 For $i\geq1$, the map $\Phi_i: \bigoplus_{j\leq i}U_i \to V_i$ has the following structure
 \beq\label{phidecomposition}
 \Phi_i(X^0,\dots,X^i)=A^{X_0}_i[X^i]+P_i(X^0,\dots,X^{i-1})
 \eeq
 where $A_i: \mf{h} \oplus \mf{g}^i \to V_i$ such that  $(X^0,X^i) \mapsto A_i^{X^0}[X^i]$ is linear in $X^i$.
 \end{lemma}
\begin{proof}
 Let $\mc{L}\in U$  be a 1-Bessel operator, of the form $\mc{L}=\partial_x+f+x^{-1}X$ with $X\in\mf{b}_+$. Let $X=\sum_{i\geq 0}X^i$, with $X^i\in\mf{g}^i$ and for $i\geq 0$ denote
\beq\label{defLi}
\mc{L}_i=\partial_x+f+\frac{1}{x}\sum_{j=0}^{i}X^j.
\eeq
By Lemma \ref{lem:simpletokostant}, in order to reconstruct $\Phi_i$ is is enough to find a
Gauge transformation $Y_i$ such that $\Pi_f\big(Y_i \mc{L}_i \big)$ -- the projection of $Y_i \mc{L}_i$ onto
the space $[f,\mf{n}_+]$ -- 
is a linear combination of terms of total degree  $\geq i$.
Indeed, this condition is satisfied if and only if the operator $Y_i\mc{L}_i$
is canonical up to total degree $i-1$, in which case the projection of $Y_i\mc{L}_i$ onto $V_i$ coincides --
by Lemma \ref{lem:simpletokostant} --  with the map $\Phi_i$.

For $i=0$, then $\mc{L}_0=\partial_x+f+x^{-1}X^0$ with $X^0\in\mf{h}$, and the Gauge $Y_0$ was obtained in Lemma \ref{lem:Phi0}. We now construct $Y_i$ recursively with respect to the total gradation as $Y_i=\exp{y_i}Y_{i-1}$, where 
\beq\label{eq:ndef}
y_i(x)=\sum_{j= i}^{\hv-1}\frac{y_i^{j}}{x^{j-i}}
\eeq
is an element of the loop algebra of total degree $i$. Notice that if we let $\mc{L}$ vary as a function of the variables $X^0,\dots X^{\hv-1}$, then  the transformation $Y_i$ is a function of $X^0,\dots X^i$ only. Since $\Phi_j(X^0,\dots,X^{j-1})\in V_i$, each map $\Phi_i$ admits the decomposition
\beq\label{pf:fii}
\Phi_j(X^0,\dots,X^{j-1})=\sum_{\ell\geq j}\frac{\Phi^\ell_j}{x^{\ell+1-j}},
\eeq
for some $\Phi^\ell_j\in\mf{s}^{\ell}$. Now assume that for $j\leq i-1$ the Gauge $Y_{j}$ and the maps $\Phi_j(X^0,\dots,X^j)$ have been obtained, so that by construction the projection $\Pi_f(Y_{i-1}\mc{L}_{i-1})$ contains elements of total degree $\geq i-1$ only.  
By construction, the operator $Y_{i-1} \mc{L}_{i-1}$ has the form
 $$
 Y_{i-1} \mc{L}_{i-1}=\partial_x+f+\sum_{j\leq i-1}\sum_{\ell\geq j}\frac{\Phi^\ell_{j}}{x^{\ell+1-j}}+
 \sum_{j\geq i-1} b_j(x) \; ,
 $$
 where  $b_j(x)=\sum_{l\geq 0} \frac{b^l_j}{x^{l-j}}$, with $b^j_l\in\mf{g}^j$, is a remainder term of
 of total degree $j\geq i-1$.  We now look for an element $y_i$ of the form \eqref{eq:ndef} such that $\Pi_f\big(\exp{y_i} Y_{i-1} \mc{L}_i \big)$ contains only elements of total degree $\geq i$, thus proving the induction step. Due to \eqref{defLi} we have
$$\exp y_i.Y_{i-1}\mc{L}_i=\exp y_i.Y_{i-1}(\mc{L}_{i-1}+x^{-1}X^i),$$
and since $y_i$ is of total degree $i$, it follows that the terms of total degree $\leq i-2$ are already in canonical form. This is equivalent to say that $\Pi_f\big(\exp{y_i} Y_{i-1} \mc{L}_i \big)$ contains elements of total degree $\geq i-1$ only. It remains to consider the terms in $\exp{y_i} Y_{i-1} \mc{L}_i$ of total degree equal to $i-1$. These are given by
$$b_{i-1}(x)+\frac{1}{x}Y_0.X^i+[y_i(x),\partial_x+f+\sum_{\ell\geq 0}\frac{\Phi_0^\ell}{x^{\ell+1}}],$$
and the required condition is obtained imposing that the above quantity belongs to $V_i$. Due to the definition of the map $\Phi$, this is equivalent to say that
$$\Phi_i(X^0,\dots,X^{i-1})=b_{i-1}(x)+\frac{1}{x}Y_0.X^i+[y_i(x),\partial_x+f+\sum_{\ell\geq 0}\frac{\Phi_0^\ell}{x^{\ell+1}}].$$
The above is a system of equations for the coefficients $\Phi_i^j\in\mf{g}^j$, related to $\Phi_i$ by \eqref{pf:fii} and $y^{j}_i\in\mf{g}^{j}$, related to $y_i(x)$ by \eqref{eq:ndef}. Applying the Gauge $x^{\ad\rho^\vee}$ on both sides, we obtain the following set of equations in $\mf{g}$:
$$\sum_{j\geq i}(\Phi^j_i+[f,y_i^{j+1}])=\sum_{j\geq 0} b^j_{i-1}+\exp y_0.X^i+\sum_{j\geq i}(j-i) y^{j}_i-[\sum_{j\geq i}y^{j}_i,\sum_{\ell\geq 0}\Phi_0^\ell],$$
where $y_0\in\mf{n}_+$ is the element obtained in Lemma \ref{lem:Phi0}. 
Decomposing according to the principal gradation and projecting onto the subspaces $[f,\mf{n}_+]$ and $\mf{s}$, we obtain the following system  for the elements $y_i^{j}$ and $\Phi^j_i$, with $j\geq i$:
\begin{subequations}\label{eq:systemPhii}
\begin{align} 
 &\Phi^{j}_i=\Pi_{\mf{s}}\big((j-i) y^{j}+b^{j}_{i-1}+\big(\exp y_0. X^i\big)^{j}+
 \sum_{m+l=j} [y^{m},\Phi^{l}_{0}]\big), \\
& [f,y^{j+1}]=\Pi_f \big((j-i) y^{j}+b^{j}_{i-1}+\big(\exp y_0.X^i\big)^{j}+
 \sum_{m+l=j} [y^{m},\Phi^{l}_0]\big), 
 \end{align}
 \end{subequations}
 where $\big(\exp y_0. X^i\big)^{j}$ denotes the projection of $\exp y_0. X^i$ onto $\mf{g}^j$.
The system has a unique solution since $(\ker \ad_f ) \cap \mf{n}^+ =0$. 
  
 We can now study how the map $\Phi_i$ depends on the variables $X_0,\dots,X_i$ when we let $\mc{L}$ vary, in order to prove the decomposition \eqref{phidecomposition}. By construction, the quantity $b^{j}_{i-1}$, depends on $X^0,\dots X^{i-1}$ only. In addition,  since $\exp y_0$ and $\Phi_0$ depend on $X^0$ only, then the quantity $\tilde{\Phi}^{j}_i:=\Phi^{j}_i -\Pi_{\mf{s}}(b^{j}_{i-1})$ depend exclusively on $X^0$ and  $X^i$. Moreover, it depends linearly on the variable $X^i$. Indeed, both $\tilde{\Phi}^{i}_i$ and $y^{i+1}_i$ are  linear in $X^i$, and at each subsequent steps  $\tilde{\Phi}^{j}_i$ and $y^{j}_i$ depend linearly on the previous $\tilde{\Phi}$'s and $y$'s. This proves the thesis.
 \end{proof}

 We now consider the behaviour of the map $\Phi(X^0,\dots,X^{h^\vee-1})$ for fixed values of the first entry $X^0\in\mf{h}$. 
 \begin{definition}\label{def:PhiX0}
  Fixed $X^0 \in \mf{h}$, we denote  
  $\Phi^{X^0}:\mf{n}^+ \to \bigoplus_{i\geq1} V_i,$
  the map
  $$(X^1,\dots X^{\hv-1})\mapsto\sum_{i\geq 1}\Phi_i(X^0,X^1,\dots X^{i}),$$
  so that the decomposition
  $$\Phi(X^0,\dots,X^{\hv-1})=\Phi_0(X^0)+\Phi^{X^0}(X^1,\dots,X^{\hv-1})$$\
  holds true.
 \end{definition}

\begin{proposition}\label{cor:surjinj}
\begin{enumerate}
 \item The map $\Phi^{X^0}$ is injective if and only if it is surjective. 
 \item If the map $\Phi^{X_0}:\mf{n}_+ \to \sum_{i\geq1}V_i$ fails to be surjective then
there exists an $i\geq2$ and a non-zero element $y \in \mf{g}^i$ such that
$[X_0,y]= y$.
 \item There exists an open and dense subset $\mc{A} \in \mf{h}$ such that the map $\Phi^{X^0}$
is surjective and injective for all $X^0 \in \mc{A}$
\end{enumerate}
\begin{proof}
 1) Due to the triangular decomposition \eqref{eq:Phidecomposition} and to Lemma \ref{lem:Phidecomposition}, the map $\Phi^{X^0}$ is surjective if and only if 
 $\Phi_i^{X_0}$ is surjective for all $i\geq1$, which is equivalent to the condition $\det A_i^{X_0} \neq0$ for all $i\geq1$,
 which is equivalent to the condition
 $\Phi_i$ is injective for all $i\geq1$, which is equivalent to the condition that $\Phi^{X^0}$ is injective.
 
 2)Let $U_{X^0}$ be the affine space $\partial_x+f+x^{-1}(X^0+Y)$ with $Y \in \mf{n^+}$.
If $\Phi^{X^0}$ is not surjective, then by part (1) it is not injective. Hence, there exist operators $\mc{L},\overline{\mc{L}} \in U_{X^0}$ and (due to  Lemma \ref{lem:simpletokostant}) Gauge transformations $M,\bar{M} \in \overline{\mc{N}}_{loc}$ such that $M\mc{L}=\bar{M}\overline{\mc{L}}$.
Thus, $Y=\bar{M}^{-1}M \in \overline{\mc{N}}_{loc}$ satisfies $Y\mc{L}=\overline{\mc{L}}$. Since $Y\in \overline{\mc{N}}_{loc}$,  we take it to be of the form $Y=\exp\sum_{i=1}^{\hv-1}\sum_{j=1}^i\frac{y_j^i}{x^j}$ with $y^i_j \in \mf{g}^i$.
Now let $I\geq1$ the minimal index $i$ such that $y^i_j\neq0$ for at least one $1\leq j\leq i$.
Then, a direct calculation shows that the only term of principal degree $I-1$ in $\mc{L}-\overline{\mc{L}}$ is given by 
\beq\label{qI}
q^I=\sum_{j}[f,\frac{y^I_j}{x^j}] \in\mf{g}^{I-1},
\eeq
Since $\mc{L},\overline{\mc{L}} \in U_{X^0}$, then they terms of total degree $<0$ coincide, from which it follows that necessarily $q^I$ is
of total degree greater than $0$. By looking at \eqref{qI} we thus obtain that $q^I=O(\frac{1}{x})$,
from which we deduce that $I\geq2$ and $y^{I}_j=0$ for all $j\geq2$.
Collecting the terms of principal degree $I$ in $\mc{L}-\overline{\mc{L}}$ according to
formula \eqref{eq:gaugeaction}, we get
$$ \frac{[y^I_1,X_0]}{x^{2}}+\frac{ y^I_1}{x^{2}}.$$
Since $Y\mc{L}=\overline{\mc{L}}$ belongs to $U_{X_0}$, the above term must vanish. Then, the non-zero element $y^I_1\in\mf{g}^I$, $I\geq 2$, satisfies $[X_0,y^I_1]=y^I_1$. 
 
 3) Due to (2), the set of $X^0$ such that $\Phi^{X^0}$ is not bijective has positive codimension.
 
\end{proof}
\end{proposition}

The following result, which is of crucial importance in our construction, is a straightforward corollary of the previous
Proposition.
\begin{corollary}\label{cor:theta0}
The map $\Phi^{-\theta^\vee}:\mf{n}_+ \to \sum_{i\geq1}V_i$ is surjective.
\begin{proof}
The spectrum of $\ad_{-\theta_0^\vee}$ restricted to $\mf{n}^+$ does not contain positive integers:
indeed it is $\lbrace 0,-1,-2\rbrace$
if $\mf{g}\neq A_1$, and $\lbrace 0,-2\rbrace$ otherwise.
\end{proof}
\end{corollary}

\begin{example}
 The case $\mf{g}=A_1$. In this case $\mf{n}^+=\mf{g}^{1}$. Since
 $\mf{g}^{j}$ with $j\geq 2$ is empty, then - by Proposition \ref{cor:surjinj}(2)- the map
 $\Phi^{X^0}$ is bijective for every $X^0 \in \mf{h}$.
\end{example}

\begin{example}
The case $\mf{g}=A_2$. 
We have that 
$\mf{n}^+=\mf{g}^1\oplus \mf{g}^2 $, with $\mf{g}^2=\bb{C} e_\theta$. Given $X^0\in\mf{h}$, we show that the  map $\Phi^{X^0}: \mf{n}^+\to V_1 \oplus V_2$ fails to be surjective if and only if $[X^0,e_\theta]=e_\theta$.  We deduce that, for this particular Lie algebra, the necessary condition described in Proposition \ref{cor:surjinj} (2), for $\Phi^{X^0}$ not being injective, is also sufficient. 
Consider the Gauge $Y=\exp{\frac{e_\theta}{x}}$. If  $\mc{L}=\partial_x+f+\frac{X^0+X}{x}$, with  $X \in \mf{n}_+$ is a 1-Bessel operator, then $\overline{\mc{L}}=Y.\mc{L}=\mc{L}+\frac{ [f,e_\theta]}{x}+\frac{e_\theta+[e_\theta,X^0]}{x^2}$. So, $\overline{\mc{L}}$ is a 1-Bessel operator (namely, it belongs to the domain of $\Phi_0$) if and only if $X^0$ satisfies $[X^0,e_\theta]=e_\theta$. If this is the case, then $\overline{\mc{L}}\neq \mc{L}$ but $\Phi^{X^0}(\mc{L})=\Phi^{X^0}(\overline{\mc{L}})$. Hence
$\Phi^{X^0}$ is not injective, nor surjective.
\end{example}

\subsection{Extended Miura map for Quantum-KdV opers}
Building on the theory of Bessel operators described above, we address here the main topic of the present section. The problem is to find a Gauge transformation mapping the oper
\begin{equation}\label{eq:slocthm}
 \mc{L}_{\mf{s}}=\partial_z+f+\sum_{i=1}^n \frac{\bar{r}^{d_i}}{z^{d_i+1}}+\sum_{j \in J}\sum_{i=1}^{n}
 z^{-d_i}\sum_{l=0}^{d_i}\frac{s^{d_i}_l(j)}{(z-w_j)^{d_i+1-l}},
 \end{equation}
 with $\bar{r}^{d_i}, s^{d_i}(j) \in \mf{s}^{d_i}$,  to an operator of the form
\begin{equation}\label{eq:1locthm}
 \mc{L}_1=\de_z+f+\frac{r}{z}+\sum_{j \in J}\frac{1}{z-w_j}\sum_{i=0}^{h^\vee-1}
 \frac{X^i(j)}{z^i},
 \end{equation}
 with $r \in \mf{h}$ and $X^i(j) \in \mf{g}^i$.  As a first step, we prove that the canonical form of the operator \eqref{eq:1locthm} is of type \eqref{eq:slocthm}.
 \begin{lemma}\label{lem:tonormalform}
 For an arbitrary choice of its parameters, the canonical form of the operator $\mc{L}_1$ is an operator of the form \eqref{eq:slocthm}.
  Moreover, one has that $\sum_{i=1}^n \bar{r}^{d_i}=\Phi_0(r) $ and $\sum_i^n s^{d_i}_0(j)=\Phi_0(X^0(j)),$  for $j \in J$.
 \end{lemma}
  \begin{proof}
  The operator $\mc{L}_1$ satisfies the following property: it is regular outside $0,\infty$ and $w_j,j \in J$,
  and these points are at most fuchsian singularities. From Corollary \ref{cor:normals}, the canonical form of $\mc{L}_1$ satisfies the same property.  To prove the first part of the Lemma, it is then sufficient  to show that any operator in canonical form with such a property is an operator of the form \eqref{eq:slocthm}.

So let  $\mc{L}_{\mf{s}}=\partial_z+f+s$, for some $s\in \mf{s}(K_{\bb{P}^1})$.  From the definition of Fuchsian singularity it follows that if $0,\infty$ are Fuchsian then $s$ must satisfy  $s^{d_i}=O(z^{-d_i-1})$ as $z \to 0$, and $s^{d_i}=O(z^{-d_i-1})$ as $z \to \infty$. Now let 
$$\bar{r}^{d_i}=\lim_{z \to 0}z^{d_i+1}s^{d_i}(z),\qquad \hat{s}(z)=s(z)-\sum_{i=1}^n \frac{\bar{r}^{d_i}}{z^{d_i+1}}.$$
One has that $\hat{s}(z)=O(z^{-d_i})$ as $z \to 0$, and $\hat{s}(z)=O(z^{-d_i+1})$ as $z \to \infty$. Due to Lemma \ref{lem:partialfraction}, the function $\hat{s}$ admits the decomposition
 $$\hat{s}(z)=\sum_{j \in J}\sum_{i=1}^{n}z^{-d_i}\sum_{l=0}^{m_i(j)}\frac{s^{d_i}_l(j)}{(z-w_j)^{m_i(j)+1-l}},$$
 for some $m_i(j)$. Since  $w_j$ is Fuchsian for every $j \in J$, then $m_i(j)=d_i$. Hence, $\mc{L}_{\mf{s}}$ is of the form \eqref{eq:slocthm}.

   By Lemma \ref{lem:Nloc}(ii), the the principal coefficients of $\mc{L}_1$ and $\mc{L}_{\mf{s}}$
   at any singularity must be conjugated; this is equivalent -- see Lemma
   \ref{lem:Phi0} -- to the conditions $\sum_{i=1}^n \bar{r}^{d_i}=\Phi_0(r) $ and    $\sum_is^{d_i}_0(j)=\Phi_0(X^0(j))$. 
  \end{proof}

 Since the case $J$ is empty was already addressed in Proposition \ref{prop:LGs}, we suppose that $J=\lbrace 1,\dots,N\rbrace$
 for some $N \in \bb{Z}_+$. The space of $\mc{L}_1,\mc{L}_{\mf{s}}$ operators can be identified
 with the linear space of free coefficients in their defining formulas \eqref{eq:slocthm}, \eqref{eq:1locthm}.
 Because of the above Lemma, we have a (nonlinear) map between the two spaces, which satisfies the constraints  
 $\sum_{i=1}^n \bar{r}^{d_i}=\Phi_0(r) $ and $\sum_i^n s^{d_i}_0(j)=\Phi_0(X^0(j)), j=1,\dots,N$.
 
\begin{definition}
Let
$$U_N=\bigoplus_{j=1}^N\bigoplus_{i=1}^{\hv-1}\{ X^i(j) \in \mf{g}^i \},\qquad V_N=\bigoplus_{j=1}^N\bigoplus_{i=1}^{\hv-1}V_i(j),$$
where $V_i(j)=\mbox{span} \lbrace s^{d_l}_i(j) \in \mf{s}^{d_l}, i\leq d_l\leq \hv-1 \rbrace $. For every $(r,X^0(1),\dots,X^0(N)) \in \mf{h}^{\oplus N+1}$, we denote
\beq\label{eq:FrX}
F=F^{r,X^0(1),\dots,X^0(N)}:U_N \to V_N \; ,
\eeq
the map which associates to an operator \eqref{eq:slocthm}  its canonical form. 
\end{definition}
Recall from the local theory that, fixed the part of total degree $-1$ by the choice of an element $X^0 \in \mf{h}$,
there is a map $\Phi^{X^0}$ from the space of $1$-Bessel operators
to the space of the corresponding $\mf{s}$-Bessel operators. 
In the following theorem  we prove that the map \eqref{eq:FrX} is bijective
if and only if, for every $j=1,\dots,N$, $X^0(j)$  is such that the map $\Phi^{X^0(j)}$ is bijective.
Due to Proposition \ref{cor:surjinj}, the latter conditions are verified in an open and dense subset of the parameters $X^0(j) \in \mf{h}^{\oplus N}$. In particular they are verified, by Corollary \ref{cor:theta0},
when $X^0(j)=-\theta^\vee$ for all $j=1,\dots,N$, which is the case relevant for the Quantum $\mf{g}-$KdV opers.
 \begin{theorem}\label{thm:miuraKdV}
 The map $F^{r,X^0(1),\dots,X^0(N)}$ is bijective if and only if for every $j=1,\dots,N$
 the map $\Phi^{X^0(j)}$ is bijective.
\end{theorem}
\begin{proof}
Fix $m \in \lbrace 1,\dots,N\rbrace$, and consider the localisation at $w_m$ of the operators $\mc{L}_1$ of type \eqref{eq:slocthm} and $\mc{L}_{\mf{s}}$ of type \eqref{eq:1locthm}. We have
   \begin{align}
   & (\mc{L}_1)_{w_m}=\partial_x+f+\sum_{l,k\geq0}\frac{u^l_k}{x^{1-k}}, \qquad u^l_k \in \mf{g}^l \label{L1loc}\\ 
   & (\mc{L}_{\mf{s}})_{w_m}=\partial_x+f+\sum_{l,k\geq0}\frac{t^{d_l}_k}{x^{d_l+1-k}},\qquad t^{d_l}_k \in \mf{s}^{d_l}.\label{Lsloc}
\end{align}
The coefficients $u^l_k$ appearing in \eqref{L1loc} can be written in terms of the original variables 
of $\mc{L}_1$ as
\begin{align}\label{eq:X-to-t}
u^i_{0}=\frac{X^i(m)}{w_m^i}, \qquad u^{l}_{l-i}=\sum_{j \in J \setminus \lbrace m\rbrace }a_{i,l,j}X^{l}(j), \quad  l=0,\dots i-1.
\end{align}
for some complex coefficients $a_{i,l,k}$. We define $\widetilde{V_i}(m)$, $i=1,\dots, \hv-2$, as the subspace of coefficients $t^{d_l}_k$ appearing in \eqref{Lsloc} which have total degree $i-1$ and principal degree at least $i$. Namely, we have 
 $$\widetilde{V_i}(m)=\lbrace t^{d_l}_i, d_l\geq i \rbrace.$$
For each pair of indices $m,i$, the Gauge transformation $\mc{L}_1\to \mc{L}_{\mf{s}}=Y.\mc{L}_1$ from $\mc{L}_1$ to its canonical form $\mc{L}_{\mf{s}}$  induces  a map $\overline{F_i}(m): U_N \to \widetilde{V_i}(m)$, obtained by first localizing the image operator $Y.\mc{L}_1$ at $z=w_m$ and then restricting to the terms in  $\widetilde{V_i}(m)$. As we prove below, the map $\overline{F_i}(m)$ admits the decomposition
    \begin{equation}\label{eq:DecinTheorem}
   \overline{F_i}(m)=A_i^{X^0(m)}(\frac{X^i(m)}{w_m^i})+\overline{P}_{i,m}\;,
  \end{equation}
where for each $X^0\in\mf{h}$ the map $A_i^{X^0}: \mf{g}^i \to \widetilde{V}_i$,  is linear  and coincides with the map $A_i$ defined in Lemma \ref{lem:Phidecomposition}, while $\bar{P}_{i,m}$ is a function of the variables $X^{l}(m)$, with $l\leq i-1$, and  $m \in \lbrace 1, \dots, N \rbrace$. 

In order to prove the decomposition \eqref{eq:DecinTheorem}, we adapt the proof of Lemma \ref{lem:Phidecomposition} to the present case. Let $(Y)_{w_m}$ be the localization at $z=w_m$ of the Gauge transform $Y$ mapping $\mc{L}_1$ to $\mc{L}_{\mf{s}}$. We obtain $(Y)_{w_m}\in \mc{N}_{loc}^{\geq0}$ as the direct limit $\varinjlim Y_i$, where  $Y_i$ maps a truncation of $(\mc{L}_1)_{w_m}$ to its canonical form, up to terms of high (enough) total and principal degrees in such a way that the functions $\overline{F}_i(m)$, with $i=1,\dots, \hv-1$, are completely determined by the action of $Y_0,\dots,Y_{\hv-2}$ only.
   This is done as follows. Let $\mc{L}_i$ be the projection of $(\mc{L}_1)_{w_m}$ onto the subspace of total degree $\leq i-1$ and principal degree $\leq i$. Then from \eqref{L1loc} we get
   $$\mc{L}_i=\sum_{l=0}^{i}\sum_{k=0}^{l+i}\frac{u^l_m}{x^{1-k}}\; .$$
   Then we look for $Y_i$ such that $\Pi_f(Y_{i} \mc{L}_{i})$ is a linear combination of terms
   of total degree $>i-1$, and moreover the terms of total degree equal to $i$ have principal degree $\leq i$.
   
For $i=0$ we choose the Gauge transformation $Y_0=x^{-\rho^\vee}Y$, with $Y \in \mc{N}$, which maps $\partial_x+f+x^{-1}u^0_0$ to
   $\partial_x+f+\sum_{l=1}^n\frac{t^{d_l}_0}{x^{d_l+1}}$. We then reconstruct $Y_i$ recursively as follows. (1) We look for $y_{i}=\sum_{k\geq1}\frac{y^{i+k}}{x^{k}}$, with  $y^{i+k}\in \mf{g}^{i+k}$ such that the projection
   $$\Pi_f(\exp{y_{i}}Y_{i-1}. \mc{L}_{i})$$ only contains terms of total degree greater than $i$. Notice that $y_i$ is non-trivial only if $i\leq \hv-2$.  (2) We obtain $Y_{i}$ as $ \exp{y'_{i}}\exp{y_{i}}Y_{i-1}$ where
  $y'_{i}=\sum_{k=0}^{l}\frac{p^{k}}{x^{k-i-1}}$, for some $p^k \in \mf{g}^k$. 
  
  We implement (1) following the proof of Lemma \ref{lem:Phidecomposition}, and we obtain 
   a linear system for $y^{i+k}$ and thus for $\overline{F}_i(m)$. This coincides
  with system \eqref{eq:systemPhii} after we rename the variables
  $X^i \to y^i_0=\frac{X^i_0(m)}{w^m}$, $\Phi_i \to \overline{F}_i(m)$. In this system, the known terms $b$'s
  are shown, recursively, to depend on the coefficients
  $u^l_k, l\leq i-1,k\leq l $ of the local expansion of $\mc{L}_1$;
  after \eqref{eq:X-to-t}), these terms depend on $X^l(j)$, for $l\leq i-1$, and $j =1,\dots,N$.
  Hence the same reasoning as in  \ref{lem:Phidecomposition} after equation \eqref{eq:systemPhii} 
  proves the decomposition \eqref{eq:DecinTheorem}.

  We now prove that the decomposition \eqref{eq:DecinTheorem} implies the thesis. First we
  compute the coefficients $t^{d_l}_i$ spanning $\widetilde{V}_i(m)$ in terms of the
  original coeffcients of the oper $\mc{L}_{\mf{s}}$.  We have
   \begin{align}
   \label{eq:changeinTheorem}
   t^{d_l}_i=\frac{s^{d_l}_i(m)}{w_m^{d_l}}+\sum_{0\leq k \leq i-1}b_{i,l,k}s^{d_l}_{k}(m) \; ,
   \end{align}
   for some complex $b_{i,l,k}$. By above formulas, we have an invertible map
   $$\bigoplus_{m=1}^N \bigoplus_{i\geq 1} V_i(m) \to \bigoplus_{m=1}^N\bigoplus_{i\geq 1}\widetilde{V_i}(m),$$
   which associates to an oper of type \eqref{eq:slocthm}  the totaliy of local coefficients belonging to the spaces
   $\widetilde{V_i}(m)$'s. It follows that $F:U_N \to V_N$ is bijective if and only if the map $$\overline{F}=\sum_{m=1}^N \sum_{i=1}^{\hv-1}
   \overline{F}_i(m):U_N \to \bigoplus_{m=1}^N \bigoplus_{i=1}^{\hv-1}
   \widetilde{V_i}(m)$$ is bijective. Due to the decomposition \eqref{eq:DecinTheorem}, one deduces recursively that $\overline{F}$ is bijective if and only if the maps
  $\overline{F}_j(m)$
  are bijective for every $i$ (and $m$), if and only if the linear maps $A_i^{X^0(m)}:\mf{g}^i \to \widetilde{V}_i(m)$
  are bijective for every $m$ and $i$. After Lemma \ref{lem:Phidecomposition}(ii), the linear maps $A_i^{X^0(m)}:\mf{g}^i \to \widetilde{V}_i(m)$
  are bijective for every $m=1,\dots,N$ and every $i=1\dots \hv-1$ if and only if the maps $\Phi^{X^0(m)}$ are bijective for all $m=1\dots N$.
\end{proof}

 The main result of this section is a direct corollary of the previous theorem: any Quantum KdV oper can be uniquely written
 in the following form $\mc{L}_1+z^{-\hv+1}(1+  \la z^{-\hat{k}})e_{\theta}$ where $\mc{L}_1$ is as
 in \eqref{eq:1locthm} with $X^0(j)=-\theta^\vee$ for any $j \in J$.
 \begin{corollary}
 Fix $(r,\hat{k}) \in \mf{h}\times (0,1)$. Any operator $\mc{L}(z,\la)$
 satisfying Assumptions \ref{asu1}, \ref{asu2}, \ref{asu3} defining the Quantum $\mf{g}$-KdV opers, with a (possibly empty)
  set $\{w_j\,,\,j\in J\}$ of additional poles, is Gauge equivalent to a unique operator of the form
 \begin{align}\label{eq:ouroperators}
 \mc{L}(z,\la)=  \de_z+f+\frac{r}{z}+\sum_{j \in J}\frac{1}{z-w_j}\left(-\theta^\vee+\sum_{i=1}^{h^\vee-1}
 \frac{X^i(j)}{z^i} \right)  +z^{-\hv+1}
 (1+  \la z^{-\hat{k}})e_{\theta} \, ,   
 \end{align}
 for some $X^i(j) \in \mf{g}^i$ .
 \begin{proof}
 The case when $J$ is empty was proved in Lemma \ref{prop:LGs}. Here we suppose $J=\lbrace 1\dots N \rbrace$ with $N \in \bb{Z}_+$.
  After Proposition \ref{pro:quasinormal}, the Quantum $\mf{g}$-KdV opers satisfying the first Assumptions \ref{asu1}, \ref{asu2}, \ref{asu3} are of the form
 $\mc{L}_{\mf{s}} +z^{-\hv+1}
 (1+  \la z^{-\hat{k}})e_{\theta}$,
 where $\mc{L}_{\mf s}$ is the oper \eqref{eq:slocthm} such that $\bar{r}=\Phi_0(r)$ and
 $\sum_i^n s^{d_i}_0(j)=\Phi_0(-\theta^\vee)$, for $ j=1,\dots,N$.
 
Due to Corollary \ref{cor:theta0}, the map $\Phi^{-\theta^\vee}$ is bijective. It then follows from Theorem \ref{thm:miuraKdV},that the operator $\mc{L}_{\mf{s}}$ is gauge equivalent to an operator of the form \ref{eq:1locthm} with $X^0(j)=-\theta^\vee$ for $j=1,\dots,N$. Since the term $z^{-\hv+1}
 (1+  \la z^{-\hat{k}})e_{\theta}$ is Gauge invariant under $\mc{N}(K_{\bb{P}^1})$, the thesis follows.
 \end{proof}

 \end{corollary}
 We remark that the operator \eqref{eq:ouroperators} is Gauge equivalent to an operator where all regular singularities are simple poles. Indeed we have
 \begin{equation*}
 z^{\rho^\vee} \mc{L}_1=\de_z+\frac{f-\rho^\vee+r}{z}+\sum_{j \in J}\frac{-\theta^\vee+X(j)}{z-w_j} +
 (1+  \la z^{-\hat{k}})e_{\theta},
 \end{equation*}
 where $X(j)=\sum_{i=0}^{\hv-1}X^i(j)\in\mf{n}_+$ for $j=1,\dots,N$.

\vspace{20pt}
 \section{The gradation induced by the highest root}\label{sec:gradation}
In order to proceed with our program, we have to impose on the operators
\eqref{eq:ouroperators} the trivial monodromy conditions
at any additional singularity $w_j,j \in J$. These operators have locally, at any $w_j,j \in J$, the expansion
$$
\partial_x+\frac{-\theta^\vee+ \eta}{x}+O(1), \; \eta \in \mf{n}^+ \; ,
$$
where $\theta^\vee$ is the dual to the highest root of the Lie algebra, $\theta^\vee=\nu^{-1}(\theta)$.
As we will see in the next section, for an operator with a simple pole,
the monodromy is computed by decomposing its coefficients in the eigenspaces of the adjoint action of the \textit{residue} $-\theta^\vee+ \eta$. As a necessary preliminary tool, we therefore devote this section to the study
of the eigen-decomposition of $\mf{g}$ with respect to the adjoint action of $-\theta^\vee+\eta$ with $\eta \in \mf{n}^+$.
\subsection*{The gradation induced by $\theta$}
We need the following lemma, which can be found in \cite[Section 9]{hum72}.
\begin{lemma}\label{lemmahumphreys}
Let $\alpha,\beta$ be nonproportional roots. Then
\begin{enumerate}[(i)]
\item If $(\alpha | \beta)>0$ then $\alpha-\beta$ is a root. If $(\alpha | \beta)<0$ then $\alpha+\beta$ is a root.
\item If $\mf{g}$ is simply-laced then $(\alpha | \beta)\in\{-1,0,1\}$.
\end{enumerate}
\end{lemma}
Let $\theta\in\Delta$ be the highest root of $\mf{g}$, and denote $\theta^\vee=\nu^{-1}(\theta)\in\mf{h}$. We want to study the spectrum of $\ad\theta^\vee$ in the adjoint representation. It is clear that for every $\alpha\in\Delta$, if $x\in\mf{g}$ belongs to the root space of $\alpha$ then we have
$$[\theta^\vee,x]=(\theta | \alpha) x.$$
Now we can apply Lemma \ref{lemmahumphreys} (ii) to the case when one of the two roots is $\theta$, the highest root. The only roots proportional to $\theta$ are $\pm \theta$, and we have $(\theta | \pm\theta)=\pm 2$. Due to the lemma, then $(\theta | \beta)\in\{-1,0,1\}$ for every $\beta\in\Delta\setminus\{-\theta,\theta\}$. The spectrum of $\ad \theta^\vee$ in the adjoint representation is then given by
\beq\label{specA-1}
\sigma(\theta^\vee)=
\begin{cases}
\left\{-2,-1,0,1,2\right\}\qquad & \text{ if } \mf{g}\neq \mf{sl}_2\\
\left\{-2,0,2\right\}\qquad & \text{ if } \mf{g}= \mf{sl}_2,
\end{cases}
\eeq
and we obtain a $\bb{Z}-$gradation of $\mf{g}$ as 
\beq\label{highestgrad}
\mf{g}=\bigoplus_{i=-2}^2\mf{g}_i, \qquad \mf{g}_i=\left\{x\in\mf{g}\;\vert\; [\theta^\vee,x]=ix\right\},
\eeq
which we call the \emph{highest root gradation}. 
We denote 
\beq\label{projection0}
\pi_j: \;\mf{g}\to \mf{g}_j
\eeq  
the natural projection from $\mf{g}$ onto the $j-$th component of the gradation, and we set
\beq
x_j=\pi_j(x),\qquad x\in \mf{g}.
\eeq
We describe in more detail the structure of the gradation \eqref{highestgrad}.  
Note that  $\mf{h}\subset \mf{g}_0$, and that $\mf{n}_+$ uniquely decomposes as
\beq\label{n+decomp}
\mf{n}_+=(\mf{g}_0\cap\mf{n}_+)\oplus\mf{g}_1\oplus\mf{g}_2.
\eeq
Let 
$I_{\theta^\vee}=\left\{j\in I\;\vert\;\langle \theta^\vee,\alpha_j\rangle=0\right\}\subset I,$
and denote by $\mf{g}[I_{\theta^\vee}]$ the semisimple Lie algebra generated by
$\{e_{i}, f_{i}, i\in I_{\theta^\vee}\}$. Then, we have
\beq\label{g0decomposition}
\mf{g}_0=\mf{g}[I_{\theta^\vee}]\oplus\bigoplus_{i\in I\setminus I_{\theta^\vee}} \bb{C}\alpha_i^\vee. 
\eeq
\begin{remark}
The set $I_{\theta^\vee}$ is depicted in Table \ref{fig:dynkin} as
the subset of white vertices in the Dynkin diagram. The subalgebra $\mf{g}[I_{\theta^\vee}]$ is isomorphic to the semi-simple
Lie algebra whose Dynkin diagram is obtained by the Dynkin diagram of $\mf{g}$, by removing
the black vertices and all the edges to which
they are connected.  These subalgebras are explicitly computed in Table \ref{tab:gItheta}. Moreover, setting $\mf{p}=\mf{g}_0\oplus\mf{u}$, with
$\mf{u}=\bigoplus_{i> 0}\mf{g}_i,$
then $\mf{p}$ is a parabolic subalgebra of $\mf{g}$, with $\mf{g}_0$ a reductive (Levi) subalgebra and $\mf{u}$ the nilradical of $\mf{p}$. 

\end{remark}

\begin{table}[H]
\caption{Dynkin diagrams for simple Lie algebras of $ADE$ type. White vertices correspond to roots perpendicular to $\theta^\vee$.}
\label{fig:dynkin}
\begin{align*}
&
A_n
\quad
\begin{tikzpicture}[start chain]
\blackdnode{1}
\dnode{2}
\dydots
\dnode{n-1}
\blackdnode{n}
\end{tikzpicture}
&\quad\quad&
E_6
\quad
\begin{tikzpicture}
\begin{scope}[start chain]
\foreach \dyni in {1,2,3,5,6} {
\dnode{\dyni}
}
\end{scope}
\begin{scope}[start chain=br going above]
\chainin (chain-3);
\blackdnodebr{4}
\end{scope}
\end{tikzpicture}
\\
&
D_n
\quad
\begin{tikzpicture}
\begin{scope}[start chain]
\dnode{1}
\blackdnode{2}
\node[chj,draw=none] {\dots};
\dnode{n-2}
\dnode{n}
\end{scope}
\begin{scope}[start chain=br going above]
\chainin(chain-4);
\dnodebr{n-1}
\end{scope}
\end{tikzpicture}
&\quad\quad&
E_7
\quad
\begin{tikzpicture}
\begin{scope}[start chain]
\foreach \dyni in {1,2,3,4,6} {\dnode{\dyni}}
\blackdnode{7}
\end{scope}
\begin{scope}[start chain=br going above]
\chainin (chain-4);
\dnodebr{5}
\end{scope}
\end{tikzpicture}
\\
&
&\quad\quad&
\!\!\!\!\!\!\!\!\!\!\!\!\!\!\!\!\!\!\!\!\!\!\!\!\!\!\!\!\!\!\!\!
\!\!\!\!\!\!\!\!\!\!\!\!\!\!\!\!\!\!\!\!\!\!\!\!\!\!\!\!\!\!\!\!
E_8
\quad
\begin{tikzpicture}
\begin{scope}[start chain]
\blackdnode{1}
\foreach \dyni in {2,3,4,5,7,8} {
\dnode{\dyni}
}
\end{scope}
\begin{scope}[start chain=br going above]
\chainin (chain-5);
\dnodebr{6}
\end{scope}
\end{tikzpicture}
&
\end{align*}
\end{table}


\begin{table}[H]\caption{The (dual) Coxeter number $h^\vee$, the semi-simple subalgebra $\mf{g}[I_{\theta}]$, its dimension, and the set $I\setminus I_\theta$ for any simply-laced Lie algebra $\mf{g}$.}
\label{tab:gItheta}
\normalsize
\begin{center}
\begin{tabular}{c||c|c|c|c} 
$\mf{g}$ &  $h^\vee$ & $\mf{g}[I_{\theta}]$   & $\dim\mf{g}[I_\theta]$ &  $I\setminus I_\theta$ 
\\
\hline\hline
$A_n$ & $n+1$  &  $A_{n-2}$ & $n^2-n$ & $\{1,n\}$ 
\\
\hline
$D_4$ &  $6$ & $A_1\oplus A_1\oplus A_1$  & $9$ &   $\{2\}$ 
\\
\hline
$D_n$, $n\geq 5$ &  $2n-2$  & $A_1\oplus D_{n-2}$ & $2n^2-9n+13$ &  $\{2\}$ 
\\
\hline
$E_6$ &  $12$  & $A_5$ &  $35$ & $\{4\}$ 
\\
\hline
$E_7$ &  $18$  & $D_6$ & $66$ & $\{7\}$ 
\\
\hline
$E_8$ &  $30$  & $E_7$ &  $133$ & $\{1\}$ 
\\
\hline
\end{tabular}
\end{center}
\end{table}

The dimension of the graded components of \eqref{highestgrad} is computed in the following
\begin{proposition}\label{propg1} Let $\mf{g}$ be simply laced, and consider the gradation \eqref{highestgrad}. Then a) $\dim\mf{g}_i=\dim\mf{g}_{-i}$, $i=0,1,2$. b)
We have:
\begin{align*}
\dim\mf{g}_0&=n(h^\vee+1)-4h^\vee+6,\\
\dim\mf{g}_1&=2(h^\vee-2), \\
\dim\mf{g}_{2}&=1.
\end{align*}
In particular, $\mf{g}_{2}=\mf{g}^{h^\vee-1}$ coincides with the root space of the highest root $\theta$.  
\end{proposition}
\begin{proof}
Part a) is obvious. For part b) we proceed as follows. The dimension of $\mf{g}_2$ is computed by first noticing that the roots proportional to $\theta$ are $\pm\theta$, with $(\theta,\pm\theta)=\pm2$, and then using Lemma \ref{lemmahumphreys} (ii) which implies that $\mf{g}_2=\mf{g}^{h^\vee-1}$, the root space of the highest root $\theta$. In particular, $\dim\mf{g}_2=1$. (Incidentally, the same argument shows that $\mf{g}_{-2}=\mf{g}^{1-h^\vee}$, the root space of the lowest root $-\theta$, so that $\dim\mf{g}_{-2}=1$). To compute $\dim \mf{g}_0$ note that due to  \eqref{g0decomposition} we have $\dim\mf{g}_0=\dim\mf{g}[I_{\theta^\vee}]+\#(I\setminus I_{\theta})$. By looking in Table \ref{tab:gItheta} at the values of $h^\vee$, $\dim\mf{g}[I_{\theta^\vee}]$ and $\#(I\setminus I_{\theta})$, one proves (case by case) that $\dim \mf{g}_0=n(h^\vee+1)-4h^\vee+6$. Finally, using part a) we get $\dim\mf{g}=\dim\mf{g}_0+2\dim\mf{g}_1+2\dim\mf{g}_2$. 
Substituting the values of $\dim\mf{g}_0$ and $\dim\mf{g}_2$ just obtained, and recalling that dimension of the simple Lie algebra $\mf{g}$ is given by $n(h+1)$, where $n$ is the rank and $h$ the Coxeter number (with $h=h^\vee$ since $\mf{g}$ is simply laced), then the last identity becomes $n(h^\vee+1)=n(h^\vee+1)-4h^\vee+6+2\dim\mf{g}_1+2$, which gives $\dim\mf{g}_1=2(h^\vee-2)$. 
\end{proof}
%


%
\subsection{The gradation induced by $R=-\theta^\vee+\eta$} Later we will be interested in gradations of
$\mf{g}$ induced by elements of the form $R=-\theta^\vee+\eta$, with $\eta\in\mf{n}_+$.
We begin by recalling the definition of Jordan-Chevalley decomposition.
\begin{definition}
 Let $R \in \mf{g}$. There exists a unique decomposition, named Jordan-Chevalley decomposition, of the following form $R=R_s+R_n$, with $R_s$ semisimple, $R_n$ nilpotent, and $[R_s,R_n]=0$. We denote $\sigma(R)=\sigma(R_s)$ the spectrum of $R$ in the adjoint representation.
\end{definition}
The following lemma will be very useful.
\begin{lemma}\label{lemma:semisimpleR}
Let $R=-\theta^\vee+\eta$, with $\eta\in\mf{n}_+$, and write
$\eta=\eta_0+\eta_1+\eta_2$ with $\eta_i \in \mf{g}_i$.
Then
\begin{enumerate}
 \item $\sigma(R)=\sigma(\theta^\vee)$;
 \item $R$ is semisimple if and only if $\eta_0=0$;
 \item If $R$ is semisimple then
 \beq\label{AdAtheta}
R=-e^{\ad(\eta_1+\frac{1}{2}\eta_2)}\theta^\vee.
\eeq
\end{enumerate}

\end{lemma}
\begin{proof}
Let $\bar{\eta}_1\in\mf{g}$ satisfy $(\ad \eta_0-1)\bar{\eta}_1=\eta_1$, and
$\bar{\eta}_2\in\mf{g}$ be such that  $(\ad \eta_0-2)\bar{\eta}_2=\eta_2+\frac{1}{2}[\bar{\eta}_1,\eta_1]$.
Then $\bar{\eta}_i\in\mf{g_i}\subset \mf{n}_+$, $i=1,2$, and $[\bar{\eta}_1,\bar{\eta}_2]=0$. Moreover, we have
$$e^{\ad(\bar{\eta}_1+\bar{\eta}_2)}R=-\theta^\vee+\eta_0.$$
Since $\theta^\vee$ is semisimple, $\eta_0$ is nilpotent and $[\theta^\vee,\eta_0]=0$, then we obtain that
\beq\label{RsRn}
R_s=-e^{-\ad(\bar{\eta}_1+\bar{\eta}_2)}\theta^\vee,\qquad R_n=e^{-\ad(\bar{\eta}_1+\bar{\eta}_2)}\eta_0
\eeq
are, respectively, the semisimple and nilpotent parts of the Jordan-Chevalley decomposition of $R$. From this, we obtain: (1) $\sigma(R)=\sigma(\theta^\vee)$, (2) $R$ is semisimple if and only if $\eta_0=0$, (3) if $R$ is semisimple so that  $\eta_0=0$, we have $ \bar{\eta}_1=-\eta_1$ and $\bar{\eta}_2=-\frac{\eta_2}{2}$, so that $R_s$ given by \eqref{RsRn} coincides with $R$ given by \eqref{AdAtheta}.
\end{proof}
%
Let us now consider the gradation
\beq\label{gradation-R}
\mf{g}=\bigoplus_{i\in\sigma(R)}\mf{g}_i(R), \qquad \mf{g}_i(R)=\left\{x\in\mf{g}\;\vert\; [R,x]=ix\right\},
\eeq
in case $R$ is semi-simple.
Note that due to \eqref{AdAtheta}, the gradation \eqref{gradation-R} is conjugated to the highest root gradation
\eqref{highestgrad}, namely
\beq\label{gtildeg}
\mf{g}_i(R)=\left\{e^{\ad(\eta_1+\frac{1}{2}\eta_2)}x\:\vert\:x\in\mf{g}_{-i}\right\}.
\eeq
For $j\in\sigma(R)$ we denote by $\pi^{R}_j$ the natural projection from $\mf{g}$ onto $\mf{g}_j(R)$.
Note that from \eqref{gtildeg} we have that 
$$\pi_j^R=e^{\ad(\eta_1+\frac{1}{2}\eta_2)} \pi_{-j} e^{-\ad(\eta_1+\frac{1}{2}\eta_2)}.$$
We write this formula in a very explicit form that we will need
in the following section, when dealing with trivial monodromy conditions. Let $x\in\mf{g}$ and denote by $x_i=\pi_i(x)$ the projection \eqref{projection0} of $x$ onto $\mf{g}_i$, then we have
\begin{align}\label{remk:expansion}
\pi^R_{2}(x)&=e^{\ad(\eta_1+\frac{1}{2}\eta_2)}x_{-2}\\ \nonumber
\pi^R_{1}(x)&=e^{\ad(\eta_1+\frac{1}{2}\eta_2)}(x_{-1}-\ad_{\eta_1}x_{-2})\\ \nonumber
\pi^R_{0}(x)&=e^{\ad(\eta_1+\frac{1}{2}\eta_2)}(x_0-\ad_{\eta_1} x_{-1}+
\frac{1}{2}\ad^2_{\eta_1}x_{-2}-\frac{1}{2}\ad_{\eta_2}x_{-2})\\ \nonumber
\pi^R_{-1}(x)&=e^{\ad(\eta_1+\frac{1}{2}\eta_2)}(x_1-\ad_{\eta_1}x_0+\frac{1}{2}\ad^2_{\eta_1}x_{-1}-
\frac{1}{2}\ad_{\eta_2}x_{-1}-\frac{1}{6}\ad^3_{\eta_1}x_{-2}\\ \nonumber
&\hspace{60pt}+\frac{1}{2}\ad_{\eta_1}\ad_{\eta_2}x_{-2})\\ \nonumber
\pi^R_{-2}(x)&=x_2-\ad_{\eta_1}x_1-\frac{1}{2}\ad_{\eta_2}x_{0}+
\frac{1}{2}\ad^2_{\eta_1}x_0-\frac{1}{6}\ad^3_{\eta_1}x_{-1}+\frac{1}{2}\ad_{\eta_1}\ad_{\eta_2}x_{-1}\\ \nonumber
&+\frac{1}{24}\ad^4_{\eta_1}x_{-2}+\frac{1}{8}\ad^2_{\eta_2}x_{-2}-\frac{1}{2}\ad^2_{\eta_1}\ad_{\eta_2}x_{-2}.
\end{align}
The above identities have been obtained by means of the following expansion:
\begin{align}
e^{\ad(\eta_1+\frac{1}{2}\eta_2)}=&1+\ad\eta_1+\frac{1}{2}\ad\eta_2+\frac{1}{2}\ad^2\eta_1+
\frac{1}{2}\ad\eta_1\ad\eta_2+\frac{1}{8}\ad^2\eta_2\notag\\
&+\frac{1}{6}\ad^3\eta_1+\frac{1}{4}\ad^2\eta_1\ad\eta_2+\frac{1}{24}\ad^4\eta_1.\label{gexp}
\end{align}
\begin{remark}
Note that 
$\mf{g}_1(R)\oplus\mf{g}_2(R)\subset \mf{n}_+$, while $\mf{g}_0(R)=\mf{h}^{R}\oplus (\mf{g}_0(R)\cap\mf{n}_+)$,
where $\mf{h}^{R}$ is the Cartan subalgebra conjugated to $\mf{h}$ under the automorphism $\exp(\ad(\eta_1+\frac{1}{2}\eta_2))$. 
\end{remark}

\subsection{A symplectic subspace}
Consider the vector subspace
\beq\label{symplecticspace}
\mf{t}=\bb{C}\theta^\vee\oplus\mf{g}_1\oplus\mf{g}_2\subset\mf{g},
\eeq
which due to Proposition \ref{propg1} is even dimensional,  of dimension $\dim\mf{t}=2(h^\vee-1)$. Note that we can write $\mf{t}$ as $\mf{t}=\bb{C}\theta^\vee\oplus [\theta^\vee,\mf{n}_+]$. Define on $\mf{t}$ the skew-symmetric bilinear form
\beq\label{simplg1}
\omega(x,y)=(E_{-\theta}\vert [x,y]), \qquad x,y\in\mf{t},
\eeq
where  $(\cdot\vert\cdot)$ is the normalized invariant bilinear form \eqref{normalizedkilling} on $\mf{g}$,
and $E_{-\theta}$ is the lowest weight vector of $\mf{g}$ introduced in Section \ref{sec:algebras}. The form
\eqref{simplg1} is also non-degenerate, therefore it defines a symplectic structure on $\mf{t}$.
We outline here two different approaches to prove this fact. The first approach is based on the fact
\cite{LPR18} that $\mf{t}=(\ker\ad E_{-\theta})^\perp$, the orthogonal complement (with respect to $(\cdot\vert\cdot)$ )
of the vector subspace $\ker\ad E_{-\theta}$. Therefore, $\mf{t}$ is a symplectic leaf of the  \lq frozen\rq\,
Lie-Poisson structure \cite{MaRa10}, and \eqref{simplg1} is nothing but the induced symplectic form on $\mf{t}$. 

For the second approach we present a canonical basis for \eqref{simplg1}, according to the following construction. Recall the root vectors $E_\alpha$, $\alpha\in\Delta$, of $\mf{g}$ introduced in Section \ref{sec:algebras}, satisfying the commutation relations \eqref{commrel}. From \eqref{highestgrad} we have that $E_{\alpha}\in\mf{g}_i$ if and only if $(\alpha | \theta)=i$, so that $\mf{g}_1=\langle E_\alpha, (\theta|\alpha)=1\rangle$ and $\mf{g}_2=\langle E_\theta\rangle$. In order to deal with elements in $\mf{t}$, we define the set $\Theta\subset\mf{h}^\ast$ as
\beq\label{Theta}
\Theta=\{0,\theta\}\cup\{\alpha\in\Delta\,|\, (\alpha | \theta)=1\}.
\eeq
Denoting $E_0=\theta^\vee/2$  then $\{E_\alpha \,|\, \alpha\in \Theta\}$ is a basis for $\mf{t}$. Recall the bimultiplicative function $\varepsilon_{\alpha,\beta}$ introduced in Section \ref{sec:algebras}.
\begin{lemma}\label{lemmaioio}
If $\alpha\in \Theta$, then $\theta-\alpha\in \Theta$ and $\theta-\alpha\neq\alpha$. Moreover, $\varepsilon_{\alpha,\theta}=-\varepsilon_{\theta-\alpha,\theta}$.
\end{lemma}
\begin{proof}
If $\alpha$ is equal to $\theta$ or $0$ the only nontrivial assertion is the last, where due to Lemma \ref{lemmaepsilon} we have $\varepsilon_{0,\theta}=1$ and $\varepsilon_{\theta,\theta}=-1$. Now let $\alpha\in\Theta$ with $(\alpha | \theta)=1$. Then, due to Lemma \ref{lemmahumphreys}(i) we have $\theta-\alpha\in\Delta$, and $(\theta-\alpha | \theta)=1$,  so that $\theta-\alpha\in\Theta$. Moreover, $\theta-\alpha=\alpha$ implies $\theta=2\alpha$ for $\alpha\in\Delta$, which is impossible. Thus, $\theta-\alpha\neq\alpha$. Finally, due to Lemma \ref{lemmaepsilon}, we have $\varepsilon_{\theta-\alpha,\theta}=\varepsilon_{\theta,\theta}\varepsilon_{-\alpha,\theta}=-\varepsilon_{\alpha,\theta}$.
\end{proof}
The function $\varepsilon$ takes values $\pm 1$. We introduce the subset
\beq\label{Thetatilde}
\widetilde{\Theta}=\{\alpha\in\Theta\,|\,\varepsilon_{\theta,\alpha}=1\}\subset \Theta,
\eeq
and due to Lemma \ref{lemmaioio} we have $\Theta=\{\alpha,\theta-\alpha\, |\, \alpha\in\widetilde{\Theta}\}$. 
\begin{proposition}
For $x,y\in\mf{t}$, with $x=\sum_{\alpha\in \Theta}x^\alpha E_\alpha$ and $y=\sum_{\alpha\in \Theta}y^\alpha E_\alpha$ then \eqref{simplg1} takes the form
$$\omega(x,y)=-\sum_{\alpha\in\widetilde{\Theta}}\left(x^{\alpha}y^{\theta-\alpha}-x^{\theta-\alpha} y^{\alpha}\right),$$
where $\widetilde{\Theta}$ is given by \eqref{Thetatilde}. Thus, $\{E_\alpha \,|\,\alpha\in \Theta\}$ is a canonical basis for $\mf{t}$.
\end{proposition}
\begin{proof}
Let $c$ be the coefficient of $E_\theta$ in the commutator $[X,Y]$. Due to \eqref{simplg1} and \eqref{Enormalized} we have $\omega(X,Y)=-c$. An explicit computation gives 
$$c=\sum_{\alpha\in\Theta}\epsilon_{\theta,\alpha}x^\alpha y^{\theta-\alpha}=\sum_{\alpha\in\widetilde{\Theta}}(x^\alpha y^{\theta-\alpha}-x^{\theta-\alpha}y^\alpha),$$
where in the last step we used Lemma \ref{lemmaioio}.
\end{proof}
The symplectic space $\mf{t}$  will play an important role in Section \ref{sec:zeromonodromykdv}, when computing the trivial monodromy conditions for quantum-KdV opers. 
\begin{remark} In order to obtain a canonical basis for the form \eqref{simplg1} we chose in the construction above  $E_0=\theta^\vee/2$. In the computations we will perform in Sections \ref{sec:zeromonodromykdv} and \ref{sec:computations} it will be slightly more convenient to choose $E_0=\theta^\vee$. We will point out our choice whenever required.
\end{remark}

\vspace{20pt}
\section{Trivial monodromy at a regular singular point}\label{sec:singularpoint}
In this Section -- following \cite{bava83} -- we consider an arbitrary
linear operator with a first order pole, and we derive necessary and sufficient conditions (on its Laurent series) to have trivial monodromy. We then specialize to the case of the localisation of a Quantum $\mf{g}-$KdV oper \eqref{eq:ouroperators} at an additional singularity.

Let $x=0$ be the singular point, and assume that the operator has the expansion
\beq\label{genoper}
\mc{L}=\partial_x+\frac{R}{x}+\sum_{k\geq0} a^k x^k,
\eeq
with $R,a^k \in \mf{g}$.

\begin{definition}
We say that the operator \eqref{genoper} has trivial monodromy at $x=0$ if,
for any finite dimensional $\mf{g}$-module $V$,
the differential equation $\mc{L}\psi=0$, with $\psi: \bb{C} \to V$, has trivial monodromy
\end{definition}
It is well known that the eigenvalues of the monodromy matrix
are of the form $\exp{2 i \pi \la}$, where $\la$ an eigenvalue of $R$.
Since we look for conditions on $\mc{L}$ such that the monodromy matrix is the identity in every representation,
we must restrict to the case where $R$ has integer eigenvalues in any finite-dimensional representation of $\mf{g}$. For this reason, we assume that the semi-simple part $R_s$ of $R$ in the Jordan-Chevalley decomposition is conjugated to an element of the co-root lattice $Q^\vee$ of $\mf{g}$. In fact, see Proposition \ref{prop:bava}(5), the operator \eqref{genoper} has trivial monodromy if and only if it has trivial monodromy in the adjoint representation and $R_s$ is conjugated to an element of the co-root lattice.

Because of our assumption on $R_s$, we have the $\bb{Z}$-gradation
\beq\label{gradationA-12}
\mf{g}=\bigoplus_{i\in\sigma(R)}\mf{g}_i(R), \qquad \mf{g}_i(R)=\left\{x\in\mf{g}\;\vert\; [R_s,x]=ix\right\}.
\eeq
In order to compute the monodromy of $\mc{L}$ at $x=0$ we transform it into its
\emph{aligned} form. The following definition is adapted from \cite{bava83}.
\begin{definition}[\cite{bava83}]\label{def:aligned}
The operator \eqref{genoper} is said to be \emph{aligned} (at $x=0$) if $a^i\in\mf{g}_{-i-1}(R)$ for $i\geq 0$.
\end{definition}
\begin{proposition}[\cite{bava83}]\label{prop:bava}
Let $\mc{L}$ be a connection with local expansion \eqref{genoper}, such that $R_s$ is conjugated
to an element of the co-root lattice $Q^\vee$, and let $m=\max \sigma(R)$. Let moreover $G[[x]]_1$
be the sub-group of Gauge transformations of the form $G[[x]]_1=\left\{1+\sum_{i\geq 1} M_i x^i\;\vert\;
M_i\in\mf{g}\right\}$ and positive radius of convergence. 
\begin{enumerate}
 \item $\mc{L}$ is equivalent to an aligned connection by a transformation in $G[[x]]_1$.
 \item The monodromy of $\mc{L}$ coincides with the monodromy of the aligned connection equivalent to it.
 \item If $a^i \in \mf{g}_{-i-1}(R)$ for $i=0,\dots,m-1$, then $\mc{L}$ is conjugated in $G[[x]]_1$ to the aligned connection
 $\partial_x+x^{-1}R+\sum_{i=0}^{m-1} a^i x^i$.
 \item  The monodromy of the connection $\mc{L}$ depends only on the first $m+1$ coefficients of the expansion
 of $\mc{L}$ at $x=0$, namely $R, a^0,\dots,a^{m-1}$.
 \item Let $\mc{L}$ be aligned.  The monodromy is trivial
 if and only if  $R_n=0$, and $a^k=0$ for  $k=0,\dots,m-1$.\label{prop:bava5}
\end{enumerate}

\begin{proof}
 The proof can be found in \cite[Section 3]{bava83}. Here we just comment on part $(5)$. Let $\mc{L}$ be aligned, namely of the form
 \beq\label{L1}
\mc{L}=\de_x+\frac{R}{x}+\sum_{k=0}^{m-1} a^k x^k, \qquad a^k\in\mf{g}_{-k-1}(R), 
\eeq
with $m\leq\max\sigma(R)$. Since $R_s$ belongs to the co-root lattice, the Gauge transformation $x^{\ad R_s}$ is single valued in any finite dimensional $\mf{g}-$module. We have
\beq\label{euler}
x^{\ad R_s} \mc{L}=\de_x+\frac{R_n+\sum_{k=0}^{m-1} a^k}{x} \; ,
\eeq
and the monodromy of the above operator coincides with the monodromy of $\mc{L}$. We now prove that the element $R_n+\sum_{k=0}^m a^k$ is nilpotent. Indeed, by definition $R_n\in\mf{g}_0(R)$, so that $\ad_{R_n}:\mf{g}_j(R)\to \mf{g}_j(R)$, while by hypotesis $a_k\in \mf{g}_{-k-1}(R)$, so that $\ad_{a_k}:\mf{g}_j(R)\to \mf{g}_{j-k-1}(R)$, with $k\geq 0$.
Suppose that $\la \neq 0$
is a non trivial eigenvalue of $R_n+\sum_{j=0}^m a^j$ with eigenvector $0\neq y=\sum_k y_k$ and $y_k \in \mf{g}_k(R)$.
If $K$ is the minimum integer such that $y_k\neq0$, then $y_K$ is an eigenvector of $R_n$ with non-trivial eigenvalue $\la$,
which is a contradiction, because $R_n$ is nilpotent.

Now let $V$ be a non-trivial finite dimensional $\mf{g}-$module, let $\{\psi_i,\,i=1,\dots,\dim V\}$  be a basis of  $V$. Then
the functions $\Psi_i=x^{-\big(R_n+\sum_{j=0}^{m-1} a^j\big)}\psi_i$, with $i=1 \dots \dim V$, satisfy $(x^{\ad R_s}\mc{L})\Psi_i=0$ and form a basis of solutions. The monodromy matrix is therefore given by $$\mc{M}=\exp\big(-2 i \pi  (R_n+\sum_{j=0}^{m-1} a^j)\big) \;.$$
Since $R_n+\sum_{j=0}^{m-1} a^j$ is nilpotent then the monodromy is trivial
if and only if $R_n+\sum_{j=0}^{m-1} a^j=0$,
from which the thesis follows.

\end{proof}

\end{proposition}

Now let us consider what happens in the particular case of a quantum $\mf{g}-$KdV oper \eqref{eq:ouroperators}. Localising at $w_j$, we obtain a connection of the form \eqref{genoper} with  $R=-\theta^\vee+\eta,$ where $\eta=\sum_{i \geq 1} w_j^{-i} X^i(j) \in \mf{n}^+$.
Due to \eqref{n+decomp} we can write 
$$R=-\theta^\vee+\eta_0+\eta_1+\eta_2,$$
with $\eta_0\in\mf{g}_0\cap\mf{n}_+$, $\eta_1\in\mf{g}_1$ and $\eta_2\in\mf{g}_2$. We can then apply Lemma \ref{lemma:semisimpleR} from which we obtain that $\sigma(R)=\sigma(\theta^\vee)$. In particular, $\max \sigma(R)=2$, so that according to Proposition \ref{prop:bava}(4), the monodromy of the quantum KdV opers at $z=w_j$ depends only on the first three terms of the Laurent expansion. In the next theorem we derive necessary and sufficient conditions for the trivial monodromy of an operator of the form 
\beq\label{genoper1}
\mc{L}=\partial_x+\frac{R}{x}+a+bx+O(x^2),
\eeq
with $ R=-\theta^\vee+\eta_0+\eta_1+\eta_2$ and $a,b\in\mf{g}$.
\begin{theorem}\label{thm:zeromonodromy}
The operator  \eqref{genoper1} has trivial monodromy at $x=0$ if and only if 
\begin{subequations}\label{zeromonodromy}
\begin{align}
&\eta_0=0,\label{zeromonodromya}\\
&\pi^{R}_{-1}(a)=0,\label{zeromonodromyb}\\
&\pi^{R}_{-2}(b)=[\pi^{R}_{-2}(a),\pi^{R}_{0}(a)],\label{zeromonodromyc}
\end{align}
\end{subequations}
where $\pi^{R}_{i}$ denotes the projection of $\mf{g}$ onto $\mf{g}_i(R)=\lbrace x \in \mf{g}, [-\theta^\vee+\eta_1+\eta_2,x]= ix \rbrace$.
\end{theorem}
\begin{proof} 
Let $\mc{L}$ be an operator of the form \eqref{genoper1}. Due to Proposition \ref{prop:bava} the monodoromy of $\mc{L}$ coincides with the monodromy of its aligned form, and the residue $R$ of $\mc{L}$ coincides with the residue of the aligned operator. Moreover, again by Proposition \ref{prop:bava}, if an aligned oper has trivial monodromy then $R_n=0$. It follows that if the monodromy of $\mc{L}$ is trivial then $R_n=0$. By Lemma \ref{lemma:semisimpleR}, the element $R=-\theta^\vee+\eta_0+\eta_1+\eta_2$ is semisimple if and only if $\eta_0=0$. This proves condition \eqref{zeromonodromya}.

By Proposition \ref{prop:bava}, the operator \eqref{genoper}
with $R=-\theta^\vee+\eta_1+\eta_2$ is equivalent to an operator of the form
\beq\label{proofC(z)}
\mc{L}_1=\de_x+\frac{-\theta^\vee+\eta_1+\eta_2}{x}+C^0+C^1x + O(x^2), \qquad C^i\in \mf{g}_{-i-1}(R), \, i=0,1,
\eeq
by a transformation in $G[[x]]_1$, and the monodromy is trivial if and only if
$C^0=C^1=0$. 
The thesis is proved once we show that $C^0=\pi^{R}_{-1}(a)$ and $C^1=\pi^{R}_{-2}(b)-[\pi^{R}_{-2}(a),\pi^{R}_{0}(a)]$. In order to obtain $C^0,C^1$ we look for a Gauge transformation of the form $e^{x^2 T'}e^{x T}\in G[[x]]_1$, such that
$e^{xT}\mc{L}=\partial_x+\frac{R}{x}+C^0+O(x)$, and $e^{x^2 T'}e^{x T}.\mc{L}=\partial_x+\frac{R}{x}+C^0+C^1 x+O(x^2)$. We have
\beq\label{proof:TL1}
e^{x T}\mc{L}=\de_x+\frac{R}{x}+D^0+D^1 x+\dots,
\eeq
where
\begin{align}
D^0&=a-T+[T,R]\label{aligncond1}\\
D^1&=b+[T,a]+\frac{1}{2}[T,[T,R]].\label{aligncond2}
\end{align}
 Now we look for a $T$ such that $D^0$ is aligned, namely $D^0=\pi_{-1}^R(D^0)$.
Writing $T=\sum_i \pi^R_i(T)$, with $\pi^R_i(T)\in\mf{g}_i(R)$, and choosing
\begin{equation}\label{eq:Tinproof}
 \pi_{-1}^R(T)=0, \qquad \pi_j^R(T)=\frac{\pi^R_j(a)}{j+1},\quad  j\neq-1,
\end{equation}
we get $D^0=\pi^R_{-1}(a)$, which is aligned. Hence
\begin{equation}\label{eq:C0inproof}
 C^0=D^0=\pi^R_{-1}(a) \; .
\end{equation}
This proves the second condition \eqref{zeromonodromyb}. Inserting now \eqref{aligncond1}, \eqref{eq:C0inproof}, and \eqref{zeromonodromyb}, into \eqref{aligncond2} we obtain
\begin{equation}
 D^1=b+\frac{1}{2}[T,a] \; .
\end{equation}
Now we look for a $T' \in \mf{g}$ such that $\exp{x^2T'}\exp{x T}.\mc{L}=\partial_x+\frac{R}{x}+C^0+C^1 x+ O(x^2)$.
By repeating the same steps as above, one shows that
\beq\label{pf:C2}
C^1=\pi^R_{-2}(D^1)=\pi^R_{-2}(b)+\frac{1}{2}\pi^{R}_{-2}([T,a]) \; .
\eeq
Since $\pi_{-2}^R[T,a]=\sum_{j\geq0}[\pi_{-j}^R(T),\pi_{-2+j}(a)] $, and using \eqref{eq:Tinproof} together with
\eqref{zeromonodromyb}, we obtain that 
$$C^1=\pi^R_{-2}(b)-[\pi^R_{-2}(a),\pi^R_0(a)],$$
from which the third condition \eqref{zeromonodromyc} follows.
\end{proof}
\begin{remark}
The Quantum KdV opers \eqref{eq:ouroperators} depend on two set of unknowns, the location of the poles $w_j,j\in J$
and the local coefficients $X^i(j) \in \mf{g}^i, j \in J$. If $J=\lbrace 1\dots N\rbrace $,
these are $ (1+\dim\mf{n}_+) N=(1+\frac{n \hv }{2} )N$ variables. However, due to the previous theorem,
the condition $\pi_0(X^i(j))=0$, holds for every  $i=1,\dots,h^\vee-1$ and $j=1,\dots,N$, implying that
$$X^i(j)\in \mf{g}_1 \oplus \mf{g}_2,\qquad i=1,\dots,h^\vee-1,\quad j=1,\dots,N.$$
The space $ \mf{g}_1 \oplus \mf{g}_2$ is a codimension $1$ vector subspace of the symplectic space $\mf{t}$ introduced in the previous section. As a consequence, the number of non-trivial unknowns reduces to $N\dim \mf{t}=2 N (\hv-1)$.  This fact represents a major advantage of working with the Gauge where all singularities are first order poles.
Working for instance in a canonical Gauge $\mf{s}$, then necessarily the singularities are higher order poles, and the total number of non-trivial unknowns we need to consider is again $(1+\frac{n \hv }{2} )N$ -- a number which grows quadratically with $n$ -- rather than $2 N(\hv-1)$, which grows linearly.
\end{remark}

The trivial monodromy conditions \eqref{zeromonodromy} for the operator \eqref{genoper1} are written in terms of the gradation
\eqref{gradation-R}. This gradation depends on the coefficients $\eta_1,\eta_2$, which are unknowns
of our problem. In order to be able to derive an explicit system of equations for these unknowns, we
write conditions \eqref{zeromonodromy} with respect to the fixed gradation
\eqref{highestgrad}, namely the gradation induced
by the highest root. From here below we restrict to operators with local expansion \eqref{genoper1}
such that $a \in f+\mf{b}_+$, $b \in \mf{b}_+$, for this is the case of the Quantum $\mf{g}-$KdV oper \eqref{eq:ouroperators}.
From now on $\mf{g}\neq \mf{sl}_2$, and the $\mf{sl}_2$ case in treated in a separate section of the paper.

\begin{lemma}\label{lemmacoeffs}
Let $\mc{L}\in\ope(K_{\bb{P}^1})$ be given by \eqref{genoper1}, with
$a \in f+\mf{b}^+$, $b \in \mf{b}^+$. Moreover, for $j\in \bb{Z}$ set
$a_j=\pi_j(a)$, $b_j=\pi_j(b)$, $j\in\bb{Z}$, where $\pi_j$ is the projection defined in
\eqref{projection0}.
If $\mf{g}$ is not of type $\mf{sl}_2$, then $a_{-2}=b_{-2}=b_{-1}=0$.
\end{lemma}
\begin{proof} 
We have $b\in \mf{b}\subset \mf{g}_0\oplus\mf{g}_1\oplus\mf{g}_2$, which implies that $b_{-1}=b_{-2}=0$,
while $a\in f+\mf{b}\subset \mf{g}_{-1}\oplus\mf{g}_0\oplus\mf{g}_1\oplus\mf{g}_2$, implying $a_{-2}=0$.
\end{proof}
We can now write the trivial monodromy conditions  \eqref{zeromonodromy} in terms of the gradation \eqref{highestgrad}.
\begin{proposition}\label{prop:zeroetaab}
Let $\mc{L}\in \ope(K_{\bb{P}^1})$ be given by \eqref{genoper1}, with
$a \in f+\mf{b}_+$, $b \in \mf{b}_+$. 
If $\mf{g}$ is not of type $\mf{sl}_2$, then the trivial monodromy conditions
\eqref{zeromonodromy} are equivalent to the following system of equations for $\eta\in\mf{n}_+$, $a\in f+\mf{b}_+$, $b\in\mf{b}_+$:
\begin{subequations}\label{gimonodromy}
\begin{align}
&\eta_0=0\label{gimonodromy0}\\
&2a_1-2[\eta_1, a_0]+[\eta_1,[\eta_1,a_{-1}]]-[\eta_2,a_{-1}]=0,\label{gimonodromy1}\\
&2b_2+[\eta_1,[\eta_1,b_0]]-2[\eta_1,b_{1}]-[\eta_2,b_0]\notag\\
&\,\,\, -\left[2a_2-[\eta_2,a_0]+\frac{1}{3}[\eta_1,[\eta_1,a_0]+2[\eta_2,a_{-1}]-4a_1],a_0-[\eta_1,a_{-1}]\right]=0,\label{gimonodromy2}
\end{align}
\end{subequations}
where $\eta_i=\pi_i(\eta),\,a_i=\pi_i(a),\,b_i=\pi_i(b)\in\mf{g}_i$.
\end{proposition}
\begin{proof}
Due to Lemma
\ref{lemmacoeffs} we have $a_{-2}=0$ and plugging $x=a$ in the
relations \eqref{remk:expansion} we get
\begin{align*}
&\pi^R_{2}(a)=0\\
&\pi^R_{1}(a)=e^{\ad(\eta_1+\frac{1}{2}\eta_2)}a_{-1}\notag\\
&\pi^R_{0}(a)=e^{\ad(\eta_1+\frac{1}{2}\eta_2)}(a_0-\ad_{\eta_1}a_{-1})\notag\\
&\pi^R_{-1}(a)=e^{\ad(\eta_1+\frac{1}{2}\eta_2)}(a_1+\frac{1}{2}\ad_{\eta_1}^2 a_{-1}-\frac{1}{2}\ad_{\eta_2} a_{-1}-\ad_{\eta_1} a_0)\\
&\pi^R_{-2}(a)=a_2-\frac{1}{6}\ad_{\eta_1}^3 a_{-1}+\frac{1}{2}\ad_{\eta_1}^2 a_0-\frac{1}{2}\ad_{\eta_2} a_0+\frac{1}{2}\ad_{\eta_1}\ad_{\eta_2}a_{-1}-\ad_{\eta_1} a_1.
\end{align*}
On the other hand, after Lemma \ref{lemmacoeffs} we have $b_{-1}=b_{-2}=0$, so that for
$x=b$ the relations \eqref{remk:expansion} become $\pi^R_{2}(b)=\pi^R_{1}(b)=0$, and
\begin{align*}
&\pi^R_{0}(b)=e^{\ad(\eta_1+\frac{1}{2}\eta_2)}b_0\notag\\
&\pi^R_{-1}(b)=e^{\ad(\eta_1+\frac{1}{2}\eta_2)}(b_1-\ad_{\eta_1} b_0)\\
&\pi^R_{-2}(b)=b_2+\frac{1}{2}\ad_{\eta_1}^2 b_0-\frac{1}{2}\ad_{\eta_2} b_0-\ad_{\eta_1} b_1.
\end{align*}
Plugging these formulae into \eqref{zeromonodromy}, one gets \eqref{gimonodromy}.
\end{proof}

\vspace{20pt}
\section{Trivial monodromy for quantum-KdV opers}\label{sec:zeromonodromykdv}
In this section we address the final assumption, namely Assumption \ref{asu4}, on quantum $\mf{g}-$KdV opers:
 for any value of the loop algebra parameter $\la$, the monodromy at any regular non-zero singular point must be trivial. As a result, we completely characterise quantum $\mf{g}$-KdV opers by means of a system of rational equations, see Proposition \ref{systemonc}.

We thus consider an oper of the form
\eqref{eq:ouroperators} such that the set $J$ of additional singularities is
non-empty, namely $J=\lbrace 1,\dots, N\rbrace$,
for some $N\in\bb{Z}_+$. Hence \eqref{eq:ouroperators} reads
\beq\label{norformguess}
\mc{L}=\de_z+f+\frac{r}{z}+(z^{1-h^\vee}+\lambda z^k)E_{\theta}+\sum_{j=1}^{N}\frac{1}{z-w_j}
\left(-\theta^\vee +\sum_{i=1}^{h^\vee\!-1}\frac{X^i{(j)}}{z^i}\right), 
\eeq
where we put, and use from now on for convenience, 
\beq\label{eq:k}
k=1-h^\vee-\hat{k}\in (-h^\vee,1-h^\vee)
\eeq
in place of $\hat{k}\in (0,1)$. In formula \eqref{norformguess}, the quantities  $r\in\mf{h}$ and $k\in (-h^\vee,1-h^\vee)$  are given, $\lambda\in\bb{C}$ is arbitrary,
while the non-zero pairwise distinct complex numbers $w_j$,
and the Lie algebra elements
$X^i(j)\in\mf{g}^i$ are to be determined by the trivial monodromy conditions.

For any fixed $\ell=1,\dots,N$, the localization of $\mc{L}$ at $z=w_\ell$ yields an expansion of the form \eqref{genoper1}. Indeed, using  $x=z-w_\ell$ as local coordinate, we get
\beq\label{ellexpansion}
\partial_x+ \frac{R(\ell)}{x}+a(\ell)+b(\ell)x+O(x^2), 
\eeq
where the coefficients $R(\ell)=-\theta^\vee+\eta(\ell)$ with $\eta(\ell)\in\mf{n}_+$, $a(\ell)\in f+\mf{b}_+$ and $b(\ell)\in\mf{b}_+$  can be obtained from
\eqref{norformguess}. Since $\mc{L}$ is of
type \eqref{genoper1}, the trivial monodromy conditions at $z=w_\ell$ are provided by Proposition \ref{prop:zeroetaab}.
Imposing that those trivial monodromy conditions are fulfilled for any value of $\la$, and
using the expression of $\eta(\ell), a(\ell), b(\ell)$ in terms of the coefficients of \eqref{norformguess}, we obtain below a complete set of equations for the unknowns $w_j,X^i(j),j=1,\dots,N$.

\subsection{Vector notation}\label{subsec:vectornotation}
In order to deal with all singularities $\{w_j,j=1,\dots,N\}$ at once, it will be useful to consider the following construction. For every pair of vectors ${\bf v}=(v_1,\dots,v_N)$ and ${\bf v}'=(v'_1,\dots,v'_N)$ in $\bb{C}^N$,  denote by ${\bf v}\circ {\bf v}'$ their Hadamard product:
$${\bf v}\circ {\bf v}'=(v_1v'_1,\dots,v_Nv'_N)\in\bb{C}^N,$$
and extend the Lie algebra structure form $\mf{g}$ to the tensor product $\bb{C}^N\!\otimes\mf{g}$ by letting $[{\bf v}\otimes x,{\bf v}'\otimes y]=({\bf v}\circ {\bf v}')\otimes [x,y]$, for ${\bf v},{\bf v}'\in \bb{C}$ and $x,y\in\mf{g}$. Setting ${\bf 1}=(1,\dots,1)\in\bb{C}^N$, then we have an injective homomorphism of Lie algebras $\mf{g}\hookrightarrow \bb{C}^N\!\otimes\mf{g}$ given by $x\mapsto {\bf 1}\otimes x$, $x\in\mf{g}$. By abuse of notation we denote in the same way elements of  $\mf{g}$ and their images in $\bb{C}^N\!\otimes\mf{g}$ under this homomorphism.
The elements $X^i(j)$ appearing in \eqref{norformguess} can now be written in the more compact form
\beq\label{Xi0906}
X^i=(X^i(1),\dots,X^i(N))\in \bb{C}^N\otimes \mf{g}^i, \quad i=0,\dots,h^\vee-1.
\eeq
Moreover, for the additional poles $w_j$, $j=1,\dots, N$ denote
\beq\label{W}
{\bf w}=(w_1,\dots,w_n)\in\bb{C}^N,
\eeq
and for $s\in\bb{R}$ put ${\bf w}^s=(w^s_1,\dots,w^s_n)\in\bb{C}^N$. Let  $\bb{C}({\bf w})=\bb{C}(w_1,\dots,w_N)$
be the field of fractions of the polynomial ring
$\bb{C}[{\bf w}]=\bb{C}[w_1,\dots,w_N]$.
For $i\in\bb{Z}$, introduce the $N\times N$ matrices -- with values in $\bb{C}({\bf w})$ -- given by
\beq\label{operatorAB}
(A_i)_{\ell j}=
\begin{cases}
\frac{w_\ell}{w_\ell-w_j},& \ell\neq j\\
-i& \ell=j
\end{cases},
\qquad 
(B_i)_{\ell j}=
\begin{cases}
(i+\frac{w_\ell}{w_\ell-w_j})\frac{w_\ell}{w_\ell-w_j},& \ell\neq j\\
-\frac{i(i+1)}{2}& \ell=j
\end{cases},
\eeq
and define $A,B\in \End_{\bb{C}({\bf w})}(\bb{C}^N\otimes \mf{g})$ as
\begin{align}
A({\bf v}\otimes x)&=A_i({\bf v})\otimes x,\qquad x\in \mf{g}^i,\\
B({\bf v}\otimes x)&=B_i({\bf v})\otimes x,\qquad x\in \mf{g}^i
\end{align}
In addition, introduce $M,S,Y\in \End_{\bb{C}({\bf w})} (\bb{C}^N\otimes \mf{g})$ as follows. For $X\in\bb{C}^N\otimes \mf{g}$ let
\begin{subequations}\label{MSY}
\begin{align}
&M(X)=A(X)-A_0({\bf 1})\circ X+[r,X],\\
&S(X)=2B(X)-B_0({\bf 1})\circ X+\frac{4}{3}k A(X)-\frac{k}{3}A_0({\bf 1})\circ X+(1+\frac{k}{3})[r,X],\\
&Y(X)=\frac{2}{3}k^2 X+\frac{2}{3}k A(X)+2B(X)+(1+\frac{k}{3})[r,X].
\end{align}
\end{subequations}
We can now express the trivial monodromy conditions for the operator \eqref{norformguess}.
Recall that 
\beq\label{fprojections}
\pi_0(f)=-\sum_{i\in I_{\theta}}E_{-\alpha_i},\quad \pi_{-1}(f)=-\sum_{i\in I\setminus I_{\theta}}E_{-\alpha_i}
\eeq
are the projections \eqref{projection0} of the principal nilpotent element $f$
with respect to the highest root gradation \eqref{highestgrad}. The trivial monodromy conditions for the general case are expressed as follows:
\begin{proposition}\label{prop:monodromycg}
Let $\rank{\mf{g}}>1$. The operator \eqref{norformguess} has trivial monodromy at $z=w_\ell$ for every $\ell=1,\dots,N$ and
every $\lambda\in\bb{C}$ if and only if the following conditions are satisfied:
\begin{enumerate}
\item For $i=0,\dots,h^\vee-1$, the variables $X^i$ belong to $\bb{C}^N\otimes \mf{t}^i$, where $\mf{t}=\bb{C}\theta^\vee\oplus\mf{g}_1
\oplus \mf{g}_2\subset\mf{b}^+$ is
the symplectic vector space defined in \eqref{symplecticspace}, and $\mf{t}^i=\mf{t}\cap \mf{g}^i$.\label{xint}
\item the set of variables $\{{\bf w}\}\cup\{X^i\,|\, i=1,\dots,h^\vee-1\}$ satisfies the following system of equations
\end{enumerate}
\begin{subequations}\label{zmnf1}
\begin{align}
&[[X^1,\pi_{-1}(f)],E_\theta]=(k{\bf 1}+\theta(r){\bf 1}-2A_0({\bf 1}))\otimes E_{\theta},\label{zmnf2}\\
&[X^{i+1},\pi_0(f)]=M(X^i)+\frac{1}{2}\sum_{s=0}^i[X^s,[X^{i+1-s},\pi_{-1}(f)]],\qquad i=1,\dots,h^\vee-2,\label{zmnf3}\\
&2(1-h^\vee-k){\bf w}\otimes E_\theta+Y(X^{h^\vee-1})+\sum_{i=1}^{h^\vee-1}[X^{h^\vee-1-i},S(X^i)]\notag\\
&\hspace{70pt}=\frac{k}{3}\sum_{i=2}^{h^\vee-2}[X^i,[X^{h^\vee-i},\pi_0(f)]].\label{zmnf4}
\end{align}
\end{subequations}
Here, $X^0=-{\bf 1}\otimes\theta^\vee$, $E_\theta\in\mf{g}$ is the highest root vector defined in $\S$\ref{gbasis},
the operator $A_0\in\End_{\bb{C}({\bf w})}(\bb{C}^N)$ is given in \eqref{operatorAB} and
$M,S,Y\in\End_{\bb{C}({\bf w})}(\bb{C}^N\otimes \mf{g})$ are given in \eqref{MSY}.
\end{proposition}
\begin{proof}
Recall that the monodromy about $z=w_\ell$ of an operator with expansion \eqref{ellexpansion} is encoded in the the quantities 
$R(\ell),a(\ell),b(\ell)$. These in turn can be expressed in terms of the variables $X^i,i=1\dots\hv-1$.
Defining ${\bf R}=(R(1),\dots,R(N) $, ${\bf a}=(a(1),\dots a(N))$, ${\bf b}=(b(1),\dots,b(N))$ we have  
\begin{align}\label{vectoreta}
{\bf R}&=-{\bf 1}\otimes \theta^\vee+{\bm \eta},\qquad
{\bm\eta}=\sum_{i=1}^{h^\vee-1}{\bf w}^{-i}\circ X^i\in\bb{C}^N\otimes\mf{g}.\\ \nonumber
{\bf a}&={\bf 1}\otimes f+{\bf w}^{-1}\otimes r+({\bf w}^{1-h^\vee}+
\lambda {\bf w}^k)\otimes E_{\theta}+\sum_{i=0}^{h^{\vee}-1}{\bf w}^{-i-1}\circ A(X^i)\\ \nonumber
{\bf b}&=-{\bf w}^{-2}\otimes r+((1-h^\vee){\bf w}^{-h^\vee}+k\lambda {\bf w}^{k-1})
\otimes E_{\theta}-\sum_{i=0}^{h^\vee-1}{\bf w}^{-i-2}\circ B(X^i).
\end{align}
Note that the projections \eqref{projection0} onto the eigenspaces of $\ad_{\theta^\vee}$
can be extended uniquely to $\bb{C}^N\otimes \mf{g}$ by the rule $\pi_i({\bf v}\otimes x)={\bf v}\otimes \pi_i(x)$,
with ${\bf v}\in\bb{C}^N$ and $x\in\mf{g}$. From \eqref{gimonodromy}, the trivial monodromy conditions
at all points $w_j$, $j=1,\dots,N$  can thus be written in the following compact form 
\begin{subequations}\label{vectorzeromon}
\begin{align}
&{\bm \eta}_0=0, \label{vectorzeromon0}\\
&2{\bf a}_1-2[{\bm \eta}_1,{\bf a}_0]+[{\bm\eta}_1,[{\bm\eta}_1,{\bf a}_{-1}]]-[{\bm\eta}_2,{\bf a}_{-1}]=0,\label{vectorzeromon1}\\
&2{\bf b}_2+[{\bm\eta}_1,[{\bm\eta}_1,{\bf b}_0]]-2[{\bm\eta}_1,{\bf b}_{1}]-[{\bm\eta}_2,{\bf b}_0]\notag\\
&\,\,\, =\left[2{\bf a}_2-[{\bm\eta}_2,{\bf a}_0]+\frac{1}{3}[{\bm\eta}_1,[{\bm\eta}_1,{\bf a}_0]+2[{\bm\eta}_2,{\bf a}_{-1}]-4{\bf a}_1],{\bf a}_0-[{\bm\eta}_1,{\bf a}_{-1}]\right].\label{vectorzeromon2}
\end{align}
\end{subequations}
Recall that by definition $X^i\in\bb{C}^N\otimes \mf{g}^i$, and that $X^0=-{\bf 1}\otimes\theta^\vee\in\bb{C}^N\otimes \mf{t}^0$. Due to  \eqref{vectorzeromon0} then from \eqref{vectoreta} we obtain $\pi_0(X^i)=0$ for $i=1,\dots, h^\vee-1$, and --  by the definition of $\mf{t}$ -- this implies $X^i\in\bb{C}^N\otimes \mf{t}^i$, proving part (1). In particular, we obtain:
$${\bm \eta}_1=\sum_{i=1}^{h^\vee-2}{\bf w}^{-i}\circ X^i,\qquad {\bm \eta}_2={\bf w}^{1-h^\vee}\circ X^{h^\vee-1}.$$ 
To prove part (2), we consider \eqref{vectorzeromon1} and \eqref{vectorzeromon2}. First, note that we have ${\bf a}_{-2}=0$, while
\begin{align*}
&{\bf a}_{-1}={\bf 1}\otimes\pi_{-1}(f),\\
&{\bf a}_0={\bf 1}\otimes\pi_{0}(f)+{\bf w}^{-1}\circ ({\bf 1}\otimes r-A({\bf 1}\otimes\theta^\vee)),\\
&{\bf a}_1=\sum_{i=1}^{h^\vee-2}{\bf w}^{-i-1}\circ A(X^i),\\
&{\bf a}_2=({\bf w}^{1-h^\vee}+\lambda {\bf w}^k)\otimes E_{\theta}+{\bf w}^{-h^\vee}\circ A(X^{h^\vee-1}).
\end{align*}
On the other hand, ${\bf b}_{-2}=0$, ${\bf b}_{-1}=0$, and
\begin{align*}
{\bf b}_0&={\bf w}^{-2}\circ (-{\bf 1}\otimes r+B({\bf 1}\otimes\theta^\vee)),\\
{\bf b}_1&=-\sum_{i=1}^{h^\vee-2}{\bf w}^{-i-2}\circ B(X^i),\\
{\bf b}_2&=((1-h^\vee){\bf w}^{-h^\vee}+k\lambda {\bf w}^{k-1})\otimes E_{\theta}-{\bf w}^{-h^\vee-1}\circ B(X^{h^\vee-1}).
\end{align*}
Plugging the above quantities into \eqref{vectorzeromon1} we obtain
$$\sum_{i=1}^{h^\vee-2}{\bf w}^{-i-1}\circ \left(-[X^{i+1},\pi_{0}(f)]+M(X^i)+\frac{1}{2}\sum_{s=0}^i[X^s,[X^{i+1-s},\pi_{-1}(f)]]\right)=0.$$
Since for each $i=1,\dots,h^\vee-2$ the element in the sum above which is multiplied by ${\bf w}^{-i-1}$ belongs to $\bb{C}^N\otimes \mf{g}^i$, and since each component of  ${\bf w}\in\bb{C}^N$ is  different from zero, we get \eqref{zmnf3}. Now we consider \eqref{vectorzeromon2}, which is linear with respect to $\lambda$. The vanishing of the coefficient of order $1$ in $\lambda$ reads
\beq\label{prooflambda1}
k{\bf w}^{-1}\otimes E_\theta=[{\bf 1 }\otimes E_\theta,{\bf a}_0-[{\bm\eta}_1,{\bf a}_{-1}]],
\eeq
which is equivalent to \eqref{zmnf2}. Note that from \eqref{prooflambda1} it follows that for every
${\bf Q}\in \bb{C}^N\otimes\mf{g}_2=\bb{C}^N\otimes\mf{g}^{h^\vee-1}$ one has
\beq\label{prooflambda12}
[{\bf Q},{\bf a}_0-[{\bm\eta}_1,{\bf a}_{-1}]]=k{\bf w }^{-1}\circ {\bf Q}.
\eeq
We now consider the vanishing of the coefficient of order zero in $\lambda$ in \eqref{vectorzeromon2}. Since the term $2{\bf a}_2-[{\bm\eta}_2,{\bf a}_0]+\frac{1}{3}[{\bm\eta}_1,[{\bm\eta}_1,{\bf a}_0]+2[{\bm\eta}_2,{\bf a}_{-1}]-4{\bf a}_1]$ belongs to $\bb{C}^N\otimes\mf{g}_2$ then using \eqref{prooflambda12} we get that \eqref{vectorzeromon2} can be written as
$2{\bf b}_2+[{\bm\eta}_1,[{\bm\eta}_1,{\bf b}_0]]-2[{\bm\eta}_1,{\bf b}_{1}]-[{\bm\eta}_2,{\bf b}_0]=k{\bf w }^{-1}\circ(2{\bf a}_2-[{\bm\eta}_2,{\bf a}_0]+\frac{1}{3}[{\bm\eta}_1,[{\bm\eta}_1,{\bf a}_0]+2[{\bm\eta}_2,{\bf a}_{-1}]-4{\bf a}_1])$. By a direct computation, the latter identity is shown to be equivalent to \eqref{zmnf4}. Part (2) of the proposition is proved.
\end{proof}

\subsection{Trivial monodromy: system in $\bb{C}^{2N(\hv-1)}$} The trivial monodromy conditions \eqref{zmnf1} for
the operator \eqref{norformguess} are a system of equations in $\bb{C}^N\otimes \mf{n}^+$.
Due to Proposition
\ref{prop:monodromycg}\eqref{xint},
the variables $X^i,i\geq 1$ defined in \eqref{Xi0906} belong to $\bb{C}^N\otimes \mf{t}^i$, where $\mf{t}$ is
the symplectic vector space defined in \eqref{symplecticspace} and $\mf{t}^i=\mf{t}\cap\mf{g}^i$.
As it was already remarked, this implies the the total number of non-trivial variables $\lbrace \mathbf{w},X^i \rbrace$ is $2N(\hv-1)$.
By choosing an explicit basis of $\mf{t}$, we now write the system \eqref{zmnf1} as an equivalent system in $\bb{C}^{2N(\hv-1)}$.

Recall the set $\Theta\subset \mf{h}^\ast$ defined in \eqref{Theta}, and define 
\beq
\Theta^i=\{\alpha\in \Theta\,|\,\Ht(\alpha)=i\}, \quad i=0,\dots,h^\vee-1.
\eeq
Recall the root vectors $\{E_\alpha, \alpha\in\Delta\}$ of $\mf{g}$ defined in $\S$\ref{gbasis}. For $\alpha\in \Theta$, introduce the variables ${\bf x}^\alpha\in\bb{C}^N$, and write $X^i$ as
\beq\label{Xitheta}
X^i=\sum_{\alpha \in \Theta^i}{\bf x}^\alpha\otimes E_{\alpha},\qquad i=0,\dots, h^\vee-1,
\eeq
with ${\bf x}^0=-{\bf 1}$ and $E_0=\theta^\vee$. Note that we always have $X^{h^\vee-1}={\bf x}^\theta\otimes E_{\theta}$.  For $\alpha\in \Theta$ define  $M_\alpha,S_\alpha, Y_\alpha\in \End_{\bb{C}({\bf w})}(\bb{C}^N)$
as the unique linear operators satisfying the relations: 
\begin{subequations}\label{MSYalpha}
\begin{align}
M_\alpha({\bf v})\otimes E_\alpha&=M({\bf v}\otimes E_\alpha),\label{Malpha}\\
S_\alpha({\bf v})\otimes E_\alpha&=S({\bf v}\otimes E_\alpha),\label{Salpha}\\
Y_\alpha({\bf v})\otimes E_\alpha&=Y({\bf v}\otimes E_\alpha),\label{Yalpha}
\end{align}
\end{subequations}
where $M,S,Y\in\End_{\bb{C}({\bf w})}(\bb{C}^N\otimes \mf{g})$ were introduced in  \eqref{MSY}.
Explicitly,  for ${\bf v}\in\bb{C}^N$ we have
\begin{align}
M_\alpha({\bf v})=&A_{\Ht(\alpha)}({\bf v})-A_0({\bf 1})\circ {\bf v}+\alpha(r){\bf v},\label{Malphav}\\
S_\alpha({\bf v})=&2B_{\Ht(\alpha)}({\bf v})+\frac{4}{3}kA_{\Ht(\alpha)}({\bf v})-B_0({\bf 1})\circ {\bf v}\notag\\
& -\frac{k}{3}A_0({\bf 1})\circ{\bf v} +(1+\frac{k}{3})\alpha(r){\bf  v},\\
Y_\alpha({\bf v})=&\frac{2}{3}k^2{\bf v}+\frac{2}{3}k A_{\Ht(\alpha)}({\bf v})+
2B_{\Ht(\alpha)}({\bf v})+(1+\frac{k}{3})\alpha(r){\bf v}.
\end{align}
with $\alpha\in \Theta$.
Recall the bimultiplicative function $\epsilon_{\alpha,\beta}$ ($\alpha,\beta\in Q$) defined in \eqref{bimultiplicative1}, \eqref{bimultiplicative2}.
\begin{proposition}\label{prop:20181128}
Let $\rank\mf{g}>1$, and let $A_0$ be given by \eqref{operatorAB} and $M_\alpha, S_\alpha, Y_\alpha$ by \eqref{MSYalpha}. The trivial monodromy conditions \eqref{zmnf1} are equivalent to the following system of $2h^\vee-2$ ($\bb{C}^N$-valued) equations in the $2h^\vee-2$  ($\bb{C}^N$-valued) variables $\{{\bf w},{\bf x}^\alpha \,|\,\alpha\in \Theta\}$:
\begin{subequations}\label{systemonc}
\beq\label{systemonc1}
\sum_{\alpha \in \Theta^1}{\bf x}^\alpha=(k+\theta(r)){\bf 1}-2A_0({\bf 1}),
\eeq
for every $\alpha\in \Theta^i$, $i=1,\dots,h^\vee-3$:
\begin{align}\label{multirooteq}
\sum_{\substack{j\in I_\theta:\\ \alpha+\alpha_j\in\Theta^{i+1}}}\!\!\!\epsilon_{\alpha,\alpha_j}\,
{\bf x}^{\alpha+\alpha_j}&=M_\alpha({\bf x}^\alpha)-\frac{1}{2}\sum_{\beta \in \Theta^1}(\alpha\vert \beta)\,
{\bf x}^\alpha\circ {\bf x}^{\beta}\notag\\
+&\frac{1}{2}\sum_{s=1}^{i-1}\sum_{\beta\in\Theta^1}\sum_{\substack{ \gamma\in \Theta^s:
\\ \alpha-\gamma\in\Delta\\ \alpha-\gamma+\beta\in \Delta}}\epsilon_{\alpha,\gamma}
\epsilon_{\alpha,\beta}\epsilon_{\gamma,\beta}{\bf x}^\gamma\circ {\bf x}^{\alpha-\gamma+\beta},
\end{align}
for every $\alpha\in \Theta^{h^\vee-2}$:
\begin{align}
\varepsilon_{\theta,\alpha} {\bf x}^\theta&=2M_\alpha({\bf x}^\alpha)-\sum_{\beta \in \Theta^1}(\alpha\vert \beta)
{\bf x}^\alpha\circ {\bf x}^{\beta}\notag\\
&+\sum_{s=1}^{h^\vee-3}\sum_{\beta\in\Theta^1}\sum_{\substack{ \gamma\in \Theta^s\\
\alpha-\gamma\in\Delta\\ \alpha-\gamma+\beta\in\Delta}}\epsilon_{\alpha,\gamma}\epsilon_{\alpha,\beta}
\epsilon_{\gamma,\beta}{\bf x}^\gamma\circ {\bf x}^{\alpha-\gamma+\beta},
\end{align}
and finally
\begin{align}
2(k+h^\vee-1){\bf w}&=Y_{\theta}({\bf x}^\theta)-\sum_{i=1}^{h^\vee-1}\sum_{\alpha\in\Theta^i}(\theta | \alpha)
\epsilon_{\theta,\alpha}{\bf x}^{\theta-\alpha}\circ S_{\alpha}({\bf x}^\alpha)\notag\\
&-\frac{k}{3}\sum_{i=2}^{h^\vee-2}\sum_{j\in I_\theta}\sum_{\substack{\gamma\in \Theta^{i}\\
\theta-\gamma+\alpha_j\in \Delta}}\varepsilon_{\theta,\gamma}\varepsilon_{\theta,\alpha_j}
\varepsilon_{\gamma,\alpha_j}{\bf x}^{\theta-\gamma+\alpha_j}\circ {\bf x}^{\gamma},
\end{align}
\end{subequations}
where ${\bf x}^0=-{\bf 1}$.
\end{proposition}
\begin{proof} System \eqref{systemonc} is obtained substituting \eqref{Xitheta} into \eqref{zmnf1} and using
the commutation relations \eqref{commrel}.
\end{proof}

\begin{remark}\label{rem:N=1}
 In the case $N=1$ the system \eqref{systemonc} greatly simplifies. The variables $\mathbf{x}^{\alpha}$ are scalars
 and the operators $M_\alpha, S_\alpha,Y_\alpha:\bb{C}\to\bb{C}$ are just multiplication operators independent on the variable $w$:
 \begin{align}
M_\alpha({\bf v})=&\left(\alpha(r)-\Ht(\alpha)\right){\bf v},\label{MalphavN=1}\\ \nonumber
S_\alpha({\bf v})=&\left((1+\frac{k}{3})\alpha(r)-(\Ht(\alpha)+1+\frac{4}{3}k) \Ht(\alpha)\right){\bf v}\\ \nonumber
Y_\alpha({\bf v})=&\left(\frac{2}{3}k^2+(1+\frac{k}{3})\alpha(r)-(\Ht(\alpha)+1+\frac{2}{3}k) \Ht(\alpha)\right){\bf v}.
\end{align}
It follows that the system decouples: the first three equations are a subsystem for the $x^{\alpha}$'s alone,
and last equation yields the location of the pole $w$ as an explicit function of the $x^{\alpha}$'s.
\end{remark}

 Let $P_n(N)$ be the number of $n$-coloured partitions of $N$.
 The Fock space of the quantum $\widehat{\mf{g}}$-KdV model 
 is generatd by the action of $n=\rank \mf{g}$ free fields \cite{kojima08,hj12}. Hence the number 
of the states of a given level $N$ of the quantum theory is less 
or equal \footnote{It may be less only if the Fock representation is not irreducible.} than $P_n(N)$
for arbitrary values of the parameters $(r,\hat{k})\in \mf{h}\times(0,1)$ of the model,
and it actually coincides with $P_n(N)$ for generic values of the parameters. According to Conjecture \ref{conj:odeim},
the solutions of \eqref{systemonc} are in bijections with the states of level $N$ of  quantum $\mf{g}-$KdV.
Hence, the ODE/IM conjecture impies the following conjecture on the number of solutions of \eqref{systemonc}.
\begin{conjecture}\label{conj:numbersolutions}
The number of solutions of \eqref{systemonc} is less or equal than $N!P_n(N)$. The set of parameters $(r,\hat{k})$
for which the number of solutions is $N!P_n(N)$ is a generic subset of $\mf{h}\times(0,1)$.
\end{conjecture}

\vspace{20pt}
\section{Explicit computations}\label{sec:computations}
System \eqref{systemonc}, providing trivial monodromy conditions for the operator \eqref{norformguess},
is a system of $2N (\hv-1)$ equations in $2N (\hv-1)$ unknowns, which
depends on the root structure of the algebra. In this section we provide an explicit
presentation of this system in case $\mf{g}=A_n, n \geq 2,D_n, n\geq 4$, and $E_6$.
We omit to show our computations in the case $E_7,E_8$ due to their excessive length\footnote{We can furnish them to the interested
reader upon request.}.
In each case, we are able to reduce \eqref{systemonc} to a system of $2N$ equations in $2N$ unknowns. Moreover,
if $N=1$, we further reduce it to a single degree $n$ polynomial equation in one variable. This is consistent with the ODE/IM
hypothesis -- see Conjecture \ref{conj:numbersolutions} -- since the dimensions of the level $1$ subspace of the quantum $\mf{g}$-KdV model is equal to $\rank{g}=n$,
for generic values of the central charge and of the vacuum parameters.

\subsection{The case $A_n$, $n\geq 2$}
For $\mf{g}$ of type $A_n$ we have $h^\vee=n+1$,  $\theta=\sum_{i\in I}\alpha_i$, and $\dim\mf{t}=2n$.
Since $I\setminus I_\theta=\{1,n\}$ then from \eqref{fprojections} we get
$$\pi_0(f)=-\sum_{i=2}^{n-1}E_{-\alpha_i},\quad \pi_{-1}(f)=-E_{-\alpha_1}-E_{-\alpha_n}.$$ 
Defining
\beq
\beta_j=\sum_{i=n+1-j}^n\alpha_i,\quad \gamma_j=\sum_{i=1}^j\alpha_i, \qquad j=1,\dots,n-1
\eeq
then we get $\Theta=\{0,\theta\}\cup\{\beta_i,\gamma_i,i=1,\dots,n-1\}$ and $\{E_{\alpha},\alpha\in \Theta\}$ is a basis of  $\mf{t}$, with $E_0=\theta^\vee$.  Thus, \eqref{Xitheta} reads
\beq\label{XiAn}
X^i=
\begin{cases}
-{\bf 1}\otimes \theta^\vee & i=0,\\
{\bf x}^{\beta_i}\otimes E_{\beta_i}+{\bf x}^{\gamma_i}\otimes E_{\gamma_i} & i=1,\dots,n-1,\\
{\bf x}^\theta\otimes E_\theta & i=n,
\end{cases}
\eeq
and system \eqref{systemonc} takes the simpler form
\begin{align*}
&{\bf x}^{\beta_1}+{\bf x}^{\gamma_1}=(k+\theta(r)){\bf 1}-2A_0({\bf 1}),\\
&{\bf x}^{\beta_{i+1}}=M_{\beta_i}({\bf x}^{\beta_i})-{\bf x}^{\beta_1}\circ {\bf x}^{\beta_{i}},\hspace{106pt} i=1,\dots,n-2,\\
&{\bf x}^{\gamma_{i+1}}=-M_{\gamma_i}({\bf x}^{\gamma_i})+{\bf x}^{\gamma_1}\circ {\bf x}^{\gamma_{i}},\hspace{100pt} i=1,\dots,n-2,\\
&{\bf x}^\theta=2M_{\beta_{n-1}}({\bf x}^{\beta_{n-1}})+\sum_{s=1}^{n-1}{\bf x}^{\beta_s}\circ {\bf x}^{\gamma_{n-s}}-2{\bf x}^{\beta_1}\circ {\bf x}^{\beta_{n-1}}\\
&-{\bf x}^\theta=2M_{\gamma_{n-1}}({\bf x}^{\beta_{n-1}})+\sum_{s=1}^{n-1}{\bf x}^{\gamma_s}\circ {\bf x}^{\beta_{n-s}}-2{\bf x}^{\gamma_1}\circ {\bf x}^{\gamma_{n-1}}\\
&2(n+k){\bf w}=Y_\theta({\bf x}^\theta)+\sum_{i=1}^{n-1}\left({\bf x}^{\beta_{n-i}}\circ S_{\gamma_i}({\bf x}^{\gamma_i})-{\bf x}^{\gamma_{n-i}}\circ S_{\beta_i}({\bf x}^{\beta_i})\right)\\
&\hspace{50pt}-2S_\theta({\bf x}^\theta)+\frac{2}{3}k\sum_{i=1}^{n-2}{\bf x}^{\beta_{n-i}}\circ {\bf x}^{\gamma_{i+1}}.
\end{align*}
The above system can be further simplified as follows. Introduce the polynomial functions $P_i,\widetilde{P}_i:\bb{C}^N\to\bb{C}^N$, depending on the parameters ${\bf w}\in\bb{C}^N$ and defined recursively by the relations
\begin{subequations}\label{recursionAn}
\begin{align}
P_{i+1}({\bf x})&=M_{\beta_i}(P_i({\bf x}))-{\bf x}\circ P_i({\bf x}),\\
\widetilde{P}_{i+1}({\bf x})&=-M_{\gamma_i}(\widetilde{P}_i({\bf x}))+{\bf x}\circ \widetilde{P}_i({\bf x}),
\end{align}
\end{subequations}
for every ${\bf x}\in\bb{C}^N$, and with $P_0({\bf x})=-{\bf 1},$ $\widetilde{P}_0({\bf x})={\bf 1},$ and $\beta_0=\gamma_0=0$. 
\begin{proposition}\label{prop:An}
Let $\mf{g}$ be of type $A_n$, $n\geq 2$. The operator \eqref{norformguess} with $X^i$ as in \eqref{XiAn} has trivial monodromy at
all $w_{\ell}$, $\ell=1,\dots,N$ for all values of $\lambda$ if and only if:

\begin{enumerate}
 \item  The variables ${\bf x}^{\beta_1},{\bf x}^{\gamma_1},{\bf w} \in\bb{C}^N$ satisfy the system
\begin{subequations}\label{Ansys}
\begin{align}
&{\bf x}^{\beta_1}+{\bf x}^{\gamma_1}=(k+\theta(r)){\bf 1}-2A_0({\bf 1}),&\\
&\sum_{s=0}^n P_s({\bf x}^{\beta_1})\circ \widetilde{P}_{n-s}({\bf x}^{\gamma_1})=0 \\
\label{Anpoles}
2(n+k){\bf w}&=Y_\theta(P_n({\bf x}^{\beta_1})+\widetilde{P}_n({\bf x}^{\gamma_1}))+
\frac{2}{3}k\sum_{i=1}^{n-2}P_{n-i}({\bf x}^{\beta_{1}})\circ \widetilde{P}_{i+1}({\bf x}^{\gamma_{1}})\notag\\
&+\sum_{i=1}^{n}\left(P_{n-i}({\bf x}^{\beta_{1}})\circ S_{\gamma_i}(\widetilde{P}_i({\bf x}^{\gamma_1}))
-\widetilde{P}_{n-i}({\bf x}^{\gamma_{1}})\circ S_{\beta_i}(P_i({\bf x}^{\beta_1}))\right),
\end{align}
\end{subequations}
where $\gamma_n=\beta_n=\theta$.
\item The variables ${\bf x}^\alpha\in\bb{C}^N$, $\alpha\in \Theta$, are given in terms of ${\bf x}^{\beta_1},{\bf x}^{\gamma_1}$ as
\begin{align*}
& {\bf x}^{\beta_i}=P_i({\bf x}^{\beta_1})& i=1,\dots,n-1,\\
& {\bf x}^{\gamma_i}=\widetilde{P}_i({\bf x}^{\gamma_1})& i=1,\dots,n-1,\\
& {\bf x}^\theta=P_n({\bf x}^{\beta_1})+\widetilde{P}_n({\bf x}^{\gamma_1}),
\end{align*}
\end{enumerate}

\end{proposition}

\begin{corollary} Let $N=1$. The system \eqref{Ansys} admits, for generic values of $r \in \mf{h},k \in \bb{R}$,
$n$ solutions.
\end{corollary}
\begin{proof}
For $N=1$, the recursion relations \eqref{recursionAn} can be explicitely solved.
Indeed, in this case we have ${\bf x}^\alpha=x^\alpha\in\bb{C}$ and the operator \eqref{Malphav}
reduces to the scalar operator $M_\alpha(x)=(\alpha(r)-\Ht(\alpha))x$, with $x\in\bb{C}$ and $\alpha\in\Theta$.
Noting that $\Ht(\beta_i)=\Ht(\gamma_i)=i$, then the polynomials
\begin{subequations}\label{PPAnN1}
\begin{align}
P_i(x)&=(-1)^{i+1}\prod_{j=1}^{i}(x-\beta_{j-1}(r)+j-1),\\
\widetilde{P}_i(x)&=\prod_{j=1}^{i}(x-\gamma_{j-1}(r)+j-1)
\end{align}
\end{subequations}
satisfy \eqref{recursionAn}  with $P_0(x)=-1$ and $\widetilde{P}_0(x)=1$ (and $\beta_0=\gamma_0=0$).
Since the polynomials $P,\widetilde{P}$ do not depend on the pole $w_1$, the system \eqref{Ansys}
splits into a subsystem for $x^{\beta_1},x^{\gamma_1}$ and a linear equation for ${\bf w}=w_1$, the additional pole,.
Explicitly we have:
\begin{align*}
&x^{\beta_1}+x^{\gamma_1}=k+\theta(r),&\\
&\sum_{s=0}^n (-1)^{s+1}\prod_{j=1}^{s}(x^{\beta_1}-\beta_{j-1}(r)+j-1) \prod_{j=1}^{n-s}(x^{\gamma_1}-\gamma_{j-1}(r)+j-1) =0.
\end{align*}
Substituting the first equation in the second, one obtain a polynomial equation for the variable $x^{\beta_1}$,
which has -- for generic values of $r$ and $k$ -- $n$ distinct solutions.
Once a solution of the above system is chosen, the additional pole is given by
\begin{align*}
2&(n+k)w_1=\frac{2}{3}k(k+n)\left((-1)^{n+1}\prod_{j=1}^{n}(x^{\beta_1}-\beta_{j-1}(r)+j-1)+\prod_{j=1}^{n}(x^{\gamma_1}-\gamma_{j-1}(r)+j-1)\right)\\
&+\frac{2}{3}k\sum_{i=1}^{n-2}\left((-1)^{n-i+1}\prod_{j=1}^{n-i}(x^{\beta_1}-\beta_{j-1}(r)+j-1)\prod_{j=1}^{i+1}(x^{\gamma_1}-\gamma_{j-1}(r)+j-1)\right)\\
&+\sum_{i=1}^{n-1}m_{n,i,r,k}\left((-1)^{n-i+1}\prod_{j=1}^{n-i}(x^{\beta_1}-\beta_{j-1}(r)+j-1)\prod_{j=1}^{i}(x^{\gamma_1}-\gamma_{j-1}(r)+j-1)\right),
\end{align*}
with $m_{n,i,r,k}=(1+\frac{k}{3})(\gamma_i(r)-\beta_{n-i}(r))+(n+1+\frac{4}{3}k)(n-2i)$, 
and
\begin{align*}
& x^{\beta_i}=(-1)^{i+1}\prod_{j=1}^{i}(x^{\beta_1}-\beta_{j-1}(r)+j-1),\qquad  i=1,\dots,n-1,\\
& x^{\gamma_i}=\prod_{j=1}^{i}(x^{\gamma_1}-\gamma_{j-1}(r)+j-1),\qquad i=1,\dots,n-1,\\
& x^\theta=(-1)^{n+1}\prod_{j=1}^{n}(x^{\beta_1}-\beta_{j-1}(r)+j-1)+\prod_{j=1}^{n}(x^{\gamma_1}-\gamma_{j-1}(r)+j-1).
\end{align*}
\end{proof}
\subsection{The case $D_n$}
For $\mf{g}$ of type $D_n$ we have $h^\vee=2n-2$,  $\theta=\alpha_1+2\sum_{i=2}^{n-2}\alpha_i+\alpha_{n-1}+\alpha_n$, and $\dim\mf{t}=4n-6$. Since $I\setminus I_\theta=\{2\}$ then from \eqref{fprojections} we get
$$\pi_0(f)=-E_{-\alpha_1}-\sum_{i=3}^{n}E_{-\alpha_i},\quad \pi_{-1}(f)=-E_{-\alpha_2}.$$
Denoting the roots
\begin{align*}
&\beta_j=
\begin{cases}
\sum_{i=1}^{j+1}\alpha_i,& j=1,\dots,n-3,\\
\sum_{i=1}^{2n-j-3}\alpha_i+2\sum_{i=2n-j-2}^{n-2}\alpha_i+\alpha_{n-1}+\alpha_n,& j=n-1,\dots,2n-5,\\
\end{cases}
\\
&\beta_{n-2}^+=\sum_{i=2}^{n-1}\alpha_i,\qquad \beta_{n-2}^-=\sum_{i=2}^{n-2}\alpha_i+\alpha_n,
\end{align*}
and 
\begin{align*}
&\gamma_j=
\begin{cases}
\sum_{i=2}^{j+1}\alpha_i,& j=1,\dots,n-3,\\
\sum_{i=2}^{2n-j-3}\alpha_i+2\sum_{i=2n-j-2}^{n-2}\alpha_i+\alpha_{n-1}+\alpha_n,& j=n-1,\dots,2n-5,\\
\end{cases}
\\
&\gamma_{n-2}^+=\sum_{i=1}^{n-2}\alpha_i+\alpha_n,\qquad \gamma_{n-2}^-=\sum_{i=1}^{n-1}\alpha_i,
\end{align*}
then we have $\Theta=\{0,\beta^\pm_{n-2},\gamma^\pm_{n-2},\theta\}\cup\{\beta_j,\gamma_j,j=1,\dots,\widehat{n-2},\dots,2n-5\}$. 
Thus, \eqref{Xitheta} reads
\beq\label{XiDn}
X^i=
\begin{cases}
-{\bf 1}\otimes \theta^\vee & i=0,\\
{\bf x}^{\gamma_1}\otimes E_{\gamma_1} & i=1,\\
{\bf x}^{\beta_{i-1}}\otimes E_{\beta_{i-1}}+{\bf x}^{\gamma_i}\otimes E_{\gamma_i} & i=2,\dots,n-3,\\
{\bf x}^{\beta_{n-3}}\otimes E_{\beta_{n-3}}+{\bf x}^{\beta^+_{n-2}}\otimes E_{\beta^+_{n-2}}+{\bf x}^{\beta^-_{n-2}}\otimes E_{\beta^-_{n-2}}& i=n-2\\
{\bf x}^{\gamma_{n-1}}\otimes E_{\gamma_{n-1}}+{\bf x}^{\gamma^+_{n-2}}\otimes E_{\gamma^+_{n-2}}+{\bf x}^{\gamma^-_{n-2}}\otimes E_{\gamma^-_{n-2}}& i=n-1\\
{\bf x}^{\beta_{i-1}}\otimes E_{\beta_{i-1}}+{\bf x}^{\gamma_i}\otimes E_{\gamma_i} & i=n,\dots,2n-5,\\
{\bf x}^{\beta_{2n-5}}\otimes E_{\beta_{2n-5}} & i=2n-4,\\
{\bf x}^\theta\otimes E_\theta & i=2n-3.
\end{cases}
\eeq
and system \eqref{systemonc} takes the form
\begin{align*}
&{\bf x}^{\gamma_1}=(k+\theta(r)){\bf 1}-2A_0({\bf 1}),&\\
&{\bf x}^{\beta_{i}}-{\bf x}^{\gamma_{i+1}}=M_{\gamma_i}({\bf x}^{\gamma_i})-{\bf x}^{\gamma_1}\circ {\bf x}^{\gamma_{i}},&\hspace{-250pt} i=1,\dots,n-4,\\
&-{\bf x}^{\beta_{i}}=M_{\beta_{i-1}}({\bf x}^{\beta_{i-1}})-{\bf x}^{\beta_1}\circ {\bf x}^{\gamma_{i-1}},&\hspace{-250pt} i=2,\dots,n-3,\\
&{\bf x}^{\beta_{n-3}}-{\bf x}^{\beta^+_{n-2}}-{\bf x}^{\beta^-_{n-2}}=M_{\gamma_{n-3}}({\bf x}^{\gamma_{n-3}})-{\bf x}^{\gamma_1}\circ {\bf x}^{\gamma_{n-3}}\\
& {\bf x}^{\gamma^+_{n-2}}-{\bf x}^{\gamma_{n-1}}=M_{\beta^-_{n-2}}({\bf x}^{\beta^-_{n-2}})-{\bf x}^{\gamma_1}\circ {\bf x}^{\beta^-_{n-2}}\\
&  {\bf x}^{\gamma^-_{n-2}}-{\bf x}^{\gamma_{n-1}}=M_{\beta^+_{n-2}}({\bf x}^{\beta^+_{n-2}})-{\bf x}^{\gamma_1}\circ {\bf x}^{\beta^+_{n-2}}\\
&-{\bf x}^{\gamma^+_{n-2}}-{\bf x}^{\gamma^-_{n-2}}=M_{\beta_{n-3}}({\bf x}^{\beta_{n-3}})-{\bf x}^{\beta_1}\circ {\bf x}^{\gamma_{n-3}}\\
& -{\bf x}^{\beta_{n-1}}=M_{\gamma^+_{n-2}}({\bf x}^{\gamma^+_{n-2}})-{\bf x}^{\beta_1}\circ {\bf x}^{\beta^-_{n-2}}\\
&  -{\bf x}^{\beta_{n-1}}=M_{\gamma^-_{n-2}}({\bf x}^{\gamma^-_{n-2}})-{\bf x}^{\beta_1}\circ {\bf x}^{\beta^+_{n-2}}\\
&{\bf x}^{\beta_{i}}+{\bf x}^{\gamma_{i+1}}=M_{\gamma_i}({\bf x}^{\gamma_i})-{\bf x}^{\gamma_1}\circ {\bf x}^{\gamma_{i}},&\hspace{-250pt} i=n-1,\dots,2n-6,\\
&{\bf x}^{\beta_{i}}=M_{\beta_{i-1}}({\bf x}^{\beta_{i-1}})-{\bf x}^{\beta_1}\circ {\bf x}^{\gamma_{i-1}},&\hspace{-250pt} i=n,\dots,2n-5,\\
&{\bf x}^{\beta_{2n-5}}=M_{\gamma_{2n-5}}({\bf x}^{\gamma_{2n-5}})-{\bf x}^{\beta^+_{n-2}}\circ {\bf x}^{\beta^-_{n-2}}+\sum_{s=1}^{n-4}{\bf x}^{\gamma_{s+1}}\circ {\bf x}^{\gamma_{2n-5-s}},&\\
&{\bf x}^\theta=2M_{\beta_{2n-5}}({\bf x}^{\beta_{2n-5}})+\sum_{s=2}^{n-3}{\bf x}^{\beta_s}\circ {\bf x}^{\gamma_{2n-4-s}}+\sum_{s=n-1}^{2n-5}{\bf x}^{\beta_s}\circ {\bf x}^{\gamma_{2n-4-s}}\\
&\hspace{20pt}-{\bf x}^{\beta_1}\circ {\bf x}^{\gamma_{2n-5}}-{\bf x}^{\beta^+_{n-2}}\circ {\bf x}^{\gamma^+_{n-2}}-{\bf x}^{\beta^-_{n-2}}\circ {\bf x}^{\gamma^-_{n-2}},&\\
&2(k+2n-3){\bf w}=Y_\theta({\bf x}^\theta)+{\bf x}^{\beta_{n-2}^+}\circ S_{\gamma_{n-2}^+}({\bf x}^{\gamma_{n-2}^+})-{\bf x}^{\gamma_{n-2}^+}\circ S_{\beta_{n-2}^+}({\bf x}^{\beta_{n-2}^+})&\\
&\hspace{20pt}+{\bf x}^{\beta_{n-2}^-}\circ S_{\gamma_{n-2}^-}({\bf x}^{\gamma_{n-2}^-})-{\bf x}^{\gamma_{n-2}^-}\circ S_{\beta_{n-2}^-}({\bf x}^{\beta_{n-2}^-})+\frac{2}{3}k{\bf x}^{\gamma_{n-2}^+}\circ {\bf x}^{\gamma_{n-2}^-}&\\
&\hspace{20pt}+\sum_{\substack{i=1,\dots,2n-5\\i\neq n-2}}\left({\bf x}^{\beta_{2n-4-i}}\circ \left(S_{\gamma_i}({\bf x}^{\gamma_i})-\frac{k}{3} {\bf x}^{\beta_i}\right)-{\bf x}^{\gamma_{2n-4-i}}\circ S_{\beta_i}({\bf x}^{\beta_i})\right)&\\
&\hspace{20pt}+\frac{2}{3}k({\bf x}^{\beta_{n-2}^+}+{\bf x}^{\beta_{n-2}^-})\circ {\bf x}^{\beta_{n-1}}-\frac{2}{3}k({\bf x}^{\gamma_{n-2}^+}+{\bf x}^{\gamma_{n-2}^-})\circ {\bf x}^{\gamma_{n-1}}&\\
&\hspace{20pt} -\frac{2}{3}k\sum_{i=2}^{n-3}({\bf x}^{\gamma_{2n-3-i}}\circ {\bf x}^{\beta_{i}}-{\bf x}^{\beta_{2n-3-i}}\circ {\bf x}^{\gamma_{i}}).&
\end{align*}
The above system can be simplified as follows. For $\alpha\in \Theta$ and ${\bf v}\in \bb{C}^N$, denote by  $\widetilde{M}_\alpha$  the operators
\beq
\widetilde{M}_\alpha({\bf v})=M_\alpha({\bf v})-(k+\theta(r)){\bf v}+2A_0({\bf 1})\circ {\bf v},
\eeq
and introduce the polynomials $P_i,\widetilde{P}_i:\bb{C}^N\to\bb{C}^N$ defined by the recursion relations
\beq
\begin{cases}\label{DnrecursionP}
\widetilde{P}_{i+1}({\bf x})=P_i({\bf x})-\widetilde{M}_{\gamma_i}(\widetilde{P}_i({\bf x}))\\
P_{i+1}({\bf x})={\bf x}\circ \widetilde{P}_i({\bf x})-M_{\beta_i}(P_i({\bf x})),
\end{cases}
\qquad i\geq 0,
\eeq
with $P_0({\bf x})=0$, $\widetilde{P}_0({\bf x})={\bf 1}$ and $\gamma_0=\beta_0=0$. (In particular, from \eqref{DnrecursionP} we obtain $\widetilde{P}_1({\bf x})=(k+\theta(r)){\bf 1}-2A_0(\bf{1})$ and $P_1({\bf x})={\bf x}$.) Set 
$$\widetilde{M}_\pm=M_{\beta^\pm_{n-2}},\qquad M_\pm=M_{\gamma^\pm_{n-2}},$$ 
and consider the rational functions $R_\pm:\bb{C}^N\to\bb{C}^N$, depending parametrically on ${\bf w}\in\bb{C}^N$ and given by:
\begin{align*}
R_\pm({\bf x})&=[(M_++M_-)(\widetilde{M}_++\widetilde{M}_-)-4\,{\bf x}\,\circ]^{-1}[(M_++M_-)\widetilde{M}_\mp(\widetilde{P}_{n-2}({\bf x}))&\\
&\hspace{130pt}-2{\bf x}\circ \widetilde{P}_{n-2}({\bf x})\pm (M_++M_-)P_{n-2}({\bf x})].&\\
\widetilde{R}_{\pm}({\bf x})&=\frac{1}{2}P_{n-2}({\bf x})+\frac{1}{2}\widetilde{M}_\mp(R_{\mp}({\bf x}))-\frac{1}{2}M_\pm(R_{\pm}({\bf x}))&,
\end{align*}
with ${\bf x}\in\bb{C}^N$. Finally, introduce recursively the polynomials $J,\widetilde{J}:\bb{C}^N\to\bb{C}^N$ as
\beq
\begin{cases}\label{DnrecursionJ}
\widetilde{J}_{i+1}({\bf x})=-J_i({\bf x})+\widetilde{M}_{\gamma_i}(\widetilde{J}_i({\bf x}))\\
J_{i+1}({\bf x})=M_{\beta_i}(J_i({\bf x}))-{\bf x}\circ \widetilde{J}_i({\bf x}),
\end{cases}
\qquad i\geq n-1,
\eeq
with 
\begin{align*}
\widetilde{J}_{n-1}({\bf x})&=\frac{1}{2}P_{n-2}({\bf x})-\frac{1}{2}\widetilde{M}_+(R_+({\bf x}))-\frac{1}{2}\widetilde{M}_-(R_-({\bf x}))\\
J_{n-1}({\bf x})&=\frac{1}{2}{\bf x}\circ (R_+({\bf x})+R_-({\bf x}))-\frac{1}{2}M_+(\widetilde{R}_+({\bf x}))-\frac{1}{2}M_-(
\widetilde{R}_-({\bf x})),
\end{align*}
and define $K:\bb{C}^N\to \bb{C}^N$ as:
\begin{align*}
K({\bf x})=&2M_{\beta_{2n-5}}(J_{2n-5}({\bf x}))-R_+({\bf x})\circ \widetilde{R}_+({\bf x})-R_-({\bf x})\circ \widetilde{R}_-({\bf x})\\
&\hspace{20pt}+\sum_{s=1}^{n-3}(J_{2n-4-s}({\bf x})\circ \widetilde{P}_s({\bf x})+P_s({\bf x})\circ \widetilde{J}_{2n-4-s}({\bf x})).
\end{align*}

\begin{proposition}
Let $\mf{g}$ be of type $D_n$, $n\geq 4$. The operator \eqref{norformguess} with $X^i$ as in \eqref{XiDn} has trivial monodromy at
all $w_{\ell}$, $\ell=1,\dots,N$ for all values of $\lambda$ if and only if:

\begin{enumerate}
 \item 
The variables ${\bf x}^{\beta_1},\bf{w} \in\bb{C}^N$ satisfy the system
\begin{subequations} \label{Dnsys}
 \begin{align}\label{eq:Dnsys1}
&\sum_{s=0}^{n-3}\widetilde{P}_s({\bf x}^{\beta_1})\circ \widetilde{J}_{2n-4-s}({\bf x}^{\beta_1})=
R_+({\bf x}^{\beta_1})\circ R_-({\bf x}^{\beta_1}) & \\ \nonumber
&2(k+2n-3){\bf w}=Y_\theta(K({\bf x}^{\beta_1}))+R_+({\bf x}^{\beta_1})
\circ S_{\gamma_{n-2}^+}(\widetilde{R}_+({\bf x}^{\beta_1}))&\\\nonumber
&\hspace{0pt}-\widetilde{R}_+({\bf x}^{\beta_1})\circ S_{\beta_{n-2}^+}(R_+({\bf x}^{\beta_1}))+
R_-({\bf x}^{\beta_1})\circ S_{\gamma_{n-2}^-}(\widetilde{R}_-({\bf x}^{\beta_1}))&\\ \nonumber
&\hspace{0pt}-\widetilde{R}_-({\bf x}^{\beta_1})\circ S_{\beta_{n-2}^-}(R_-({\bf x}^{\beta_1}))
+\frac{2}{3}k\widetilde{R}_+({\bf x}^{\beta_1})\circ \widetilde{R}_-({\bf x}^{\beta_1})&\\ \nonumber
&\hspace{0pt}+\sum_{i=1}^{n-3}\left(J_{2n-4-i}({\bf x}^{\beta_1})\circ 
\left(S_{\gamma_i}(\widetilde{P}_{i}({\bf x}^{\beta_1}))-\frac{k}{3} P_{i}({\bf x}^{\beta_1})\right)
-\widetilde{J}_{2n-4-i}({\bf x}^{\beta_1})\circ S_{\beta_i}(P_{i}({\bf x}^{\beta_1}))\right)&\\ \nonumber
&\hspace{0pt}+\sum_{i=n-1}^{2n-5}\left(P_{2n-4-i}({\bf x}^{\beta_1})\circ 
\left(S_{\gamma_i}(\widetilde{J}_{i}({\bf x}^{\beta_1}))-\frac{k}{3} J_{i}({\bf x}^{\beta_1})\right)-
\widetilde{P}_{2n-4-i}({\bf x}^{\beta_1})\circ S_{\beta_i}(J_{i}({\bf x}^{\beta_1})\right)&\\ \nonumber
&\hspace{0pt}+\frac{2}{3}k(R_+({\bf x}^{\beta_1})+R_-({\bf x}^{\beta_1}))\circ J_{n-1}({\bf x}^{\beta_1})-
\frac{2}{3}k(\widetilde{R}_+({\bf x}^{\beta_1})+\widetilde{R}_-({\bf x}^{\beta_1}))\circ \widetilde{J}_{n-1}({\bf x}^{\beta_1})&\\
&\hspace{0pt} -\frac{2}{3}k\sum_{i=2}^{n-3}\left(\widetilde{J}_{2n-3-i}({\bf x}^{\beta_1})
\circ P_i({\bf x}^{\beta_1})-J_{2n-3-i}({\bf x}^{\beta_1})\circ \widetilde{P}_i({\bf x}^{\beta_1})\right) \; ,&
 \end{align}
\end{subequations}
 \item The variables ${\bf x}^\alpha\in\bb{C}^N$, $\alpha\in \Theta$, are given in terms of ${\bf x}^{\beta_1}$ as
\begin{align*}
& {\bf x}^{\gamma_i}=\widetilde{P}_i({\bf x}^{\beta_1})\hspace{100pt}i=1,\dots,n-3,\\
& {\bf x}^{\beta_i}=P_i({\bf x}^{\beta_1}) \hspace{100pt} i=1,\dots,n-3,\\
&{\bf x}^{\beta^\pm_{n-2}}=R_\pm({\bf x}^{\beta_1}),\\
&{\bf x}^{\gamma^\pm_{n-2}}=\widetilde{R}_\pm({\bf x}^{\beta_1}),\\
& {\bf x}^{\gamma_i}=\widetilde{J}_i({\bf x}^{\beta_1})\hspace{100pt} i=n-1,\dots,2n-5,\\
& {\bf x}^{\beta_i}=J_i({\bf x}^{\beta_1})\hspace{100pt} i=n-1,\dots,2n-5,\\
& {\bf x}^\theta=K({\bf x}^{\beta_1}).
\end{align*}
\end{enumerate}
\end{proposition}

\begin{corollary} Let $N=1$. The system \eqref{Dnsys} admits, for generic values of $r \in \mf{h},k \in \bb{R}$,
$n$ solutions.
\begin{proof}Skecth of the proof.
 The system \eqref{Dnsys} splits into
 an algebraic equation for the sole variable ${\bf{x}^{\beta_1}}=x^{\beta_1}$, and a linear equation for ${\bf{w}}=w_1$.  By recursively computing  the degree of the functions $\tilde{P}_s,\tilde{J},R_{\pm}$,
 one shows that \eqref{eq:Dnsys1} is an equation of the form 
 $(x-a)^{-1}\Pi_n(x^{\beta_1})=0$, where $a$ is a complex number and $\Pi_n$ a polynomial of degree $n$.
 The coefficients of $\Pi_n$, as well as the number $a$, depend on $r,k$, so that for generic values of these parameter
 the equation has exactly $n$ solutions.
\end{proof}
\end{corollary}

\subsection{The case $E_6$} For $\mf{g}$ of type $E_6$ we have $h^\vee=12$,  $\theta=\alpha_1+2\alpha_2+3\alpha_3+2\alpha_4+2\alpha_5+\alpha_6$, and $\dim\mf{t}=22$. Since $I\setminus I_\theta=\{4\}$ then from \eqref{fprojections} we get
$$\pi_0(f)=-E_{-\alpha_1}-E_{-\alpha_2}-E_{-\alpha_3}-E_{-\alpha_5}-E_{-\alpha_6},\quad \pi_{-1}(f)=-E_{-\alpha_4}.$$
Denoting the roots:
\begin{align*}
&\gamma_1=\alpha_4&& \beta_1=\alpha_3+\alpha_4\\
&\gamma_2=\alpha_2+\alpha_3+\alpha_4&& \beta_2=\alpha_1+\alpha_2+\alpha_3+\alpha_4\\
& \gamma_3=\alpha_3+\alpha_4+\alpha_5&& \beta_3=\alpha_2+\alpha_3+\alpha_4+\alpha_5&&\\
&\gamma_4=\alpha_1+\alpha_2+\alpha_3+\alpha_4+\alpha_5&&  \beta_4=\alpha_3+\alpha_4+\alpha_5+\alpha_6\\
& \gamma_5=\alpha_2+\alpha_3+\alpha_4+\alpha_5+\alpha_6&&  \beta_5=\alpha_1+\alpha_2+\alpha_3+\alpha_4+\alpha_5+\alpha_6\\
& \gamma_6=\alpha_2+2\alpha_3+\alpha_4+\alpha_5 && \beta_6=\alpha_1+\alpha_2+2\alpha_3+\alpha_4+\alpha_5 \\
& \gamma_7=\alpha_1+2\alpha_2+2\alpha_3+\alpha_4+\alpha_5 && \beta_7=\alpha_2+2\alpha_3+\alpha_4+\alpha_5+\alpha_6\\
& \gamma_8=\alpha_1+\alpha_2+2\alpha_3+\alpha_4+\alpha_5+\alpha_6 &&  \beta_8=\alpha_1+2\alpha_2+2\alpha_3+\alpha_4+\alpha_5+\alpha_6 \\
&\gamma_9=\alpha_2+2\alpha_3+\alpha_4+2\alpha_5+\alpha_6 &&\beta_9=\alpha_1+\alpha_2+2\alpha_3+\alpha_4+2\alpha_5+\alpha_6\\
&\gamma_{10}=\alpha_1+2\alpha_2+2\alpha_3+\alpha_4+2\alpha_5+\alpha_6 &&\beta_{10}=\alpha_1+2\alpha_2+3\alpha_3+\alpha_4+2\alpha_5+\alpha_6
\end{align*}
then we have $\Theta=\{0,\theta\}\cup\{\beta_i,\gamma_i: i=1,\dots,10\}$. Thus, \eqref{Xitheta} reads
\beq\label{XiE6}
X^i=
\begin{cases}
-{\bf 1}\otimes \theta^\vee & i=0,\\
{\bf x}^{\gamma_1}\otimes E_{\gamma_1} & i=1,\\
{\bf x}^{\beta_1}\otimes E_{\beta_1} & i=2,\\
{\bf x}^{\gamma_2}\otimes E_{\gamma_2}+{\bf x}^{\gamma_3}\otimes E_{\gamma_3} & i=3,\\
{\bf x}^{\beta_2}\otimes E_{\beta_2}+{\bf x}^{\beta_3}\otimes E_{\beta_3}+{\bf x}^{\beta_4}\otimes E_{\beta_4} & i=4,\\
{\bf x}^{\gamma_4}\otimes E_{\gamma_4}+{\bf x}^{\gamma_5}\otimes E_{\gamma_5}+{\bf x}^{\gamma_6}\otimes E_{\gamma_6}& i=5,\\
{\bf x}^{\beta_5}\otimes E_{\beta_5}+{\bf x}^{\beta_6}\otimes E_{\beta_6}+{\bf x}^{\beta_7}\otimes E_{\beta_7}& i=6,\\
{\bf x}^{\gamma_7}\otimes E_{\gamma_7}+{\bf x}^{\gamma_8}\otimes E_{\gamma_8}+{\bf x}^{\gamma_9}\otimes E_{\gamma_9} & i=7,\\
{\bf x}^{\beta_8}\otimes E_{\beta_8}+{\bf x}^{\beta_9}\otimes E_{\beta_9} & i=8,\\
{\bf x}^{\gamma_{10}}\otimes E_{\gamma_{10}} & i=9,\\
{\bf x}^{\beta_{10}}\otimes E_{\beta_{10}} & i=10,\\
{\bf x}^\theta\otimes E_\theta & i=11.
\end{cases}
\eeq
and introducing the operator
$$\widetilde{M}_\alpha({\bf v})=M_\alpha({\bf v})-{\bf x}^{\gamma_1}\circ {\bf v},$$
for $\alpha\in\Theta$ and ${\bf v}\in\bb{C}^N$, then system \eqref{systemonc} takes the form
\begin{align*}
{\bf x}^{\gamma_1}&=(k+\theta(r)){\bf 1}-2A_0({\bf 1}),\\
 {\bf x}^{\beta_1}&=\widetilde{M}_{\gamma_1}({\bf x}^{\gamma_1}),\\
{\bf x}^{\gamma_2}-{\bf x}^{\gamma_3}&=\widetilde{M}_{\beta_1}({\bf x}^{\beta_1}),\\
{\bf x}^{\beta_2}-{\bf x}^{\beta_3}&=\widetilde{M}_{\gamma_2}({\bf x}^{\gamma_2}),\\
{\bf x}^{\beta_4}-{\bf x}^{\beta_3}&=\widetilde{M}_{\gamma_3}({\bf x}^{\gamma_3}),\\ 
-{\bf x}^{\gamma_4}&=\widetilde{M}_{\beta_2}({\bf x}^{\beta_2}),\\ 
{\bf x}^{\gamma_5}&=\widetilde{M}_{\beta_4}({\bf x}^{\beta_4}),\\ 
{\bf x}^{\gamma_6}&=\widetilde{M}_{\beta_2}({\bf x}^{\beta_2})+\widetilde{M}_{\beta_3}({\bf x}^{\beta_3})+\widetilde{M}_{\beta_4}({\bf x}^{\beta_4}),\\
2{\bf x}^{\beta_5}&=M_{\gamma_6}({\bf x}^{\gamma_6})-\widetilde{M}_{\gamma_4}({\bf x}^{\gamma_4})+\widetilde{M}_{\gamma_5}({\bf x}^{\gamma_5})+{\bf x}^{\gamma_2}\circ {\bf x}^{\gamma_3}-{\bf x}^{\beta_1}\circ {\bf x}^{\beta_3},\\
2{\bf x}^{\beta_6}&=M_{\gamma_6}({\bf x}^{\gamma_6})+\widetilde{M}_{\gamma_4}({\bf x}^{\gamma_4})+\widetilde{M}_{\gamma_5}({\bf x}^{\gamma_5})+{\bf x}^{\gamma_2}\circ {\bf x}^{\gamma_3}-{\bf x}^{\beta_1}\circ {\bf x}^{\beta_3},\\
2{\bf x}^{\beta_7}&=-M_{\gamma_6}({\bf x}^{\gamma_6})+\widetilde{M}_{\gamma_4}({\bf x}^{\gamma_4})+\widetilde{M}_{\gamma_5}({\bf x}^{\gamma_5})-{\bf x}^{\gamma_2}\circ {\bf x}^{\gamma_3}+{\bf x}^{\beta_1}\circ {\bf x}^{\beta_3},\\
{\bf x}^{\gamma_8}&=\widetilde{M}_{\beta_5}({\bf x}^{\beta_5}),
\end{align*}
\begin{align*}
{\bf x}^{\gamma_7}&=\widetilde{M}_{\beta_5}({\bf x}^{\beta_5})+M_{\beta_6}({\bf x}^{\beta_6})+{\bf x}^{\beta_2}\circ {\bf x}^{\gamma_3}-{\bf x}^{\beta_1}\circ {\bf x}^{\gamma_4},\\
{\bf x}^{\gamma_9}&=\widetilde{M}_{\beta_5}({\bf x}^{\beta_5})-M_{\beta_7}({\bf x}^{\beta_7})-{\bf x}^{\gamma_2}\circ {\bf x}^{\beta_4}+{\bf x}^{\gamma_1}\circ {\bf x}^{\beta_5},\\
{\bf x}^{\beta_8}&=-M_{\gamma_7}({\bf x}^{\gamma_7})+{\bf x}^{\gamma_2}\circ {\bf x}^{\gamma_4}-{\bf x}^{\beta_2}\circ {\bf x}^{\beta_3},\\
{\bf x}^{\beta_9}&=M_{\gamma_9}({\bf x}^{\gamma_9})+{\bf x}^{\beta_3}\circ {\bf x}^{\beta_4}-{\bf x}^{\gamma_3}\circ {\bf x}^{\gamma_5},\\
M_{\gamma_9}({\bf x}^{\gamma_9})&={\bf x}^{\gamma_2}\circ {\bf x}^{\gamma_4}+{\bf x}^{\gamma_3}\circ {\bf x}^{\gamma_5}-{\bf x}^{\beta_2}\circ {\bf x}^{\beta_3}-{\bf x}^{\beta_3}\circ {\bf x}^{\beta_4}-{\bf x}^{\beta_2}\circ {\bf x}^{\beta_4}\\
&\quad -{\bf x}^{\beta_1}\circ {\bf x}^{\beta_5}-M_{\gamma_7}({\bf x}^{\gamma_7})-M_{\gamma_8}({\bf x}^{\gamma_8}),\\
2{\bf x}^{\gamma_{10}}&=M_{\beta_9}({\bf x}^{\beta_9})-M_{\beta_8}({\bf x}^{\beta_8})+{\bf x}^{\beta_4}\circ {\bf x}^{\gamma_4}-{\bf x}^{\beta_5}\circ {\bf x}^{\gamma_3}-{\bf x}^{\beta_2}\circ {\bf x}^{\gamma_5}+{\bf x}^{\gamma_2}\circ {\bf x}^{\beta_5},\\
M_{\beta_8}({\bf x}^{\beta_8})&=-M_{\beta_9}({\bf x}^{\beta_9})-{\bf x}^{\beta_4}\circ {\bf x}^{\gamma_4}+{\bf x}^{\beta_5}\circ {\bf x}^{\gamma_3}-{\bf x}^{\beta_2}\circ {\bf x}^{\gamma_5}+{\bf x}^{\gamma_2}\circ {\bf x}^{\beta_5},\\
{\bf x}^{\beta_{10}}&=M_{\gamma_{10}}({\bf x}^{\gamma_{10}})+{\bf x}^{\gamma_4}\circ {\bf x}^{\gamma_5}-{\bf x}^{\beta_3}\circ {\bf x}^{\beta_5},\\
{\bf x}^{\theta}&=M_{\beta_{10}}({\bf x}^{\beta_{10}})+\sum_{i=1}^5({\bf x}^{\gamma_i}\circ {\bf x}^{\beta_{11-i}}-{\bf x}^{\gamma_{11-i}}\circ {\bf x}^{\beta_{i}}),\\
2(k+11){\bf w}&=Y_\theta({\bf x}^\theta)+\sum_{i=1}^{11}\left({\bf x}^{\beta_{11-i}}\circ S_{\gamma_i}({\bf x}^{\gamma_i})-{\bf x}^{\gamma_{11-i}}\circ S_{\beta_i}({\bf x}^{\beta_i})\right)\\
&-\frac{2}{3}k\Big({\bf x}^{\beta_1}\circ {\bf x}^{\beta_{10}}+{\bf x}^{\gamma_{10}}\circ ({\bf x}^{\gamma_{3}}-{\bf x}^{\gamma_{2}})+{\bf x}^{\beta_{9}}\circ ({\bf x}^{\beta_{2}}-{\bf x}^{\beta_{3}})\\
&\hspace{25pt}+{\bf x}^{\beta_{8}}\circ ({\bf x}^{\beta_{3}}-{\bf x}^{\beta_{4}})+{\bf x}^{\gamma_{8}}\circ ({\bf x}^{\gamma_{5}}-{\bf x}^{\gamma_{4}}-{\bf x}^{\gamma_{6}})+{\bf x}^{\gamma_{9}}\circ {\bf x}^{\gamma_{4}}\\
&\hspace{25pt}-{\bf x}^{\gamma_{7}}\circ {\bf x}^{\gamma_{7}}+{\bf x}^{\beta_{5}}\circ {\bf x}^{\beta_{6}}+{\bf x}^{\beta_{6}}\circ {\bf x}^{\beta_{7}}-{\bf x}^{\beta_{5}}\circ {\bf x}^{\beta_{7}}\Big).
\end{align*}

\vspace{20pt}
\section{The $\mf{sl}_2$ case.}\label{sec:A1}
The case when $\mf{g}$ is of type $A_1$, namely the Lie algebra $\mf{sl}_2(\bb{C})$,
requires a separate approach, essentially due to the fact that only in this case the spectrum of
$\ad\theta^\vee$ in the adjoint representation does not contain $\pm1$, and it is thus given by
\beq\label{specA-1sl2}
\sigma(\theta^\vee)=\left\{-2,0,2\right\} .
\eeq
Since this case was already considered in \cite{BLZ04} and in \cite{FF11} (see also \cite{fioravanti05}), in this
section we merely show that our approach is equivalent. To this aim we work with quantum KdV opers in
the canonical form \eqref{eq:prenormalform}, which is actually simpler in this particular case.
A simple computations shows that operator \eqref{eq:prenormalform} reads
\beq\label{eq:20180919-4}
\mc{L}_{\mf{s}}=\partial_z+
\begin{pmatrix}
0 & v(z)\\
1 &  0
\end{pmatrix},
\eeq
with 
\begin{align*}
v(z)=\frac{r_1(r_1-1)}{z^2}+\frac{1}{z}+\lambda z^k+\sum_{j=1}^N\left(\frac{2}{(z-w_j)^2}+\frac{s(j)}{z(z-w_j)}\right).
\end{align*}
Here $r_1$ is an arbitrary complex number, $-2<k=-\hat{k}-1<-1$, and $s(j),j=1,\dots N$ are free parameters to be determined.
The equation $\mc{L}\psi=0$, in the first fundamental representation, can be written in the form of a second order differential equation
\begin{equation}\label{eq:schrodsl2}
 \psi''(z)=v(z)\psi(z) \; .
\end{equation}
It is well known that the operator \eqref{eq:20180919-4}
has trivial monodromy at $w_j$ if and only if the Frobenius expansion of the dominat solution
$\psi_-=(z-w_j)^{-1}(1+O(z))$ of the latter equation does not contain logarithm terms. This condition imposes a
polynomial relation among the first few terms of the Laurent expansion of $v$ at $z_j$. Indeed, denoting
$$v(z)=\frac{2}{(z-w_j)^2}+\frac{a(j)}{z-w_j}+b(j)+c(j)(z-w_j)+O\big((z-w_j)^2\big),$$
the trivial monodromy condition reads
\begin{equation}\label{eq:trivialscalar}
-\frac{a(j)^3}{4}+a(j)b(j)-c(j)=0 \quad j=1,\dots, N \; .
\end{equation}
We notice that the coefficients $b(j),c(j)$ are affine functions of $\lambda$. Since the equation \eqref{eq:trivialscalar}
is required to
hold for any $\lambda$,
then we can separate $-\frac{a(j)^3}{4}+a(j)b(j)-c(j)$  into a constant part (in $\la$) and a linear part (in $\la$), and both parts
vanish identically.
The linear part reads
\begin{equation*}
 a(j) w_j^{k}-\frac{k}{w_j^{k-1}}=0, 
\end{equation*}
from which we deduce that $a(j)=\frac{w_j}{k}$; equivalently 
\begin{equation}\label{eq:sjsl2}
 s(j)=k .
\end{equation}
 After some algebraic manipulations, the constant part reads
\beq\label{eq:20180919-2}
\widetilde{\Delta}-(k+1)w_\ell=\sum_{\substack{j=1,\dots,N\\ j\neq \ell}}
\frac{w_\ell((k+2)^2w_\ell^2-k(2k+5)w_\ell w_j+k(k+1)w_j^2)}{(w_\ell-w_j)^3},
\eeq
with $\widetilde{\Delta}=\frac{k^3}{4}+k(k+1)-(k+2)r_1(r_1-1)$. The latter system provides the position of the poles,
and thus, together with \eqref{eq:sjsl2}, fully determines the $\mf{sl}_2$ quantum KdV opers \eqref{eq:20180919-4}.
\begin{remark}
The operator \eqref{eq:20180919-4} (with $r_1=-\ell$), subject to the relations
 \eqref{eq:sjsl2} and \eqref{eq:20180919-2}, was shown in \cite[$\S 5.5$, $\S 5.7$]{FF11} to coincide -- after
 the change of coordinates $z=\big(\frac{k+2}{2}\big)^2 x^{\frac{2}{2+k}}$ --
with the operator with `monster potential' originally proposed in \cite[equations(1,3)]{BLZ04}.
\end{remark}
\begin{remark}
According to the general theory developed in Section \ref{sec:miura}, we can write operator in the 
the form \eqref{norformguess}. This reads
\beq\label{eq:20180919-15}
\mc{L}=\partial_z+
\begin{pmatrix}
 r_1 /z & 1/z+\lambda z^{k}\\
1 &  -r_1/z
\end{pmatrix}
+\sum_{j=1}^N
\frac{1}{z-w_j}
\begin{pmatrix}
-1 & x^\theta(j)/z\\
0 &  1
\end{pmatrix},
\eeq
where $x^\theta(j)=k+2r_1-2\sum_{\ell \neq j}\frac{w_j}{w_j-w_\ell}$.
\end{remark}

\vspace{20pt}
\section*{Appendix. Frobenius solutions}
Here we prove Proposition \ref{thm:frobeniusz}. We give full details in the case when the set of additional singularities
$\lbrace w_j\rbrace_{j\in J}$
is empty (i.e. the ground state), omitting a few details 
of the general case to the reader. We recall the proposition for convenience of the reader. 
\\

\paragraph{{\bf Proposition}} \ref{thm:frobeniusz}. Let $\hat{k}$ be irrational and $(r,\hat{k})$ be a generic pair.
Let $\big(\omega(r-\rho^\vee),\chi_{\omega}\big) $ be
an eigenpair composed of an eigenvalue and a corresponding eigenvector of $f-\rho^\vee +\bar{r}$ in $L(\omega_i)$.
A unique solution in $V^i(\la)$ is determined by the expansion
\begin{equation}\label{eq:frobeniuszapp}
\chi_{\omega}(z,\la)=z^{- \rho^\vee}z^{-\omega(r-\rho^\vee)} F(z,\lambda z^{-\hat{k}} ),
\end{equation}
where $F(z,\zeta)$ is an $L(\omega_i)-$valued function which satisfies $\lim_{z \to 0}
F(z,\lambda z^{-\hat{k}})=\chi_{\omega}.$ The function $F(z,\zeta)$ is analytic in $z$ at $z=0$, and an entire
function of $\zeta$, and it admits an analytic extension to an analytic function on $\big(\bb{C}\setminus\lbrace w_j\rbrace_{j\in J}\big) \times \bb{C}$. In the case of the ground state, $F$ is an entire functions of the two variables.

%
\begin{proof}
The matrix $r+f-\rho^\vee$ has the same spectrum as $f-\rho^\vee$, and due to the genericty assumption 
the eigenvalue $\omega(r-\rho^\vee)$ is simple. In order to study the equation $\mc{L}\psi=0$, with $\mc{L}$ given by \eqref{eq:prenormalform}, we first apply the gauge transform $z^{\ad \rho^\vee}$, to get
$$
\left(\partial_z+\frac{r-\rho^\vee+f}{z}+\big(1-\lambda z^{-1+\tilde{k}} \big)e_\theta+
\sum_{j\in J}\sum_{i=1}^n\sum_{l=0}^{d_i}\frac{s^{d_i}_l(j)}{(z-w_j)^{d_i+1-l}}\right) \bar{\psi}(z)
 =0 \; ,
$$
where $\bar{\psi}(z)=z^{\rho^\vee}\psi(z)$.  To simplify our notation, we write
$$\sum_{j\in J}\sum_{i=1}^n\sum_{l=0}^{d_i}\frac{s^{d_i}_l(j)}{(z-w_j)^{d_i+1-l}}=\sum_{k\geq0}s_k z^k.$$ 
Note that the above power series has radius of convergence $\min_{j \in J}|w_j|$.
We then look for a solution of the latter equation in the form of the Frobenius-like series
\beq\label{eq:frobeniusz1}
\bar{\psi}(z)= z^{-\omega(r-\rho^\vee)} \big( \psi_{\omega} + \sum_{m,n \in \bb{N}^2\setminus (0,0)} c_{m,n} z^m (\la z^{-\hat{k}})^n\big).
\eeq
Applying the operator $\mc{L}$ to  \eqref{eq:frobeniusz1} we obtain the recursion
 \begin{equation}\label{eq:recursionfrz}
  \big( r-\rho^\vee+f-\omega(r-\rho^\vee)+m-n \hat{k} \big) c_{m,n}+
  f_0 c_{m-1,n}+\sum_{k=0}^{m-1}s_k c_{m-1-k,n}- f_0 c_{m-1,n-1}=0 \; ,
 \end{equation}
where $c_{0,0}:=\psi_{\omega}$ and $c_{-1,n}=c_{m,-1}=0$ for all $m,n$. It is readily seen that
the recursion has a unique solution and that $c_{m,n}=0$ if $n>m$. The latter identity implies that if the series converges then the function
\beq\label{Fappendix}
F(z,\zeta)= \psi_{\omega} + \sum_{m,n \in \bb{N}^2\setminus (0,0)} c_{m,n} z^m \zeta^n
\eeq
satisfies $\lim_{z\to0}F(z,z^{-\hat{k}}\la)=\psi_{\omega}$.

We now address the convergence of the series in the case when $s_k=0$ for every $k$, namely the ground state case,  and
omit the same discussion for the general case. We choose a norm $\| \cdot\|$ on $L(\omega_i)$ such that $\|\psi_{\omega}\|=1$ and denote by $\|\cdot \|$
also the corresponding operator norm.
Because of the genericity condition on $(r,\hat{k})$,
the matrix $r-\rho^\vee+f-\omega(r-\rho^\vee)+m+n (\tilde{k}-1)$ is invertible. Since moreover $m-n \hat{k}>\hat{k}$,
for all $m,n \neq (0,0), m\geq n$, there exist
$\rho,K<\infty$ such that
 $$\|\big(\ell-\omega(\ell)+m - n \hat{k} \big)^{-1}\|\leq \frac{\rho}{m-n \hat{k}}\, ,
 \; \|f_0\|\leq \frac{K}{\rho} \; . $$
Because of the above inequalities, we get
\begin{equation*}
 \| c_{m,n} \|\leq  K \frac{\|c_{m-1,n} \| + \|c_{m-1,n-1} \| }{m-n \hat{k}},
\end{equation*}
and due to Lemma \ref{lem:recursion} below, we have that
$$
\sum_{m,n \neq (0,0)}\|c_{m,n}\||z|^m |w|^n\leq  1 - e^{K|z|(1+\frac{|w|}{1-\hat{k}})}  \;.
$$
The latter estimate implies that $F(z,\zeta)$ given by \eqref{Fappendix} is an entire function of $z,\zeta$.
Substituting $\zeta=\la z^{-\hat{k}}$, we have
$$
\sum_{m,n \neq (0,0)}\|c_{m,n}\||z|^m |\la z^{-\hat{k}}|^n\leq 1 - e^{K(|z|+\frac{|\la||z|^{1-\hat{k}}}{\hat{k}})} \; ,
$$
from which it follows immediately that $\lim_{z\to 0}F(z,\la z^{-\hat{k}})=\psi_{\omega}$.
\end{proof}

\begin{lemma}\label{lem:recursion}
Let $d$ be a real function of two integer variables $m,n \in \bb{N}$ that satisfies the recursion
 \begin{equation*}
  d_{m,n}=K\frac{d_{m-1,n} + d_{m-1,n-1} }{m+n q} \, , \;  q>-1 ,  
 \end{equation*}
where, in the above formula $d_{-1,n}=0,d_{m,-1}=0$ for all $m,n$. Then
\begin{equation}\label{eq:recursion}
 d_{m,n}=d_{0,0}\frac{K^{m}}{(1+q)^{n}(m-n)!n!} \, , \; \sum_{m,n \in \bb{N}}d_{m,n}x^my^n=d_{0,0}e^{K x(1+\frac{y}{1+q})} \; .
\end{equation}
In particular $d_{m,n}=0$ if $n>m$.
\begin{proof}
We make the following change of basis in the $\bb{Z}^2$ lattice: $m'=m-n,n'=n$, to obtain the new recursion
$d'_{m',n'}=K\frac{d'_{m'-1,n'} + d'_{m',n'-1} }{m'+n' (1+q)}$.
It is easily seen that the recursion has the unique solution
$d'_{m',n'}=d'_{0,0}\frac{K^{m'+n'}}{(1+q)^{n'}m'!n'!}$ and $\sum_{m',n'}d'_{m',n'} x^{m'}y^{n'}=d'_{0,0}e^{x+\frac{y}{1+q}}$.
The thesis follows.
\end{proof}
\end{lemma}

\def\cprime{$'$} \def\cprime{$'$} \def\cprime{$'$} \def\cprime{$'$}
  \def\cprime{$'$} \def\cprime{$'$} \def\cprime{$'$} \def\cprime{$'$}
  \def\cprime{$'$} \def\cprime{$'$} \def\cydot{\leavevmode\raise.4ex\hbox{.}}
  \def\cprime{$'$} \def\cprime{$'$} \def\cprime{$'$}

\end{document}